%% file: supplementary.tex
\documentclass[10pt,journal,compsoc]{IEEEtran}

\newif\ifreport\reporttrue
\usepackage{graphicx}
\usepackage{epsfig,latexsym,amsmath,amsfonts}
\usepackage{mathtools}
\DeclareMathOperator*{\argmin}{argmin}
\usepackage{amssymb}
\usepackage{placeins}
\usepackage{subfigure}
\usepackage{verbatim}
\usepackage{cite,url}
\usepackage{epsf,bm}
\usepackage{xcolor}
\usepackage{epstopdf}
\usepackage{romannum,hyperref}
\usepackage{dsfont}
\usepackage{stfloats}% <-- added
\usepackage[utf8]{inputenc}
\usepackage[english]{babel}
\usepackage{amsthm}
\newtheorem{definition}{Definition}
\newtheorem{remark}{Remark}

\usepackage[utf8]{inputenc}
\usepackage[english]{babel}
\newtheorem{theorem}{Theorem}
\newtheorem{lemma}{Lemma}
\newtheorem{proposition}{Proposition}
\newtheorem{corollary}{Corollary}

\usepackage[lined, boxed, linesnumbered, ruled]{algorithm2e}
\usepackage{multirow}
\IEEEoverridecommandlockouts

\begin{document}
%\title{Minimizing Age of Information via Scheduling over Heterogeneous Channels}
%\title{Efficient Age-optimal Scheduling over Heterogeneous Channels}
\title{Age-optimal Scheduling over Hybrid Channels}

\author{Jiayu Pan, Ahmed M. Bedewy, 
        Yin Sun, \IEEEmembership{Senior Member, IEEE,}
        and Ness B. Shroff, \IEEEmembership{Fellow, IEEE}
\thanks{This paper was presented in part at ACM MobiHoc 2021 \cite{pan2020age}.

This work was funded in part through NSF grants: CNS-1901057, CNS- 2007231, CNS-1618520, CNS-1409336, CNS-1955561, CCF-1813050, IIS- 2112471, CNS- 2106932, CNS- 1955535, the Office of Naval Research under Grant N00014-17-1-241, and ARO grant W911NF-21-1-0244.

J. Pan is with the Department of ECE, The Ohio State University,
Columbus, OH 43210 USA (e-mail: pan.743@osu.edu).

A. M. Bedewy is with the Department of ECE, The Ohio State University,
Columbus, OH 43210 USA (e-mail: bedewy.2@osu.edu).

Y. Sun is with the Department of ECE, Auburn University, Auburn, AL
36849 USA (e-mail: yzs0078@auburn.edu).

N. B. Shroff is with the Department of ECE and the Department of
CSE, The Ohio State University, Columbus, OH 43210 USA (e-mail:
shroff.11@osu.edu).}}

\IEEEtitleabstractindextext{%
\begin{abstract}
We consider the problem of minimizing the age of information when a source can transmit status updates over two heterogeneous channels. Our work is motivated by recent developments in 5G mmWave technology, where transmissions may occur over an unreliable but fast (e.g., mmWave) channel or a slow reliable (e.g., sub-6GHz) channel. The unreliable channel is modeled as a time-correlated Gilbert-Elliot channel at a high rate when the channel is in the “ON” state. The reliable channel provides a deterministic but lower data rate. The scheduling strategy determines the channel to be used for transmission in each time slot, \textcolor{black}{aiming} to minimize the time-average age of information (AoI). The optimal scheduling problem is formulated as a Markov Decision Process (MDP), which \textcolor{black}{is challenging to solve} because super-modularity does not hold in a part of the state space.
We address this challenge and show that a multi-dimensional threshold-type scheduling policy is optimal for minimizing the age. By exploiting the structure of the MDP \textcolor{black}{and analyzing the discrete time Markov chains (DTMCs) of the threshold-type policy}, we devise a low-complexity bisection algorithm to compute the optimal thresholds. We compare different scheduling policies using numerical simulations.  

\end{abstract}

% Note that keywords are not normally used for peerreview papers.
\begin{IEEEkeywords}
Age of information, hybrid channels, \textcolor{black}{scheduling, and mmWave communications}.
\end{IEEEkeywords}}

% make the title area
\maketitle

% To allow for easy dual compilation without having to reenter the
% abstract/keywords data, the \IEEEtitleabstractindextext text will
% not be used in maketitle, but will appear (i.e., to be "transported")
% here as \IEEEdisplaynontitleabstractindextext when the compsoc 
% or transmag modes are not selected <OR> if conference mode is selected 
% - because all conference papers position the abstract like regular
% papers do.
\IEEEdisplaynontitleabstractindextext
% \IEEEdisplaynontitleabstractindextext has no effect when using
% compsoc or transmag under a non-conference mode.

% For peer review papers, you can put extra information on the cover
% page as needed:
% \ifCLASSOPTIONpeerreview
% \begin{center} \bfseries EDICS Category: 3-BBND \end{center}
% \fi
%
% For peerreview papers, this IEEEtran command inserts a page break and
% creates the second title. It will be ignored for other modes.
\IEEEpeerreviewmaketitle

\input{sections/intro}

\input{sections/related}
\input{sections/system}

\input{sections/optimal}

\input{sections/numerical}

\input{sections/conclusion}
\bibliographystyle{IEEEtran}
\bibliography{tech-report}

\input{sections/proof}

\input{sections/appendix}

% that's all folks
\end{document}

%% file: sections/intro.tex
\IEEEraisesectionheading{\section{Introduction}\label{introduction}}
%\section{Introduction}\label{introduction}
%Substantial work has been conducted on low latency data transmission as in 5G cellular standard \cite{parvezsurvey}. However,
\IEEEPARstart{T}{imely} \textcolor{black}{updates} of the system state are of great significance in cyber-physical systems, such as vehicular networks, sensor networks, and UAV navigations. In these \textcolor{black}{systems, freshly} generated data is more \textcolor{black}{valuable} than outdated data. \emph{Age of information} (AoI), or simply \emph{age}, was introduced as an end-to-end application-layer metric to \textcolor{black}{measure information} freshness \cite{yates2020age, kaul2012real,costa2016age,inoue2019general,buyukates2020age,bedewy2019minimizing,bedewy2019age,sun2018age,sun2019sampling,sun2019samplings,bedewy2019ages,hsu2019scheduling,talak2019optimizing,bedewy2021low,joo2018wireless,lu2018age,kadota2018optimizing,qian2020minimizing,liu2019minimizing,ornee2019samplings,sun2017update,altman2019forever,el2012optimal,talak2018optimizing}. The age at time $t$ is defined as $\Delta(t)=t-U_t$, where $U_t$ is the generation time of the freshest packet that has been received by time $t$. The difference between age and classical \textcolor{black}{performance metrics of wireless networks} like delay and throughput is evident even in elementary queuing systems \cite{kaul2012real}. High throughput requires frequent status updates, which would cause a long waiting time in the queue that worsens timeliness. On the other hand, delay \textcolor{black}{and waiting time} can be greatly reduced by decreasing the update frequency, which, however, may increase the age because the status is updated infrequently. 

%in a low latency case, an infrequent status update system leads to receiving more out-dated data at the destination. So low latency does not imply small age in general.            

%In recent years though, research efforts are conducted on a more promising example: hybrid high-frequency band millimeter-wave (mmWave) and low-frequency band sub-6GHz channels in future generation 5G.

 %For example, in vehicle networks, the image-to-depth estimation computation needs camera video data in each vehicle \cite{casser2019depth}.

In future wireless networks, sub-6GHz frequency spectrum is insufficient for fulfilling the high throughput demand of emerging real-time applications such as VR/AR applications, where contents must be delivered within 5-20 ms of latency, requiring a high throughput of 400-600 Mbps \cite{qualcomm2018}.      
To address this challenge, 5G technology utilizes high-frequency millimeter wave (mmWave) bands such as 28/38 GHz, which provide a much higher data rate than sub-6GHz \cite{rappaport2013millimeter}. Verizon and Samsung demonstrated that a throughput of nearly 4Gbps was achieved in their mmWave demo system, using a 28GHz frequency band with 800MHz bandwidth \cite{verizon2018}. However, unlike sub-6GHz spectrum bands, mmWave channels are highly unreliable due to blocking susceptibility, strong atmospheric absorption, and low penetration. Real-world smartphone experiments have shown that even obstructions by hands could significantly degrade the mmWave throughput \cite{narayanan2020first}. One solution to mitigate this effect is to let sub-6GHz coexist with mmWave to form two \emph{heterogeneous} channels, so that the user equipment can offload data to sub-6GHz when mmWave communications are unfeasible \cite{pi2011introduction,pi2011system,semiari2019integrated,aziz2016architecture}. Some work has already been done based on mmWave/sub-6GHz heterogeneous networks \cite{deng2017resource,elshaer2016downlink}. 
However, how to improve information freshness in such hybrid networks has remained largely unexplored.       

In this study, we consider a hybrid status updating system where a source can transmit the update  packets over an unreliable but fast mmWave channel %(referred to as \emph{Channel 1}) 
or a slow reliable sub-6GHz channel. %(referred to as \emph{Channel 2}). 
%We assume that Channel $1$ has a short transmission delay following Gilbert-Elliot model \cite{buyukates2020age}, and Channel $2$ conversely has a steady connection with a larger deterministic delay. %We assume that the source receives a delayed channel state information because in high-frequency update systems, the channel states vary from time to time, and there is a time-sensitive unpredictable amount of interference at the receiver \cite{tse2005fundamentals,trillingsgaard2017generalized}. 
Our objective is to find a dynamic channel scheduling policy that minimizes the long-term average expected age. 
%Considering the two aspects of difference between Channel $1$ and Channel $2$, and the time correlation issue in Channel $1$, finding an efficient optimal policy is challenging, which we have solved in our study.  
The main contributions of this paper are stated as follows:    
%we first show that there exists a stationary and deterministic optimal policy (Lemma $1$). Then 
%The system parameters are the first channel's state transition probabilities and the delay ratio between slow and fast channels.
\begin{itemize}
    \item  The optimal scheduling problem for minimizing the age over heterogeneous channels is formulated as a Markov Decision Process (MDP). The state transition of this MDP is complicated for two reasons: (i) the two channels have different data rates and packet transmission times, and (ii) the state of the unreliable mmWave channel is correlated over time. We prove that there exists a multi-dimensional \textcolor{black}{threshold-type} scheduling policy that is optimal. This optimality result holds for all possible values of the channel parameters. 
    %For this MDP, %The system parameters are: transition probabilities of Channel $1$, and delay fraction of Channel $2$ over Channel $1$. The optimal policy is either non-decreasing or non-increasing depending on the choice of the system parameters and the previous state of Channel $1$. 
\textcolor{black}{One of the tools for proving this result is super-modularity} \cite{topkis1998supermodularity}. Because of the complicated state transitions, \textcolor{black}{super-modularity} holds in a part of the state space \textcolor{black}{but not in the rest of the state space. This is a key }difference from the scheduling problems considered earlier in \textcolor{black}{prior studies,} e.g.,  \cite{krishnamurthy2016partially,puterman1990markov,ngo2009optimality,altman2019forever,sun2017update,yao2019integrating,sun2019sampling}. \textcolor{black}{To conquer this challenge, }we develop additional techniques to show that the \textcolor{black}{optimal scheduling policy has a threshold-type structure over the entire state space, including the part of state space where super-modularity} does not hold.

    \item \textcolor{black}{The state transition of the discrete time Markov chain (DTMC) for the threshold-type scheduling policy is complicated. Nonetheless}, we show that the thresholds of the optimal scheduling policy can be evaluated efficiently, by using closed-form expressions or a low-complexity bisection search algorithm. Compared with the algorithms for calculating the thresholds and optimal scheduling policies in, e.g.,  \cite{krishnamurthy2016partially,puterman1990markov,ngo2009optimality,altman2019forever,sun2017update,yao2019integrating,sun2019sampling}, our solution algorithms have much lower computational complexities.

    \item In the special case that the state of the unreliable mmWave channel 
    %We then provide a simple solution of a special case when Channel $1$ 
    is independent and identically distributed (i.i.d.) over time, the optimal scheduling policy is shown to possess \textcolor{black}{a simpler and interesting }form. %It is found in this special case that the optimal scheduling policy for minimizing the age is different from that minimizes delay. 
\textcolor{black}{Finally, numerical results are provided to validate our results by comparing with several other policies}. 
\end{itemize}

%% file: sections/related.tex
\section{Related Works}\label{relatedworks}
Age of information has become a popular research topic in recent years, e.g., \cite{yates2020age, kaul2012real,costa2016age,inoue2019general,buyukates2020age,bedewy2019minimizing,bedewy2019age,sun2018age,sun2019sampling,sun2019samplings,bedewy2019ages,hsu2019scheduling,talak2019optimizing,bedewy2021low,joo2018wireless,lu2018age,kadota2018optimizing,qian2020minimizing,liu2019minimizing,ornee2019samplings,sun2017update,altman2019forever,el2012optimal,talak2018optimizing}. A comprehensive survey of the area was recently provided in \cite{yates2020age}. First, there has been substantial work on age performance analysis in queuing \textcolor{black}{systems} \cite{kaul2012real,costa2016age,inoue2019general,buyukates2020age,bedewy2019age,bedewy2019minimizing}. Average age and peak age in elementary queuing systems were analyzed in \cite{kaul2012real,costa2016age,inoue2019general}. A similar setting was considered in \cite{buyukates2020age} \textcolor{black}{where the inter-arrival times or service times follow a Gilbert-Elliot two-state Markov chain model}. %In \cite{yates2018age}, the authors analyze AoI in multiple sources queuing networks. 
A Last-Generated, First-Served (LGFS) policy was shown (near) optimal in \textcolor{black}{single-source, multi-server}, and multihop networks with arbitrary \textcolor{black}{packet} generation and arrival process \cite{bedewy2019age,bedewy2019minimizing}. These results were extended to multi-source multi-server \textcolor{black}{networks} in \cite{sun2018age}. 

Next, there has been a significant effort in age-optimal sampling \cite{sun2017update,sun2019sampling,sun2019samplings,ornee2019samplings,bedewy2019ages}. The optimal sampling policy was provided for minimizing a monotonic age function in \cite{sun2017update,sun2019sampling,ornee2019samplings}. Joint Sampling and scheduling in multi-source systems were analyzed in \cite{bedewy2019ages} where \textcolor{black}{the objective problem} could be decoupled into maximum age first (MAF) scheduling \cite{sun2018age} and an optimal sampling problem. Finally, age in wireless networks has been substantially explored in \cite{hsu2019scheduling,talak2019optimizing,joo2018wireless,lu2018age,kadota2018optimizing,qian2020minimizing,liu2019minimizing}. Scheduling in a broadcast network with random arrivals was provided where Whittle index policy can achieve (near) age optimality \cite{hsu2019scheduling}. Some other age-optimal scheduling for cellular networks were considered in \cite{talak2019optimizing,joo2018wireless,lu2018age,kadota2018optimizing,talak2018optimizing}. %A similar problem with perfect channel state is studied in \cite{talak2018optimizing}. 
%Age minimization scheduling with minimum throughput constraint was described \cite{lu2018age,kadota2018optimizing}. 
A class of age-optimal scheduling policies was analyzed in the asymptotic regime when the number of sources and channels both grow to infinity \cite{qian2020minimizing}. An \textcolor{black}{age-optimal} multi-path routing strategy was introduced in \cite{liu2019minimizing}.

%A hybrid network refers to a network that supports the user for more than one type of connecting technology. One traditional real-life example is utilizing both cellular and wifi channels in 3G/4G. AoI minimizing in such a setting is considered in \cite{altman2019forever,el2012optimal}, where the papers decide whether to choose wifi, cellular ,or stay idle to maximize the expected average reward that consists of AoI utility function, energy/monetary cost from wifi and cellular channels. \cite{altman2019forever} shows that there exists an increasing threshold type of decision policy about the age that is optimal. 

However, the age-optimal \textcolor{black}{scheduling} problem via heterogeneous channels has been largely unexplored yet.
\textcolor{black}{Technical results for similar models} were reported in \cite{altman2019forever,el2012optimal}. \textcolor{black}{In these studies}, it is assumed that the first channel is unreliable but consumes a lower cost, and the second channel has the same delay \textcolor{black}{as the first channel, but depletes} a higher cost. \textcolor{black}{Optimal} scheduling policies were derived to achieve the optimal trade-off between age performance and cost.\begin{figure}[htbp]
\centerline{\includegraphics[width=0.5\textwidth]{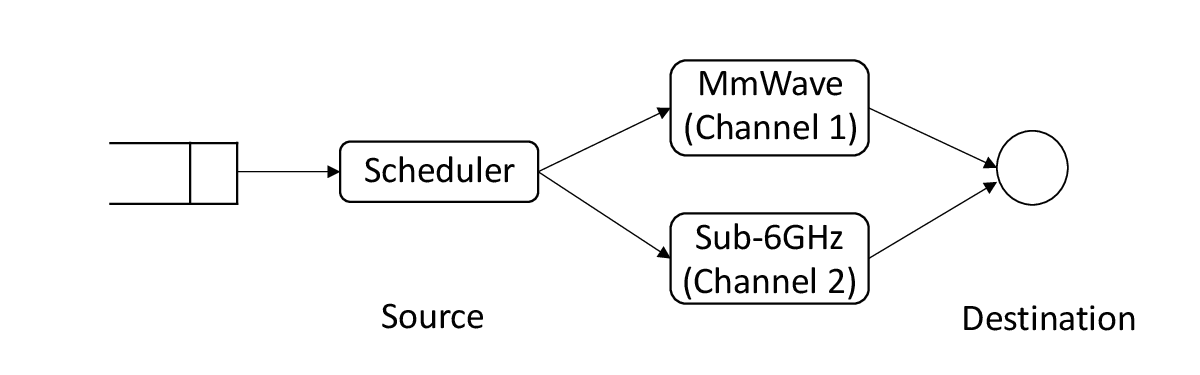}}
\caption{The system model for status updates in heterogeneous channels. The scheduler chooses mmWave (Channel 1) or sub-6GHz (Channel 2) for transmission over time.}
\label{fig1}
\end{figure}
\begin{figure}[htbp]
\centerline{\includegraphics[width=0.35\textwidth]{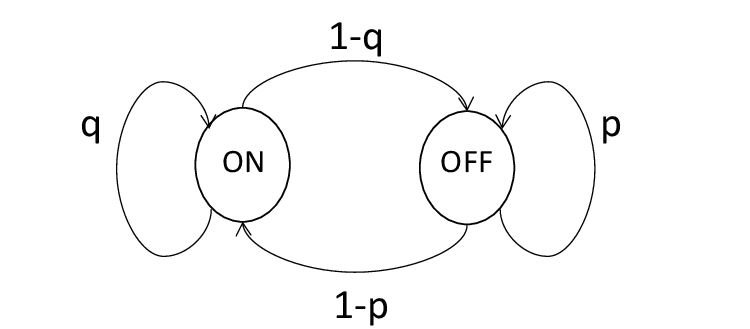}}
\caption{The Gilbert-Elliot \textcolor{black}{$ON$-$OFF$ Markov} model for Channel 1.}
\label{fig2}
\end{figure}Our study is different from \cite{altman2019forever,el2012optimal} in two aspects: (i) The study in \cite{altman2019forever,el2012optimal} show the optimality of a threshold-type policy
and efficiently \textcolor{black}{computes} the optimal threshold when the first channel is i.i.d. \cite{altman2019forever}, but our work allows a Markovian channel which \textcolor{black}{generalizes the i.i.d. channel model in \cite{altman2019forever}.}
 (ii) In addition, our study assumes that \textcolor{black}{the second sub-6GHz channel has a larger delay than the first mmWave} which complies with the property of dual mmWave/sub-6GHz channels in real applications. These two differences between mmWave and sub-6GHz make the MDP formulation more complex \textcolor{black}{than those of \cite{altman2019forever,el2012optimal}.} Thus, the techniques \textcolor{black}{in e.g., \cite{krishnamurthy2016partially,puterman1990markov,ngo2009optimality,altman2019forever,sun2017update,yao2019integrating,sun2019sampling} that can show} a nice structure of the optimal policy or solve the optimal policy with low complexity do not apply to our model. 
%Their objective is to maximize the expected average reward that consists of age utility function and cost.

%, and focus on the tradeoff between age  performance and cost Also, we focus on how to smartly selecting the two different channels to improve age performance.

%% file: sections/system.tex
\section{System Model and Problem Formulation}\label{systemod}

%In this section, we first describe our network model with scheduling controls and the age of information evolution. Then we provide the problem formulation which is equivalent to an MDP problem. %(e.g., video from automobile or unmanned aerial vehicle (UAV))

\subsection{System Models}

Consider a single-hop network as illustrated in Fig. \ref{fig1}, where a source sends status update packets to the destination. We assume that time is slotted with slot index $t \in \{ 0,1,2...\} $. The source can generate a fresh status update packet at the beginning of each time slot. The packets can be transmitted either over the mmWave channel or over the sub-6GHz channel. The packet transmission time of the mmWave channel is $1$ time slot, whereas the packet transmission time of the sub-6GHz channel is $d$ time slots ($d\ge 2$)\footnote{If $d=1$, one can readily see that it is better to choose sub-6GHz than mmWave. Thus, in this paper we study the nontrivial case of $d \ge 2$.} because of its lower data rate.
The two channels have \textcolor{black}{different} advantages, which is the key feature of our study.

The mmWave channel, called \emph{Channel 1}, follows a two-state Gilbert-Elliot  model that is shown in Fig. \ref{fig2}. We say that Channel $1$ is $ON$ in time slot $t$, denoted by $l_1(t)=1$, if the packet is successfully transmitted to the destination in  time slot $t$; otherwise Channel $1$ is said to be $OFF$, denoted by $l_1(t)=0$. If a packet is not successfully transmitted, then it is dropped, and a new status update packet is generated at the beginning of the next time slot. The self transition probability of the $ON$ state is $q$, and the self transition probability of the $OFF$ state is $p$, where $0<q<1$ and $0<p<1$. 
We assume that at the beginning of time slot $t$, the source knows $l_1(t-1)$ perfectly.
%the source has access to the state of Channel $1$, but with one time slot of feedback delay. That is, 
% {\color{red} Do you need this? 
%The steady state distribution of Channel $1$ is $P(ON)=(1-p)/(1-p+1-q)$ and $P(OFF)=(1-q)/(1-p+1-q)$.}   

The sub-6GHz channel, called \emph{Channel 2}, has a steady connection. As mentioned above, the packet transmission time of Channel 2 is $d$ time slots. Define $l_2(t)\in \{ 0,1,...,d-1 \} $ as the state of Channel 2 in time slot $t$, where $l_2(t)$ is the remaining transmission time of the packet being sent over Channel $2$ at the beginning of time slot $t$, and $l_2(t)=0$ means that Channel $2$ is currently idle and ready for sending the next packet. In time slot $t$, the source  has immediate knowledge about the state $l_2(t)$ of Channel $2$. On the other hand, because the packet transmission time of Channel 1 is 1 time slot, Channel 1 is always ready for transmission at the beginning of each time slot.

Following the application settings in \cite{pi2011introduction,pi2011system,semiari2019integrated,deng2017resource,aziz2016architecture}, a packet can be transmitted using only one channel at a time, i.e., the two channels cannot be used simultaneously. The scheduler decides which channel to use \textcolor{black}{for transmitting} a packet at each time slot. We also assume that the scheduler can choose idle (neither channel) since it has been shown that channel idling could reduce the average age in \textcolor{black}{some} systems \cite{sun2017update,bedewy2019ages,sun2019sampling}. \textcolor{black}{Hence, the} scheduling decision at the beginning of time slot $t$ \textcolor{black}{can be} denoted by $u(t)\in \{ 1,2,{none} \}$. The action $u(t)=1$ or $2$ means that the source generates a packet and assigns it to Channel $1$ or Channel $2$, respectively. 
%The decision $u(t)=none$ means that there is no packet assignment in time slot $t$ as one of the channels is busy. 
The action $u(t)=none$ means that no new packet is assigned to any channel at time slot $t$. 
Hence, u(t) = $none$ can occur if (i) a packet is was assigned to Channel 2 earlier and has not completed its transmission, i.e., $l_2(t) \in \{1, 2,\ldots, d-1 \}$ such that no packet can be assigned for transmission, or (ii) $l_2(t)=0$, but both channels are kept idle on purpose. 
%Channel idling was shown to be potentially beneficial for minimizing the age in [3, 34, 37]. However, channel idling has no benefit in our model, because it is always better to assign a packet to Channel 1 than keeping both channels idle. As a result, u(t) = none can only occur when $l_2(t) \in \{1, 2,\ldots, d-1 \}$. 
%by: $l_2(t)=d$ if $u(t-1)=2, l_2(t-1)=0$, $l_2(t-1)-1$ if $l_2(t-1)>0$, $0$ otherwise

%\vspace{-12pt}

%\subsection{Age of Information Model}

The age of information (AoI) $\Delta(t)$ is the time difference between the current time slot $t$ and the generation time of the freshest delivered packet \cite{kaul2012real}. By this definition, when a packet is delivered, the age drops to the transmission time duration of the delivered packet. Specifically, if Channel 1 is selected in time slot $t$ and Channel 1 is $ON$, then the age drops to $1$ at time slot $t+1$. If the remaining service time of Channel 2 at time slot $t$ is 1, then age drops to $d$ at time slot $t+1$. When there is no packet delivery at time slot $t$, the age increases by one in each time slot.
Hence, the time-evolution of the age  is given by
\begin{equation}
    \Delta(t+1) = \left\{
\begin{array}{lll}
  1 & \text{if } u(t)=1 \text{ and }   l_1(t)=1,\\
  d & \text{if } l_2(t)=1,\\
  \Delta(t)+1 & \text{Otherwise.} 
\end{array}
\right. %\label{age-def}
\end{equation}
%\subsection{Problem Setup and MDP Formulation}
\subsection{Problem Formulations}\label{systemod2}

%\linespread{1.3}
    \begin{table} \caption{Value of State Transition Probability} 
\begin{center}  
\begin{tabular}{|l|l|}  
\hline  
$P_{\textbf{ss}'}(u)$ & Action and State Transition  \\ \hline
$p$ &  $u=1, \textbf{s} = (\delta,0,0), \textbf{s}' = (\delta+1,0,0)$\\
  & $ u=2, \textbf{s} = (\delta,0,0), \textbf{s}' = (\delta+1,0,d-1)$\\
    & $ u=none, \textbf{s} = (\delta,0,0), \textbf{s}' = (\delta+1,0,0)$\\
  & $ u=none, \textbf{s} = (\delta,0,l_2), \textbf{s}' = (\delta+1,0,l_2-1), l_2\geq 2  $\\
  & $ u=none, \textbf{s} = (\delta,0,1), \textbf{s}' = (d,0,0)$\\
 \hline
$1-p$ & $u=1, \textbf{s} = (\delta,0,0), \textbf{s}' = (1,1,0)$  \\
      & $ u=2, \textbf{s} = (\delta,0,0), \textbf{s}' = (\delta+1,1,d-1)  $\\
      & $ u=none, \textbf{s} = (\delta,0,0), \textbf{s}' = (\delta+1,1,0)  $\\
      & $ u=none, \textbf{s} = (\delta,0,l_2), \textbf{s}' = (\delta+1,1,l_2-1),l_2\geq 2  $\\
      & $ u=none, \textbf{s} = (\delta,0,1), \textbf{s}' = (d,1,0)$\\ \hline  
$q$ & $u=1, \textbf{s} = (\delta,1,0), \textbf{s}' = (1,1,0)$  \\
      & $ u=2, \textbf{s} = (\delta,1,0), \textbf{s}' = (\delta+1,1,d-1)  $\\
      & $ u=none, \textbf{s} = (\delta,1,0), \textbf{s}' = (\delta+1,1,0)  $\\
      & $ u=none, \textbf{s} = (\delta,1,l_2), \textbf{s}' = (\delta+1,1,l_2-1),l_2\geq 2  $ \\
      & $ u=none, \textbf{s} = (\delta,1,1), \textbf{s}' = (d,1,0)$\\ \hline  
$1-q$ &  $u=1, \textbf{s} = (\delta,1,0), \textbf{s}' = (\delta+1,0,0)$\\
      & $ u=2, \textbf{s} = (\delta,1,0), \textbf{s}' = (\delta+1,0,d-1)$\\
      & $ u=none, \textbf{s} = (\delta,1,0), \textbf{s}' = (\delta+1,0,0)  $\\
   & $ u=none, \textbf{s} = (\delta,1,l_2), \textbf{s}' = (\delta+1,0,l_2-1),l_2\geq 2  $ \\
   & $ u=none, \textbf{s} = (\delta,1,1), \textbf{s}' = (d,0,0)$\\
 \hline
$0$ & Otherwise\\ \hline
\end{tabular}  
\end{center} 
\label{xsit}
\end{table}

%Our objective in this paper is to seek the optimal scheduling strategy for minimizing the expected average age. 

We use ${\pi}=\{u(0),u(1)...\}$ to denote a scheduling policy. A scheduling policy is said to be \emph{admissible} if (i) $u(t)=none$ whenever $l_2(t)\geq 1$ and (ii) $u(t)$ is determined by the current and history information that is available at the scheduler. Let $\Delta_{{\pi}}(t)$ denote the AoI induced by policy $\pi$. The expected time-average age of policy ${\pi}$ is 
\begin{equation*}
\limsup_{T\rightarrow \infty} \frac{1}{T} \sum_{t=1}^{T} \mathbb{E}[\Delta_\pi(t)]. 
\end{equation*}
Our objective in this paper is to solve the following optimal scheduling problem for minimizing the expected time-average age:
\begin{equation}
 \bar{\Delta}_{\text{opt}} = \inf_{\pi\in \Pi} \limsup_{T\rightarrow \infty} \frac{1}{T} \sum_{t=1}^{T} \mathbb{E}[\Delta_\pi(t)],
\label{avg}
\end{equation} 
where ${\Pi}$ is the set of all admissible policies.
%$V_{\pi^*}=\inf_{{\pi}\in {\Pi}}V_{{\pi}}$. Here, 
Problem \eqref{avg} can be equivalently expressed as \textcolor{black}{an infinite time-horizon} average-cost MDP problem \cite{bertsekas1995dynamic,puterman1990markov}, which is illustrated below.

\begin{itemize}
    \renewcommand\labelenumi{\bfseries\theenumi}
    \item \textbf{Markov State:} The system state in time slot $t$ is defined as \begin{equation}
    \textbf{s}(t)=(\Delta(t),l_1(t-1),l_2(t)), 
    \end{equation} 
    where $\Delta(t)\in \{ 1,2,3,... \}$ is the AoI in time slot $t$, $l_1(t-1)\in \{ 0,1 \}$ is the $ON-OFF$ state of Channel $1$ in time slot $t-1$, and $l_2(t)\in \{ 0,1,...,d-1 \}$ is the remaining transmission time of Channel $2$ at the beginning of time slot $t$. Let \textbf{S} denote the state space which is countably infinite. The time-evolution of $\Delta(t)$ is determined by the state and action in time slot $t-1$.
    \item \textbf{Action:} As mentioned before, if Channel $2$ is busy (i.e., $l_2(t)>0$), the scheduler always chooses an idle action, i.e., $u(t)=none$. Otherwise, the action $u(t)\in \{ 1,2, {none}\}$. 
    \item \textbf{Cost function:} Suppose that a decision $u(t)$ is applied at a time slot $t$, we encounter a cost $C(\textbf{s}(t),u(t))=\Delta(t)$. 
    \item \textbf{Transition probability:} We use $P_{\textbf{ss}'}(u)$ to denote the transition probability from state $\textbf{s}$ to $\textbf{s}'$ for action $u$. The value of $P_{\textbf{ss}'}(u)$ is summarized in Table \ref{xsit}. 
    
We provide an explanation of the transition probabilities $P_{\textbf{ss}'}(u)$ in Table \ref{xsit}. Due to the Markovian state transition properties of Channel $1$, there are four possible values of state transition probabilities: $p,1-p,q$ and $1-q$. For example, $P_{\textbf{ss}'}(u)=p$ \textcolor{black}{if} both the current and previous states of Channel $1$ is $OFF$. %Thus, the second dimension of the system state of $\textbf{s}$, $\textbf{s}'$ is $0$, $0$, respectively. Since the current state of Channel $1$ is $OFF$, The age state of $\textbf{s}'$ does not reduce to $1$. The age state of $\textbf{s}'$ may only reduce to $d$ when the service of Channel $2$ is completed. 
\textcolor{black}{Thus, there are two possible age state evolutions: if the remaining time slot of Channel $2$ is $1$, the age $\delta$ decreases to $d$; otherwise, the age $\delta$ increases by one time slot.}        
The transition probabilities of other cases, i.e., $P_{\textbf{ss}'}(u)=1-p,q$ and $1-q$ in Table \ref{xsit} can be \textcolor{black}{explained} in the similar way.

\end{itemize}

%Also, in this paper, we set $\bar{\Delta}_{\text{opt}} $ the solution of Problem \eqref{avg}

%% file: sections/optimal.tex
\section{Main Results}\label{mainresults}

In this section, we show that there exists a threshold-type policy that solves Problem \eqref{avg}. We then provide a low-complexity algorithm to obtain the optimal policy and optimal average age.   

%\subsection{Optimality of stationary deterministic policy} 

\subsection{Optimality of Threshold-type Policies}

As mentioned in Section \ref{systemod2}, the action space of the MDP allows $u(t)=none$ even if Channel $2$ is idle, i.e., $l_2(t)=0$. 
In the following lemma, we show that the action $u(t)=none$ can be abandoned when $l_2(t)=0$. 
Define
\begin{align}
\Pi'=\{\pi\in \Pi: u(t)\neq none, \text{ if } l_2(t)=0\}.
\end{align}
%we can remove this assumption. %Suppose $\Pi'$ is the action space that is a subset of $\Pi$, and $\Pi'$ only allows 
%to chose this action generate a packet once the channels are idle is the optimal decision.

\begin{lemma}\label{zerowait}
For any $\pi \in \Pi$, there exists a policy $\pi'\in \Pi'$ that is no worse than $\pi$.  
\end{lemma}
\begin{remark}
In \cite{sun2017update,bedewy2019ages,sun2019sampling}, \textcolor{black}{it was shown that in certain systems}, the zero wait policy (transmitting immediately after the previous update has been received) might not be optimal. However, in our model, the zero wait policy is indeed optimal. The reason is that in our model, the minimum non-zero waiting time is one time slot which is the same as the delay of Channel $1$. If $l_2(t)=0$, it is better to choose Channel $1$ than keeping both channels idle, because, by choosing Channel $1$, fresh packets could be delivered over Channel $1$. \end{remark}
The proof of Lemma \ref{zerowait} is provided in 
\ifreport
Appendix \ref{zerowaitapp}.
\else
Appendix A of the supplementary material.
\fi
\textcolor{black}{By Lemma \ref{zerowait}, the scheduler only needs to choose from the actions $u(t)\in \{ 1,2\}$} when $l_2(t)=0$. This lemma simplifies the MDP problem.

 %Recall that the self transition probabilities $(p,q)$ of Channel $1$ and the transmission time $d$ of Channel $2$ are the parameters of the hybrid channels. 
 %In this paper channel parameters set $(p,q,d)$ has a general domain $\Theta=(0,1)\times (0,1)\times \mathbb{N}^+$. 
\textcolor{black}{The parameters of the hybrid channels are $(p,q,d)$, where $p,q$ are the self transition probabilities of Channel $1$ and $d$ is the transmission delay of Channel $2$.
For the ease of presenting our main results,} we divide the possible values of channel parameters $(p,q,d)$ into four complementary regions $\textbf{B}_1,\ldots,\textbf{B}_4$. %\ref{def1} for $4$ regions that divide the channel parameters set space $\Theta$:

\begin{figure}[htbp]
\centerline{\includegraphics[width=0.45\textwidth]{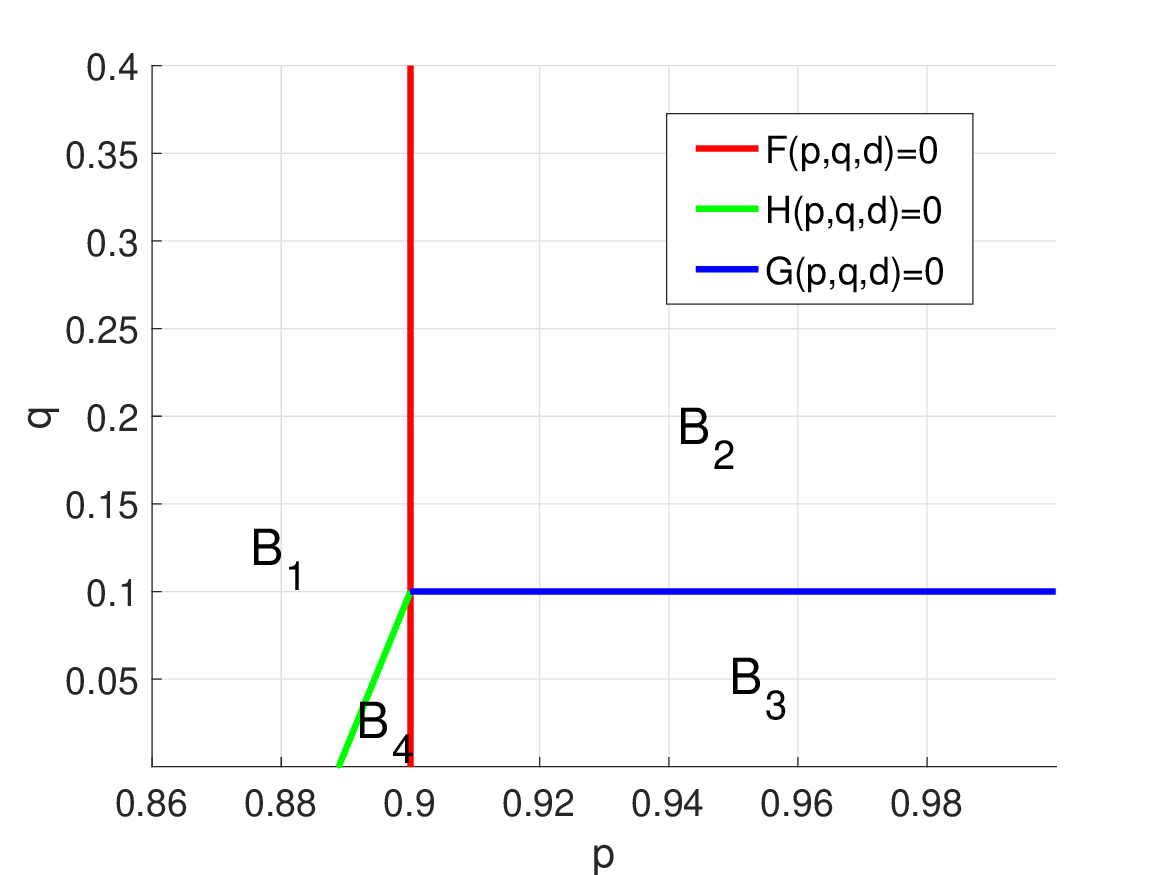}}
\caption{The Diagram of the regions $\textbf{B}_1,\ldots,\textbf{B}_4$ with an example of $d=10$. In the diagram, each function $F,G,H$ divides the whole plane $((p,q)\in (0,1)\times (0,1))$ into two half-planes respectively. Each region $\textbf{B}_1,\ldots,\textbf{B}_4$ is the intersection of some two half-plane areas. Since we emphasize the differences of the four regions, we provide the partial but enlarged diagram.}
\label{fig3}
\end{figure}

\begin{definition}
%Suppose the channel parameters $(p,q,d)\in \Theta$, then
The regions $\textbf{B}_1,\ldots,\textbf{B}_4$ are defined as
\begin{equation}\label{b1-b4}
\begin{split}
    \textbf{B}_1 & =\{ (p,q,d): F(p,q,d)\le 0, H(p,q,d)\le 0 \}, \\     
    \textbf{B}_2 & =\{ (p,q,d): F(p,q,d)> 0, G(p,q,d)\le 0 \}, \\
    \textbf{B}_3 & =\{(p,q,d): F(p,q,d)> 0, G(p,q,d)> 0 \}, \\
    \textbf{B}_4 & =\{(p,q,d): F(p,q,d)\le 0, H(p,q,d)> 0 \}, 
    \end{split}\end{equation}
where %the real value functions $F(p,q,d),G(p,q,d),H(p,q,d)$ are defined as: 
\begin{equation}\begin{split}\label{def1}
   & F(p,q,d)=\frac{1}{1-p}-d, \\  & G(p,q,d)=1-dq, \\
   & H(p,q,d)=\frac{1-q}{1-p}+1-d. 
    \end{split}\end{equation}
\end{definition}

%In this study, the regions $\textbf{B}_1 - \textbf{B}_4$ will serve as the sufficient conditions for optimality of threshold type policy. 
Note that the inequality ${1}/(1-p) > d$ also represents a comparison between the channel delay $d$ and the average \textcolor{black}{waiting time} for an $ON$ channel state given that the last channel state is $OFF$. %Similar interpretations exist for all the boundary functions $F,G,H$ of the regions $\textbf{B}_1- \textbf{B}_4$, which will be further discussed in the final version of the paper. 
Similarly, $1-dq>0$ represents a comparison between $d$ and the average \textcolor{black}{waiting time} for an $ON$ channel state given that the last channel state is $ON$. Finally, \textcolor{black}{$(1-q)/(1-p)+1>d$ represents a comparison between $d$ and the average \textcolor{black}{waiting time} of Channel $1$ under steady-state distribution of the Gilbert-Elliot model}. These comparisons interpret all the boundary functions $F,G,H$ of the regions $\textbf{B}_1 - \textbf{B}_4$. 
The four regions $\textbf{B}_1,\ldots,\textbf{B}_4$ are depicted in Fig. \ref{fig3}, \textcolor{black}{for the case that} $d=10$.

Consider a stationary policy $\mu(\delta,l_1,l_2)$. As mentioned in Lemma \ref{zerowait}, when $l_2 = 0$, the decision $\mu(\delta,l_1,0)$ can be $1$ (Channel $1$) or $2$ (Channel $2$). 
Given the value of $l_1$, $\mu(\delta,l_1,0)$ is said to be \emph{non-decreasing in the age $\delta$}, if
\begin{equation}\label{eq_increasing}
    \mu(\delta,l_1,0)=\left\{
\begin{array}{lll}
  1  & \text{if }~ \delta< \lambda;\\
  2  & \text{if }~ \delta\ge\lambda.
\end{array}
\right. 
\end{equation} 
Conversely, $\mu(\delta,l_1,0)$ is said to be \emph{non-increasing in the age $\delta$}, if 
\begin{equation}\label{eq_decreasing}
    \mu(\delta,l_1,0)=\left\{
\begin{array}{lll}
  2 & \text{if }~ \delta< \lambda;\\
  1  & \text{if }~ \delta\ge\lambda.
\end{array}
\right. 
\end{equation} %
One can observe that  scheduling policies in the form of \eqref{eq_increasing} and \eqref{eq_decreasing} are both with a \emph{threshold-type}, where $\lambda$ is the threshold on the age $\delta$ at which the value of $\mu(\delta,l_1,0)$ changes. 

%where the action $u\in \{1,2\}$ and the action $\bar{u}$ is the complement action of $u$. 
One optimal solution to Problem \eqref{avg} is of a special threshold-type structure, as stated in the following theorem:

\begin{theorem}
There exists an optimal solution $\mu^*(\delta,l_1,0)$ to Problem \eqref{avg}, which satisfies the following properties: 
\begin{itemize}
\item[(a)] if $(p,q,d)\in \textbf{B}_1$, then $\mu^*(\delta, 0, 0)$ is non-increasing in the age $\delta$ and $\mu^*(\delta, 1, 0)$ is non-increasing in the age $\delta$; 

\item[(b)] if $(p,q,d)\in \textbf{B}_2$, then $\mu^*(\delta, 0, 0)$ is non-decreasing in the age $\delta$ and $\mu^*(\delta, 1, 0)$ is non-increasing in the age $\delta$; 

\item[(c)] if $(p,q,d)\in \textbf{B}_3$, then $\mu^*(\delta, 0, 0)$ is non-decreasing in the age $\delta$ and $\mu^*(\delta, 1, 0)$ is non-decreasing in the age $\delta$; 

\item[(d)] if $(p,q,d)\in \textbf{B}_4$, then $\mu^*(\delta, 0, 0)$ is non-increasing in the age $\delta$ and $\mu^*(\delta, 1, 0)$ is non-decreasing in the age $\delta$. 

%\cup \textbf{B}_4

\end{itemize}
\label{theorem2}
\end{theorem}

\begin{proof}
%We provide a sketch of the proof. The detailed proof is provided in Section \ref{proof2}.
\ifreport
See Section \ref{proof2} for the proof. 
\else
See Section 7.2 of our supplementary material.
\fi
\end{proof}

%\begin{remark}
As is shown in Theorem \ref{theorem2}, for all possible parameters \textcolor{black}{$(p,q,d)$} of the two channels, the optimal action $\mu^*(\delta,l_1,0)$ of channel selection is a monotonic function of the age $\delta$. Whether $\mu^*(\delta,l_1,0)$ is non-decreasing or non-increasing in $\delta$ depends on \textcolor{black}{the region of} the channel parameters $(p,q,d)$ and the previous state  $l_1$ of Channel 1. %Moreover, $\textbf{B}_1 \ldots \textbf{B}_4$ provide sufficient conditions for the monotonicity of $\mu^*(\delta,l_1,0)$. %For example, in Region $\textbf{B}_1$, $\mu^*(\delta,0,0)$ is non-increasing in $\delta$ and $\mu^*(\delta,1,0)$ is non-increasing in $\delta$. 

%The state $l_1=0$ means that the state of Channel $1$ in the previous time slot is $OFF$. Also, the region $\textbf{B}_1\cup \textbf{B}_4$ implies $F\le 0$, i.e., $1-p \ge 1/d$. Theorem \ref{theorem2} (a) and (b) indicate that whether the policy $\mu^*(\delta,0,0)$ is non-increasing or non-decreasing depends on the comparison between the transition probability of Channel $1$ from $OFF$ to $ON$ $(1-p)$ and the transmission rate of Channel $2$ $(1/d)$. %If $1-p \ge 1/d$, then $\mu^*(\delta,0,0)$ is non-increasing, i.e., the policy chooses Channel $2$ if age is smaller or equal to a threshold, and chooses Channel $1$ if age is larger.

%The state $l_1=1$ means that the state of Channel $1$ in the previous time slot is $ON$. Also, the region $\textbf{B}_1\cup \textbf{B}_2$ implies $G\le 0$ or $H\le 0$, i.e., $q \ge 1/d$ and $(1-p)/(1-q+1-p)\ge 1/d$. Theorem \ref{theorem2} (c) and (d) indicate that whether the policy $\mu^*(\delta,0,0)$ is non-increasing or non-decreasing not only depends on the comparison between the transition probability of Channel $1$ from $ON$ to $OFF$ $(q)$ and $1/d$, but also depends on the comparison between the steady state probability $P(ON)$ and $1/d$.

The study in \cite{altman2019forever} assumed that the first channel is unreliable and consumes a lower cost, and the second channel the same delay as the first channel but a higher cost. They studied the scheduling policy for optimizing the trade-off between age and cost.
The optimal scheduling policy in Theorem \ref{theorem2} is quite different from that in \cite{altman2019forever}: The study in \cite{altman2019forever} assumes the first channel to be i.i.d., but our result allows a Markovian Channel $1$, which is a generalization of the i.i.d. case. \textcolor{black}{Observe that in \cite{altman2019forever}, the first channel is no better than the second channel with regard to delay and reliability. However, in our study, the two channels (i.e., Channel $1$ and $2$) have their own advantages in delay and reliability. Therefore, the optimal solution in our study is non-decreasing in age for some values of $(p,q,d)$ and non-increasing in age for the remaining values of $(p,q,d)$}. In conclusion, our study allows for general channel parameters $(p,q,d)$ and \textcolor{black}{our optimal decision $\mu^*(\delta,l_1,0)$ is non-increasing in age or non-decreasing in age depending on the choices of channel parameters}. 

\subsubsection{Insights Behind the Regions $\textbf{B}_1 - \textbf{B}_4$}  

%\subsection{Insights Behind the Regions $\textbf{B}_1 - \textbf{B}_4$} 

The regions $\textbf{B}_1- \textbf{B}_4$ were introduced in Theorem \ref{theorem2} for proving that the action value function $Q(\textbf{s},u)$ is  \emph{super-modular} or \emph{sub-modular}, where $\textbf{s}=(\delta, l_1,0)$ denotes the state of the MDP and $u$ is the action. For example, in the case of $l_1 = 0$, if ${1}/(1-p) > d$ and ${1}/{q} \le d$ (i.e., $(p,q,d)\in \textbf{B}_2$), 
\ifreport
Lemma \ref{l0m} in Section \ref{proof2}
\else
Lemma 9 in Section 7.2 of our supplementary material
\fi
showed that $Q(\delta,0,0,u)$ is sub-modular in $(\delta,u)$ (in the discounted case).
As a result, the optimal action $\mu^*(\delta,0,0)$ is increasing in $\delta$.

However, in the case $l_1 = 1$ of Theorem \ref{theorem2}, there are additional technical challenges: For example, if $(p,q,d)\in \textbf{B}_2$, $Q(\delta,1,0,u)$ is \textcolor{black}{neither super-modular nor} sub-modular. A new method was developed in 
\ifreport
Lemma \ref{l1m} in Section \ref{proof2}
\else
Lemma 10 in Section 7.2 of our supplementary material
\fi
to conquer this challenge. Technically, super-/sub-modularity is a sufficient but not necessary condition for the monotonicity of $\mu^*(\delta,l_1,0)$. \textcolor{black}{When} neither super-modularity nor sub-modularity holds, we are able to show that the optimal decision $\mu^*(\delta,l_1,0)$ does not change with $\delta$.
% given $l_1=1$. 
By this, we proved the monotonicity of $\mu^*(\delta,l_1,0)$ \textcolor{black}{for all values of $\delta$ and $l_1$}, without requiring $Q(\textbf{s},u)$ to be super-modular or sub-modular over the entire state space $\textbf{s}\in  \textbf{S}$. 

%Hence, the regions $\textbf{B}_1- \textbf{B}_4$ provide sufficient conditions for the monotonicity of the optimal action $\mu^*(\delta,l_1,0)$, %as stated in Theorem 1. %For example, in region $\textbf{B}_1$, $\mu^*(\delta,0,0)$ is decreasing in $\delta$ and $\mu^*(\delta,1,0)$ is decreasing in $\delta$. 
%But, the regions $\textbf{B}_1- \textbf{B}_4$ do notprovide sufficient conditions 
%but not for super-modularity or sub-modularity. 
The following is one of the key technical contributions of the paper: we proved that the optimal action $\mu^*(\delta,l_1,0)$ is monotonic in $\delta$ even if super-/sub-modularity does not hold. This is a key difference from prior studies, e.g., \cite{krishnamurthy2016partially,puterman1990markov,ngo2009optimality,altman2019forever,sun2017update,sun2019sampling}, where super-modularity (or sub-modularity) holds for the entire state space.

\subsection{Optimal Scheduling Policy}\label{C}

%According to Theorem \ref{theorem2}, $\mu^*(\delta,0,0)$ and $\mu^*(\delta,1,0)$ are both threshold-type, so there are two thresholds. 

In Theorem \ref{theorem2}, we have characterized the threshold structure for an optimal policy in region $\textbf{B}_1,\ldots,\textbf{B}_4$. A threshold-type policy is fully identified by its thresholds $\lambda_0,\lambda_1$, where $\lambda_0$ is the threshold given that \textcolor{black}{the} previous state of Channel $1$ is $OFF$ (i.e., $l_1=0$) and $\lambda_1$ is the threshold given that the previous state of Channel $1$ is $ON$ (i.e., $l_1=1$). Thus, for a given region $\textbf{B}_i$ $(i=1,\ldots,4)$, the MDP problem \eqref{avg} reduces to 
\begin{equation}\label{help}
   \bar{\Delta}_{\text{opt}} = \min_{\lambda_0 \in \mathbb{N}^+,\lambda_1 \in \mathbb{N}^+} \bar{\Delta}_i (\lambda_0,\lambda_1),
\end{equation}    
where $\bar{\Delta}_i (\lambda_0,\lambda_1)$ is the long term average cost of the threshold-type policy such that: $(1)$ the threshold (monotone) structure is determined by Theorem \ref{theorem2} and $\textbf{B}_i$; $(2)$ the thresholds are $\lambda_0,\lambda_1$. Note that a threshold type policy is stationary and thus can be modeled as a \textcolor{black}{discrete-time Markov chain (DTMC). Then, \eqref{help} can be solved by deriving the steady-state} distribution of the DTMC.  

We use $\lambda_0^*$ and $\lambda_1^*$ to denote the thresholds of $\mu^*(\delta,0,0)$ and $\mu^*(\delta,1,0)$, respectively. In this section, we provide the optimal scheduling policy and the thresholds.

\subsubsection{Optimal Scheduling Policy for $(p,q,d)\in \textbf{B}_1$}\label{C1}
%Given $\beta_1,\ldots,\beta_4$, we then provide our optimal average age and policy. 
\begin{theorem}\label{theorem2a}
%An optimal solution to \eqref{avg} is presented below for the
% 4 regions $\textbf{B}_1,\ldots,\textbf{B}_4$
%of the channel parameters:
If $(p,q,d)\in \textbf{B}_1$, then an optimal scheduling policy is 
\begin{align}
&\mu^*(\delta,0,0)= 1, \delta \ge 1;\\
&\mu^*(\delta,1,0)=1, \delta \ge 1. 
\end{align}
In this case, the optimal objective value of \eqref{avg} is
\begin{equation}
   \bar{\Delta}_{\text{opt}} = \frac{(1-q)(2-p)+(1-p)^2}{(2-q-p)(1-p)}.
   \label{always1}
\end{equation}    
% \triangleq \bar{\Delta}
%The thresholds $\lambda_0^*=1$, and $\lambda_1^*=1$ if $(p,q,d)\in \textbf{B}_1$. The thresholds $\lambda_0^*=1$, and $\lambda_1^*\in \{2,3,\ldots \}$ if $(p,q,d)\in \textbf{B}_4$.

\end{theorem}
We provide an insight to Theorem \ref{theorem2a}:
\textcolor{black}{As will be shown by} 
\ifreport
Lemma \ref{mu00in1and4} and Lemma \ref{mu110for1} in Section \ref{proof3}, 
\else
Lemma 11 and Lemma 12 in Section 7.3 of our supplementary material,
\fi
if $(p,q,d)\in \textbf{B}_1$, then $\mu^*(1,0,0) =1 $ and $ \mu^*(1,1,0)=1$. According to Theorem \ref{theorem2} (a), if $(p,q,d)\in \textbf{B}_1$, $\mu^*(\delta,0,0)$ and $\mu^*(\delta,1,0)$ are \textcolor{black}{both non-increasing in $\delta$}. Thus, $\mu^*(\delta,0,0)=1$ and $\mu^*(\delta,1,0)=1$ for all $\delta\geq 1$. That is, the optimal \textcolor{black}{scheduler} always chooses Channel $1$. The DTMC for a policy always choosing Channel $1$ is easy to analyze. We omit the derivation steps and \textcolor{black}{provide}  
\begin{equation}
   \bar{\Delta}_{\text{opt}} = \bar{\Delta}_1 (1,1)= \frac{(1-q)(2-p)+(1-p)^2}{(2-q-p)(1-p)}.
\end{equation} 
This result directly implies Theorem \ref{theorem2a}.

\subsubsection{Optimal Scheduling Policy for $(p,q,d)\in \textbf{B}_2$}\label{C2}

While the result of case $(p,q,d)\in \textbf{B}_1$ is easy to describe, the result of case $(p,q,d)\in \textbf{B}_2$ is not. \textcolor{black}{As shown by Theorem \ref{theorem2b}, the optimal decision $\mu^*(\delta,l_1,0)$ is not constant in age $\delta$}.  

\begin{theorem}\label{theorem2b}

If $(p,q,d)\in \textbf{B}_2$, then an optimal scheduling policy is 
\begin{align}
    \mu^*(\delta,0,0)=\left\{
\begin{array}{lll}
  1 & \text{if }~ \delta< \lambda_0^*;\\
  2  & \text{if }~ \delta\ge\lambda_0^*, 
\end{array}
\right. \\
    \mu^*(\delta,1,0)=\left\{
\begin{array}{lll}
  2 & \text{if }~ \delta< \lambda_1^*;\\
  1  & \text{if }~ \delta\ge\lambda_1^*, 
\end{array}
\right. 
\end{align} 

\begin{algorithm}[!htbp]\label{alg1}
\caption{Bisection method for solving \eqref{beta}}
\textbf{Given} function $h_i$. $l=0$, $l'$ sufficiently large, tolerance $\epsilon$ small. The value
$i\in \{1,2,3,4\}$.

\textbf{repeat}

\qquad $\beta=\frac{1}{2}(l+l')$

\qquad \textbf{if} $h_i(\beta)<0$: $l'=\beta$. \textbf{else} $l=\beta$

\textbf{until} $l'-l<\epsilon$

\textbf{return} $\beta_{i}=\beta$
\end{algorithm}

\noindent where $\lambda_0^*$ is unique, but $\lambda_1^*$ may take multiple values, given by  
\begin{equation}\label{theorem2bb}
 \!\!\!\!\!   \left\{
\begin{array}{lll}
  \lambda_0^*=s_{1}(\beta_{1}), &\lambda_1^*=1  & \text{if }~ \bar{\Delta}_{\text{opt}} =  \beta_{1},\\
  \lambda_0^*=s_{2}(\beta_{2}), &\lambda_1^*=1  & \text{if }~ \bar{\Delta}_{\text{opt}} =  \beta_{2},\\
 \lambda_0^*=1, & \lambda_1^*\in \{2,3,\ldots,d 
\}  & \text{if }~ \bar{\Delta}_{\text{opt}} = f_0/g_0,\\
\lambda_0^*=1, & \lambda_1^*\in \{d+1,\ldots 
\} & \text{if }~\bar{\Delta}_{\text{opt}} =  (3/2)d-1/2,
\end{array}
\right. \!\!\!\!\!
\end{equation}
$\bar{\Delta}_{\text{opt}}$ is the optimal objective value of \eqref{avg}, determined by
\begin{equation}
   \bar{\Delta}_{\text{opt}} =  \min \Big{\{}   \beta_{1}, \beta_{2}, \frac{f_0}{g_0},\frac{3}{2}d-\frac{1}{2} \Big{\} },
\end{equation} 
$s_{1}(\cdot)$, $s_{2}(\cdot)$, $\beta_1$, and $\beta_2$ are given in Definition \ref{defthm} below, and
\begin{align}
&   f_0 = q\sum_{i=1}^{d}i + (1-q)\sum_{i=d+1}^{2d}i +\Big{(} \frac{b'_d q+b_d}{1-b_d}+1 \Big{)} \sum_{i=d+1}^{2d}i,\\
& g_0 = \frac{b'_d q+b_d}{1-b_d} d + d+1,\\
& \begin{bmatrix} b'_d \\ b_d \end{bmatrix} = \begin{bmatrix} q & 1-q \\ 1-p & p \end{bmatrix}^d \begin{bmatrix} 0 \\ 1 \end{bmatrix}.\label{eq_matrix}
\end{align}

\end{theorem}
%\linespread{1.2}

\begin{proof}
\ifreport
See Section \ref{proof3}.
\else
See Section 7.3 of our supplementary material.
\fi
\end{proof}

%Due to the various choices of the channel parameters and the multi-dimensional state space, Theorem \ref{theorem2b} contains a number of cases which is explained in Section \ref{proof3}. The result of Theorem \ref{theorem2a} is simple: if $(p,q,d)\in \textbf{B}_1$, the optimal policy always chooses Channel $1$. However, Theorem \ref{theorem2b} contains a number of cases.  

\textcolor{black}{In order to prove Theorem \ref{theorem2b}, we have conducted steady-state analysis of four DTMCs, each of which corresponds to one case in \eqref{theorem2bb}. These four DTMCs have diverse state transmission matrices and have to be analyzed separately.}

For each case, the optimal thresholds $\lambda_0^*$ and $\lambda_1^*$ can be either expressed in closed-form, or computed by using a low-complexity bisection search method to compute the root of \eqref{beta} given in below.

\begin{definition}\label{defthm} The value of $\beta_i$ is the root of
\begin{equation}\begin{split}
    %h_{i}(\beta_i)=
% &  \lim_{\beta\rightarrow \beta_i^-}  f_{i}\big{(}s_i(\beta)\big{)}-\beta  g_{i}\big{(}s_i(\beta)\big{)}\ge 0,\\
% & \lim_{\beta\rightarrow \beta_i^+}  f_{i}\big{(}s_i(\beta)\big{)}-\beta  g_{i}\big{(}s_i(\beta)\big{)}\le 0 .
f_{i}\big{(}s_i(\beta_i)\big{)}-\beta_i  g_{i}\big{(}s_i(\beta_i)\big{)}= 0,  ~i\in\{1,2,3,4\},
\label{beta}
\end{split}
\end{equation}
%where $ \bar{f}_{i}(\beta_i)\triangleq f_{i}(s_{i}(\beta_i))$ and $ \bar{g}_{i}(\beta_i)\triangleq g_{i}(s_{i}(\beta_i))$. 
where %The function $s_{i}(\beta_i)$ is defined as
\begin{eqnarray}
\nonumber  &&  s_{i}(\beta_i)= \max \left\{ \left\lceil \frac{-k_{i}(\beta_i)}{1-d(1-p)}\right\rceil , d \right\},~ i\in \{ 1,3,4\}, \\
   &&    s_{2}(\beta_2)=\max \left\{ \min \left\{ \left\lceil \frac{-k_{2}(\beta_2)}{1-d(1-p)}\right\rceil , d \right\}  , 2 \right\},
\label{threshold1}\\
  &&  k_{i}(\beta_i)=l'_{i}-\beta_i o_{i},\label{threshold2}
\end{eqnarray}  %the function $k_{i}(\beta_i)$ is defined as 
and $\lceil x \rceil$ is the smallest integer that is greater or equal to $x$. %We denote $a\vee b=\max \{ a,b \}$, and $a\wedge b=\min \{ a,b \}$.
For the ease of presentation, $16$ closed-form expressions of  $f_{i}(\cdot)$, $g_{i}(\cdot)$, $l'_{i}$, and $o_{i}$ for $i=1,\ldots, 4$
\ifreport 
are provided in Table \ref{fg}. 
\else
are provided in Table 2 of our supplementary material.
\fi
\end{definition}

%Definition \ref{defthm} tells that $\beta_1,\ldots,\beta_4$ are the root of a closed-form function. 
Note that $\beta_3$ and $\beta_4$ in Definition \ref{defthm} will be used later \textcolor{black}{when} $(p,q,d)\in \textbf{B}_3$. 
For notational simplicity, we define 
\begin{eqnarray}
 h_i(\beta)= f_{i}(s_i(\beta))-\beta  g_{i}(s_i(\beta)),  ~i\in\{1,2,3,4\}.
\end{eqnarray}
The function $h_i(\beta)$
has the following nice property: 
\begin{lemma}
For all $i\in \{ 1,2,3,4 \}$, the function $h_{i}(\beta)$ satisfies the following properties: 
%The solution to \eqref{beta} is simplified because of the following lemma:

(1) $h_{i}(\beta)$ is continuous, concave, and strictly decreasing on $\beta$;

(2) $h_{i}(0)>0$ and $\lim_{\beta\rightarrow \infty}h_{i}(\beta)=-\infty$. \label{concave}
\end{lemma}
\begin{proof} 
\ifreport
See Appendix \ref{concaveapp}.
\else
See Appendix B of our supplementary material.
\fi
\end{proof}

Lemma \ref{concave} implies that \eqref{beta} has a unique root on $[0,\infty)$. Therefore, we can use a low-complexity bisection method to compute 
$\beta_1,\ldots,\beta_4$ 
as illustrated in Algorithm \ref{alg1}.

Lemma \ref{concave} is motivated by Lemma $2$ in  \cite{sun2019samplings} and Lemma $2$ in \cite{ornee2019samplings}. In \cite{sun2019samplings} and \cite{ornee2019samplings}, since the channel is error free, the age state at the end of each transmission is independent with history information. Thus, Lemma $2$ in \cite{sun2019samplings} and Lemma $2$ in \cite{ornee2019samplings} are related with a per-sample (single transmission) control. However, our study does not have such a property and Lemma \ref{concave} arises from \textcolor{black}{solving \eqref{help} for optimizing the thresholds $\lambda^*_0$ and $\lambda^*_1$. }     

The advantage of Theorem \ref{theorem2a}, Theorem \ref{theorem2b} is that the solution is easy to implement. In Theorem \ref{theorem2a}, we showed that the optimal policy is a constant policy that always chooses Channel $1$. 
 In Theorem \ref{theorem2b}, $\bar{\Delta}_{\text{opt}}$ is expressed as the minimization of only a few precomputed values, \textcolor{black}{and the optimal threshold-type policy are then obtained based on the value of} $\bar{\Delta}_{\text{opt}}$.  
%Observe that $\beta_{1},\ldots,\beta_{4}$ in Theorem \ref{theorem3} (b) and (c) are the root of the function $h_1(\beta_1)=0,\ldots,h_4(\beta_4)=0$ in \eqref{beta}, respectively. 

Since we can use a low complexity algorithm such as bisection method to obtain $\beta_1,\beta_2$ in Theorem \ref{theorem2b}, Theorem \ref{theorem2b} provides a solution that has much lower complexity than other solutions for MDPs such as \textcolor{black}{relative value iteration and policy iteration}.

%\subsubsection{ Insights behind Theorem \ref{theorem3}}\label{insight-theorem2}

 We now provide the sketch of the proof when $(p,q,d)\in \textbf{B}_2$: %The discussions for Theorem \ref{theorem3}(c),(d) are similar with Theorem \ref{theorem3}(b) and the proofs are provided jointly with Theorem \ref{theorem3}(b) in Section \ref{mainproof}.     

%We provide the key idea in this section and relegate the detailed proof to our technical report \cite{pan2020age}.
First, by computing the \textcolor{black}{steady-state distributions of some DTMCs with different thresholds, we have obtained the average age performance for four cases, given by} 
\begin{align}
& \bar{\Delta}_2 (\lambda_0,1) = \left\{
\begin{array}{lll}
   f_1(\lambda_0)/g_1(\lambda_0) & \lambda_0 \in \{ d+1, \ldots \}, \\
   f_2(\lambda_0)/g_2(\lambda_0) &\lambda_0 \in \{ 2, \ldots d \}, 
\end{array}
\right.  \label{discuss-24} \\
& \bar{\Delta}_2 (1,\lambda_1) = \left\{
\begin{array}{lll}
   (3/2)d - 1/2  & \lambda_1 \in \{ d+1, \ldots \}, \\
  f_0/g_0 &\lambda_1 \in \{ 1, \ldots d \}.
\end{array}
\right.  \label{discuss-25}
\end{align}
\textcolor{black}{Note that each one of the four expressions in \eqref{discuss-24} and \eqref{discuss-25} corresponds to each one of the four cases in \eqref{theorem2bb}, respectively.
One of our technical contributions is that only studying the steady-state analysis of the $4$ types of DTMCs in \eqref{discuss-24}, \eqref{discuss-25} is sufficient to solve \eqref{help}. The proof of this statement and the detailed expressions of the DTMC structure of the four cases in \eqref{discuss-24}, \eqref{discuss-25} are} relegated to 
\ifreport
Section \ref{proof3}
\else
Section 7.3 of our supplementary material\fi\footnote{Although \eqref{help} is a two-dimensional optimization problem in $(\lambda_0,\lambda_1)$, \eqref{help} has been simplified as \eqref{discuss-24} or \eqref{discuss-25}, which are one-dimensional optimization problem. For example, as is shown in \eqref{discuss-25}, the threshold-type policies with different $\lambda_1$ may have the same DTMC.}.
\textcolor{black}{Therefore, the optimal average age $\bar{\Delta}_{\text{opt}}$ chooses the smallest value of the four cases from \eqref{discuss-24}, \eqref{discuss-25},}
\begin{equation}\label{discuss-delta1}
   \bar{\Delta}_{\text{opt}} =  \min \Big{\{}   \beta'_{1}, \beta'_{2}, \frac{f_0}{g_0},\frac{3}{2}d-\frac{1}{2} \Big{\} },
\end{equation} 
 where $\beta'_1,\beta'_2$ are defined as follows: 
\begin{align}
\beta'_{i}& =
\min_{\lambda_0 \in \{ d+1,...\} }\frac{f_{i}(\lambda_0)}{g_{i}(\lambda_0)}, i\in \{1,3,4\} , \label{discuss-beta1}\\
\beta'_{2}&=
\min_{\lambda_0 \in \{2,...d\} }\frac{f_{2}(\lambda_0)}{g_{2}(\lambda_0)}.  \label{discuss-beta2}
\end{align}
\textcolor{black}{Note that $\beta_3'$ and $\beta_4'$ in \eqref{discuss-beta1} and \eqref{discuss-beta2} will be used later when $(p,q,d)\in \textbf{B}_3$}.
\ifreport
Finally, in Section \ref{proof3}, we show that
\else
Finally, in Section 7.3 of our supplementary material, we show that
\fi
\begin{equation}\label{discuss-beta1-beta}
\beta_i' = \beta_i, \  i\in \{1,2,3,4\}.
\end{equation}  

Thus, Theorem \ref{theorem2b} is solved by \eqref{discuss-24} $-$ \eqref{discuss-beta1-beta}.

\subsubsection{Optimal Scheduling Policy for $(p,q,d)\in \textbf{B}_3$}\label{C3}
\textcolor{black}{According to Theorem \ref{theorem2}, the optimal decision $ \mu^*(\delta,l_1,0)$ is non-decreasing in age $\delta$.
%The result of $(p,q,d)\in \textbf{B}_3$ has the similar form of Theorem \ref{theorem2b}.
Similar to the case $(p,q,d)\in \textbf{B}_2$ in Theorem \ref{theorem2b}, the optimal solution $ \mu^*(\delta,l_1,0)$ is not constant. Therefore, we need to solve the optimal thresholds $\lambda_0^*$ and $\lambda_1^*$ by deriving the steady-state distribution of the DTMC. The final result is presented as follows: }
\begin{theorem}\label{theorem2c}
If $(p,q,d)\in \textbf{B}_3$, then an optimal scheduling policy is 
\begin{align}
    \mu^*(\delta,0,0)=\left\{
\begin{array}{lll}
  1 & \text{if }~ \delta< \lambda_0^*;\\
  2  & \text{if }~ \delta\ge\lambda_0^*, 
\end{array}
\right. \\
    \mu^*(\delta,1,0)=\left\{
\begin{array}{lll}
  1 & \text{if }~ \delta< \lambda_1^*;\\
  2  & \text{if }~ \delta\ge\lambda_1^*, 
\end{array}
\right. 
\end{align} 
where $\lambda_0^*$ is unique, but $\lambda_1^*$ may take multiple values, given by 
%The thresholds $\lambda_0^*$ and $\lambda_1^*$ are listed as follows:
%The thresholds $\lambda_0^*$ is a unique value and $\lambda_1^*$ may take multiple values, which are listed as follows:
\begin{equation}\label{theorem2cc}
\!\!\!\!    \left\{
\begin{array}{lll}
  \lambda_0^*=s_{1}(\beta_{1}), &\lambda_1^*\in \{d+1,\ldots 
\}  & \text{if~} \bar{\Delta}_{\text{opt}} =  \beta_{1},\\
  \lambda_0^*=s_{2}(\beta_{2}), &\lambda_1^*\in \{d+1,\ldots 
\}  & \text{if~} \bar{\Delta}_{\text{opt}} =  \beta_{2},\\
 \lambda_0^*=s_{3}(\beta_{3}), & \lambda_1^*\in \{2,\ldots,d 
\} & \text{if~} \bar{\Delta}_{\text{opt}} =  \beta_{3},\\
\lambda_0^*=s_{4}(\beta_{4}), & \lambda_1^*\in \{2,\ldots,d 
\} & \text{if~}\bar{\Delta}_{\text{opt}} =  \beta_{4},\\ \lambda_0^*=1,& \lambda_1^*\in \{1,2,\ldots,d\},& \text{if~} \bar{\Delta}_{\text{opt}} =  (3/2)d-1/2, 
\end{array}
\right. \!\!\!\!
\end{equation}
 $\bar{\Delta}_{\text{opt}}$ is the optimal objective value of \eqref{avg}, determined by
\begin{equation}\label{2cmin}
   \bar{\Delta}_{\text{opt}} =  \min \Big{\{}   \beta_{1}, \beta_{2}, \beta_{3}, \beta_{4}, \frac{3}{2}d-\frac{1}{2} \Big{\} },
\end{equation}
$s_{1}(\cdot),\ldots,s_{4}(\cdot)$ and $\beta_1,\ldots,\beta_4$ are given in Definition \ref{defthm}.
\end{theorem}
\begin{proof}
\ifreport
See Section \ref{proof3}.
\else
See Section 7.3 of our supplementary material.
\fi
\end{proof}
\textcolor{black}{To show Theorem \ref{theorem2c}, we have analyzed five different DTMCs. Each of the DTMC corresponds to one case in \eqref{theorem2cc}. As is explained in 
\ifreport
Section \ref{proof3}, 
\else
Section 7.3 of our supplementary material,
\fi
the solution to each case in \eqref{theorem2cc} is closed-form or related with a one-dimensional optimization problem. Different from Theorem \ref{theorem2b} which needs to compute $\beta_1$ and $\beta_2$ in \eqref{theorem2bb}, Theorem \ref{theorem2c} needs to compute $\beta_1,\ldots,\beta_4$ in \eqref{theorem2cc}. 
By Definition \ref{defthm} and Lemma \ref{concave}, $\beta_1,\ldots,\beta_4$ can be solved by using low complexity bisection search algorithm (Algorithm~\ref{alg1}). Therefore, despite Theorem \ref{theorem2c} containing a number of cases, the optimal thresholds described in \eqref{theorem2cc} can be efficiently solved. }  

\subsubsection{Optimal Scheduling Policy for $(p,q,d)\in \textbf{B}_4$}\label{C4}
From Theorem \ref{theorem2}, $ \mu^*(\delta,0,0)$ is non-increasing in age $\delta$ and $\mu^*(\delta,1,0)$ is non-decreasing in $\delta$.
The result of $(p,q,d)\in \textbf{B}_4$ is similar to that of Theorem \ref{theorem2a}. 

\begin{theorem}\label{theorem2d}
If $(p,q,d)\in \textbf{B}_4$, then an optimal scheduling policy is
\begin{align}
%{\color{red} so left}
&\mu^*(\delta,0,0)=1,\delta \ge 1, \\
&   \mu^*(\delta,1,0)=\left\{
\begin{array}{lll}
  1,  \delta\ge 1 & \text{if }~  \bar{\Delta}_{\text{opt}}=\bar{\Delta};\\
  2, \delta\ge 1 & \text{if }~ \bar{\Delta}_{\text{opt}}=f'_0/g'_0,
\end{array}
\right. 
\end{align} 
where $\bar{\Delta}_{\text{opt}}$ is the optimal objective value of \eqref{avg}, determined by
\begin{equation}\label{theorem2cdelta}
      \bar{\Delta}_{\text{opt}} =  \min \Big{\{ }   \bar{\Delta}, \frac{f'_0}{g'_0} \Big{ \} },
\end{equation}
the constants $\bar{\Delta},f'_0,g'_0$ are given by
\begin{align}
\bar{\Delta}&  = \frac{(1-q)(2-p)+(1-p)^2}{(2-q-p)(1-p)},\\
f'_0 & = \sum_{i=1}^{d} i+\frac{1-b'_d}{b'_d} \times \sum_{i=d}^{2d-1} i+\sum_{i=d}^{\infty}ip^{i-d}, \label{f'0} \\ 
g'_0 & = \frac{d}{b'_d}+1/(1-p) \label{g'0}.
\end{align}

\end{theorem}
\begin{proof}
\ifreport
See Section \ref{proof3}.
\else
See Section 7.3 of our supplementary material.
\fi
\end{proof}
\textcolor{black}{As is illustrated in Theorem \ref{theorem2d}, the proposed optimal decision $\mu^*(\delta,0,0)$ for $(p,q,d)\in \textbf{B}_4$ is constant in age $\delta$, depending on whether $ \bar{\Delta}_{\text{opt}}=\bar{\Delta}$ or $\bar{\Delta}_{\text{opt}}=f'_0/g'_0$ from \eqref{theorem2cdelta}. The value $\bar{\Delta}$ is the expected age of the steady-state DTMC that always chooses Channel $1$. The value $f'_0/g'_0$ is the expected age of the steady-state DTMC that chooses Channel $1$ if $l_1=0$ and chooses Channel $2$ if $l_1=1$}. If $ \bar{\Delta}_{\text{opt}}=\bar{\Delta}$, then it is optimal to always choose Channel $1$; if $\bar{\Delta}_{\text{opt}}=f'_0/g'_0$, then we will select Channel $1$ when $l_1=0$ and Channel $2$ when $l_1=1$.  

We briefly summarize the results for Theorem \ref{theorem2a}---\ref{theorem2d}: An optimal solution to \eqref{avg} is presented for the 4 complementary regions $\textbf{B}_1,\ldots,\textbf{B}_4$ of the channel parameters $(p,q,d)$. If $(p,q,d)\in \textbf{B}_1 \cup \textbf{B}_4$, the solution is constant in age (Theorem~\ref{theorem2a} and Theorem~\ref{theorem2d}). \textcolor{black}{Otherwise, for $(p,q,d)\in \textbf{B}_2 \cup \textbf{B}_3$, there exists an optimal scheduling policy that has a threshold structure depending on the current age value and the previous state of Channel $1$ (Theorem \ref{theorem2b} and Theorem \ref{theorem2c}). Further, the optimal thresholds can be computed efficiently. }    

\subsection{Optimal Scheduling policy for i.i.d. Channel}

We finally consider a special case in which Channel $1$ is i.i.d., i.e., $p+q=1$.
\textcolor{black}{First, according to the following lemma, if $p+q=1$, the $4$ regions $ \textbf{B}_1,\ldots, \textbf{B}_4$ will reduce to 2 regions $ \textbf{B}_1, \textbf{B}_3$.} 
\begin{lemma}
 If $p+q=1$, then $(p,q,d)\in \textbf{B}_1$ or $(p,q,d)\in \textbf{B}_3$. Moreover, 
 \begin{align}
   \textbf{B}_1 & = \left\{ (p,q,d): \frac{1}{1-p}\le d \right\}, \label{textb1iid} \\ 
   \textbf{B}_3 & =\left\{ (p,q,d): \frac{1}{1-p}> d \right\}. \label{textb3iid}
 \end{align}
\end{lemma}
\begin{figure}[htbp]
\centerline{\includegraphics[width=0.4\textwidth]{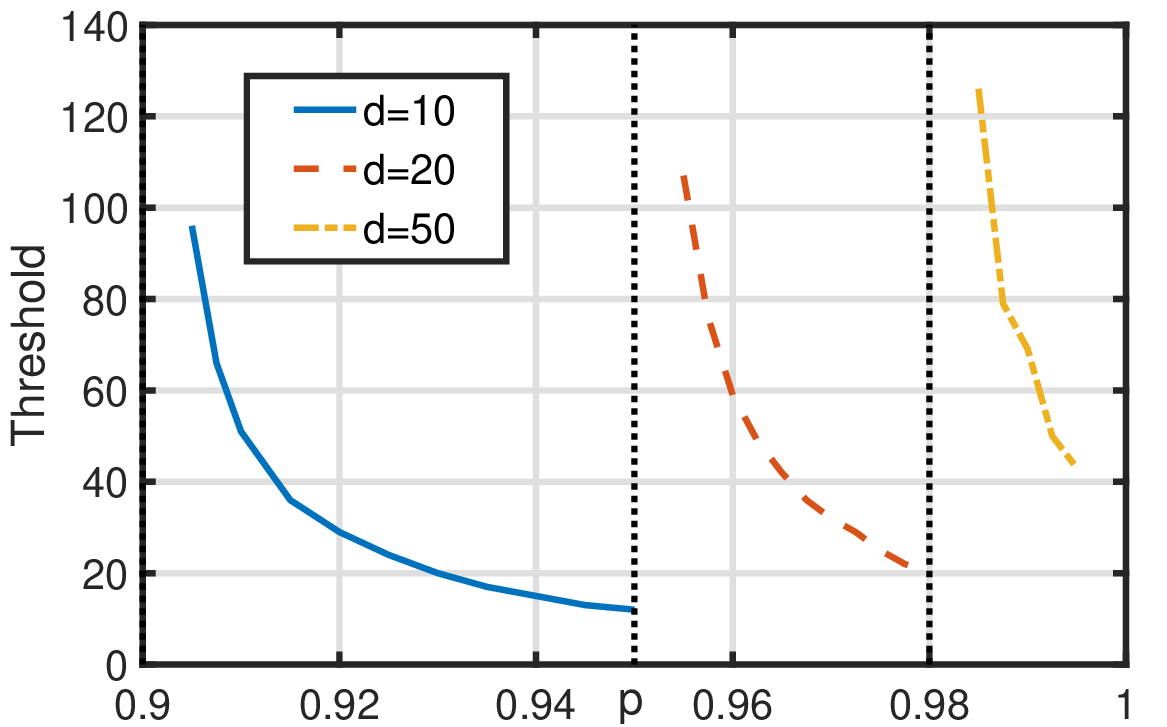}}
\caption{Thresholds of the optimal scheduling policy for i.i.d. mmWave channel state, where the packet transmission time of the sub-6GHz channel is $d= 10,20,50$.}
\label{fig8}
\end{figure}
\begin{figure}[htbp]
\centerline{\includegraphics[width=0.4\textwidth]{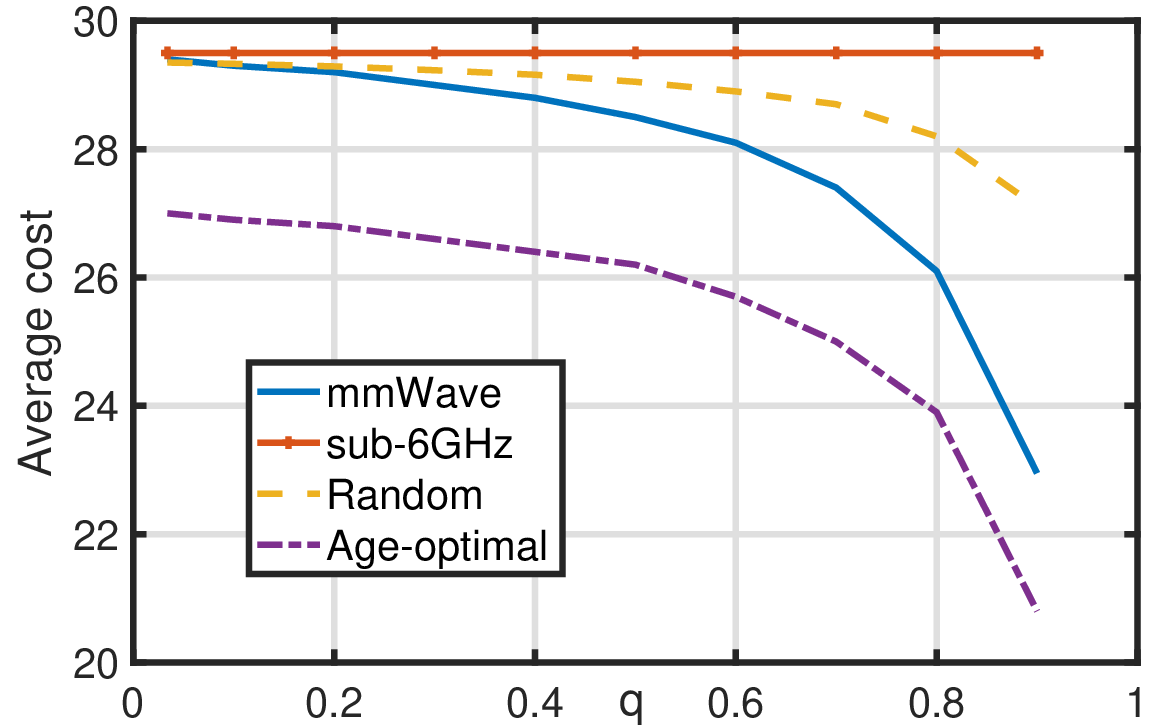}}
\caption{Time-average expected age vs. the parameter $q$ of the mmWave channel, where $d=20$ and $p=0.966$.}
\label{fig4}
\end{figure}
\begin{proof}
By \eqref{def1} and $1-q = p$, we have: (i) $F(p,q,d) = H(p,q,d)$, and (ii) $F(p,q,d)> 0$ is equivalent to $G(p,q,d)> 0$. From the two above results and the definition of  $ \textbf{B}_1,\ldots, \textbf{B}_4$ in \eqref{b1-b4}, we directly get \eqref{textb1iid} and \eqref{textb3iid}. Moreover, the definitions of \eqref{textb1iid} and \eqref{textb3iid} imply that $(p,q,d)\in \textbf{B}_1$ or $(p,q,d)\in \textbf{B}_3$. 
\end{proof}
From Theorem \ref{theorem2a}, if $(p,q,d)\in \textbf{B}_1$, then the optimal policy is always choosing Channel $1$. From Theorem \ref{theorem2c}, if $(p,q,d)\in \textbf{B}_3$, then the optimal policy chooses one of the five cases that are depicted in \eqref{theorem2cc}. However, we can reduce the five cases to two cases:  
If Channel $1$ is i.i.d., then the state information of Channel $1$ is not useful. Thus, $\lambda_0^* = \lambda_1^*$. Note that from Definition \ref{defthm}, we have $s_2(\beta)\le d$ and $s_i(\beta)\ge d+1$ for $i\in \{1,3,4\}$. Thus, only the first case and the last case in \eqref{theorem2cc} \textcolor{black}{can possibly appear} for i.i.d. channel. 

So in i.i.d. case, Theorem \ref{theorem2a} and Theorem \ref{theorem2c} reduce to the following:

%Suppose $\bar{\Delta}_{\text{opt}}$ is the optimal average age of Problem \eqref{avg} and $\lambda^*_0$ is the optimal threshold

\begin{corollary}
Suppose that $p+q=1$, i.e., Channel $1$ is i.i.d., then

(a) If $1-p\ge 1/d$, then the optimal policy is always choosing Channel $1$.
In this case, the optimal objective value of \eqref{avg} is $\bar{\Delta}_{\text{opt}} = 1/(1-p)$. 

(b) If $1-p<1/d$, then the optimal policy is non-decreasing in age and the optimal thresholds $\lambda^*_0=\lambda^*_1$. The threshold $\lambda^*_0$ may take multiple values, given by
\begin{equation}
   \left\{
\begin{array}{lll}
 \lambda^*_0 = s_{1}(\beta_{1})  & \text{if }~ \bar{\Delta}_{\text{opt}} = \beta_{1}, \\
  \lambda^*_0\in \{ 1,2,\ldots,d \}  & \text{if }~ \bar{\Delta}_{\text{opt}} = (3/2)d-1/2,

\end{array}
\right. 
\end{equation}
 \begin{figure}[htbp]
\centerline{\includegraphics[width=0.4\textwidth]{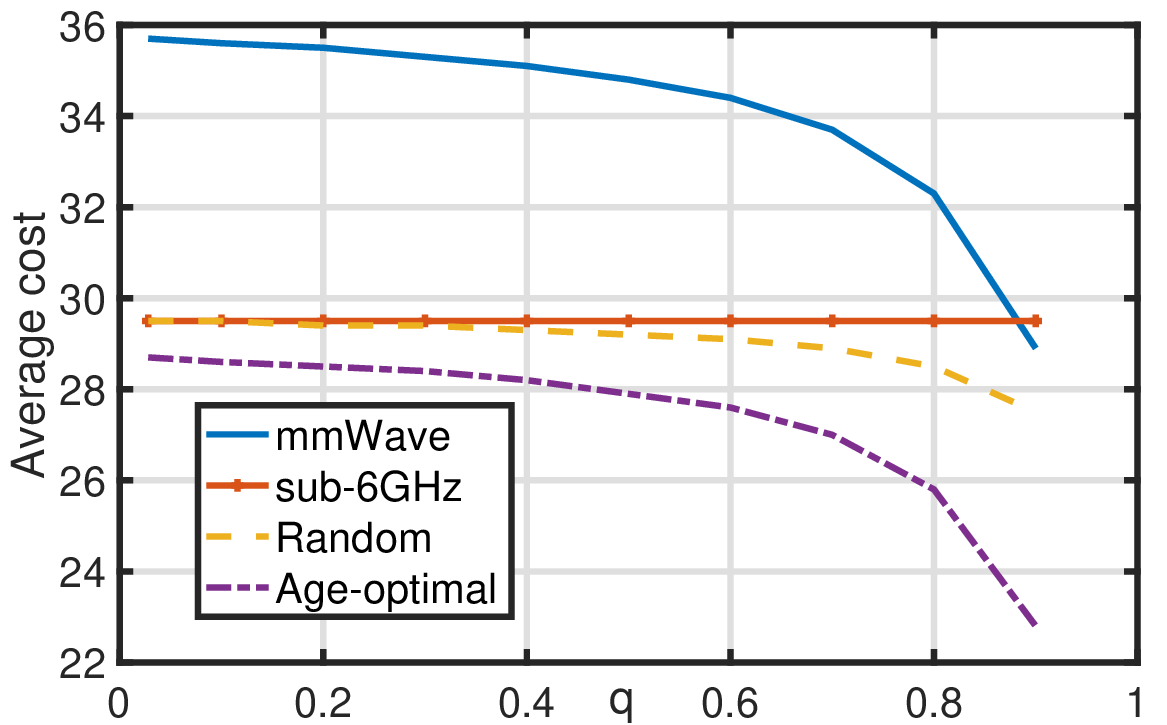}}
\caption{Time-average expected age vs. the parameter $q$ of the mmWave channel, where $d=20$ and $p=0.972$.}
\label{fig4mod}
\end{figure} 
\begin{figure}[htbp]
\centerline{\includegraphics[width=0.4\textwidth]{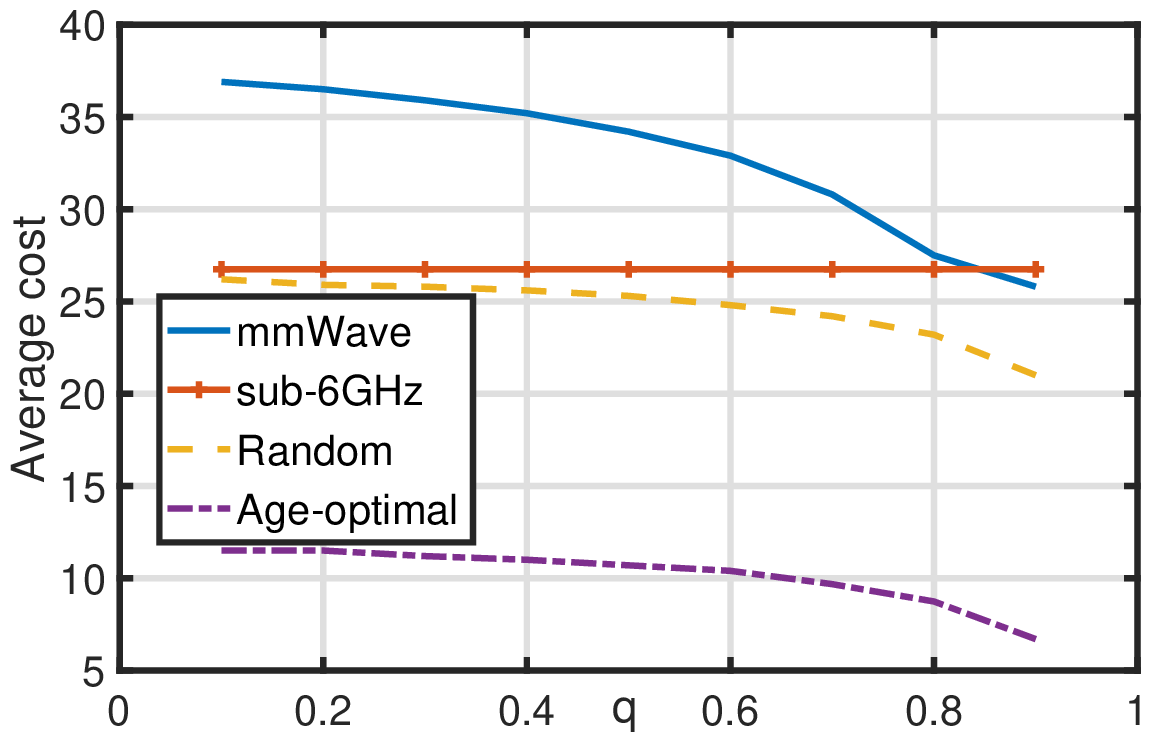}}
\caption{Time-average expected age penalty vs. the parameter $q$ of the mmWave channel, where $p = 0.9$, $d=20$, and the age penalty function is $f(\Delta)=(\frac{1}{p-0.003})^\Delta$.}
\label{fig6}
\vspace{-3pt}
\end{figure}
  $\bar{\Delta}_{\text{opt}}$ is the optimal objective value of \eqref{avg}, determined by
\begin{equation}
    \bar{\Delta}_{\text{opt}}= \min \Big{\{ }  \beta_{1},\frac{3}{2}d-\frac{1}{2} \Big{\} }.
\end{equation}
%If $\bar{\Delta}_{\text{opt}} = 3/2d+1$, then $\lambda_0\in \{ 1,2,\ldots,d+1 \} $. 
%If $\bar{\Delta}_{\text{opt}} = \bar{\Delta}_{11}$, then $\lambda_0 = s_{11}(\bar{\Delta}_{11})$ 
%\begin{equation*}\lambda_0=\underset{d+1<s}{\mathrm{argmin}}\ \frac{f_{11}(s)}{g_{11}(s)}. \end{equation*} 
\label{cor-thm3(2)}
\end{corollary}

%\begin{proof}
%See
%Appendix \ref{corapp}.
%\end{proof}

%If Channel $1$ is i.i.d., then the state information of Channel $1$ in the previous time slot should not affect the scheduling decision in the current time slot. Thus, we have only one threshold, i.e., $\lambda^*_0 = \lambda^*_1$.
%\begin{figure}[htbp]
%\centerline{\includegraphics[width=0.4\textwidth]{figures/simu.eps}}
%\caption{Thresholds of the optimal scheduling policy for i.i.d. mmWave channel state, where the packet transmission time of the sub-6GHz channel is $d= 10,20,50$.}
%\label{fig8}
%\end{figure}
%\begin{figure}[htbp]
%\centerline{\includegraphics[width=0.4\textwidth]{figures/simu1.eps}}
%\caption{Time-average expected age vs. the parameter $q$ of the mmWave channel, where $d=20$.}
%\label{fig4}
%\end{figure} 

%Note the value $1/(1-p)$ mentioned in Corollary \ref{cor-thm3(2)}(a) is the simplification of the value in \eqref{always1} by taking $q=1-p$. Since Channel $1$ is i.i.d., $1-p$ is the probability of successful transmission (transmission rate) of Channel $1$. 

%Corollary \ref{cor-thm3(2)} also reveals the relation between age-optimal and delay-optimal policies. Note that in our model when Channel $1$ is i.i.d., the delay-optimal policy always chooses the channel with the lower average delay. 
Corollary \ref{cor-thm3(2)}(a) suggests that if the transmission rate of Channel $1$ is larger than the rate of Channel $2$ (which is $1/d$), then the age-optimal policy always chooses Channel $1$. %which is the same with the delay-optimal policy.
Corollary \ref{cor-thm3(2)}(b) implies that if the transmission rate of Channel $1$ is smaller than the rate of Channel $2$, then the age-optimal policy is non-decreasing threshold-type on age. %At the same time, the delay-optimal policy always chooses Channel $2$. Thus, age-optimality does not imply delay-optimality.  

%% file: sections/numerical.tex
\section{Numerical Results}

%\vspace{-12pt}

\textcolor{black}{We first provide the optimal threshold $\lambda^*_0$ with the change of $p$ for $d=10,20,50$, respectively, where $\lambda^*_0$ is the optimal threshold in i.i.d. channel described in Corollary \ref{cor-thm3(2)}}. From Fig. \ref{fig8}, the optimal threshold diverges to boundary $p^*=0.9,0.95,0.98$ respectively. As $p$ enlarges, the mmWave channel has worse connectivity, thus the thresholds goes down and converges to always choosing the sub-6GHz channel. %Note if thresholds goes down to below $d$, then choosing Channel $2$ is optimal. 

Then we compare our optimal scheduling policy (called \emph{Age-optimal}) with three other policies, including (i) always choosing the mmWave channel (called \emph{mmWave}), (ii) always choosing the sub-6GHz channel (called \emph{sub-6GHz}), and (iii) randomly choosing the mmWave and sub-6GHz channels with equal probability (called \emph{Random}). We provide the performance of these policies for different $q$ in Fig. \ref{fig4} and Fig.~\ref{fig4mod}. %Note that $p$ is slightly increasingly configured to allow \emph{mmWave} and \emph{sub-6GHz} policies have a similar average cost. The reason is that a wide throughput disparity between two channels leads to little benefit from the optimal policy. In Fig. \ref{fig4}, our optimal policy outperforms other policies, and the benefit enlarges if Channel $1$ becomes positively correlated ($q$ is larger).
Our optimal policy outperforms other policies. If the two channels have a similar age performance, the benefit of the optimal policy enlarges as the mmWave channel becomes positively correlated ($q$ is larger). If the two channels have a large age performance disparity, the optimal policy is close to always choosing a single channel, and thus the benefit is obviously low.
Although our theoretical results consider linear age, we also provide numerical results when the cost function is nonlinear on age by using value iteration \cite{puterman1990markov}. %In Fig. \ref{fig5}, we suppose a cubic age cost. 
%The gain generally becomes larger compared with linear age. 
For exponential age in Fig. \ref{fig6}, the gain is significantly large for all $q$: other policies have more than $2$ times of average cost than the optimal policy. The numerical simulation indicates the importance of exploring optimal policy for nonlinear age cost function, which is our future research direction.

%\else

%\fi

%% file: sections/conclusion.tex
\section{conclusion}

%This study focuses on minimizing age in two heterogeneous channels with delayed channel state information. We show that for any choice of system parameters and the previous state of Channel $1$, there exists an optimal policy that is threshold type on age vale. Then we show that the optimal threshold type policy either is a constant policy or can be solved by a low complexity algorithm, depending on the choice of system parameters and the previous state of Channel $1$. Finally, we provide a simple solution of the optimal policy in a special case when Channel $1$ is i.i.d and the previous state of Channel $1$ is negligible. 

In this paper, we have studied age-optimal transmission scheduling for hybrid mmWave/sub-6GHz channels. For all possibly values of the channel parameters and the ON-OFF state of the mmWave channel, the optimal scheduling policy have been proven to be of a threshold-type on the age. Low complexity algorithms have been developed for finding the optimal scheduling policy. %Our solution to the optimal policy is simple and has much lower complexity than the state of art. %Finally, we have provided a simple solution of the optimal policy in a special case when Channel 1 is i.i.d, i.e., the previous state of Channel 1 is negligible. 
Finally, our numerical results show that the optimal policy can reduce age compared with other policies. %Our study provides a guideline for how to efficiently optimize data freshness in hybrid mmWave/sub-6GHz channels. 

%% file: sections/proof.tex
\section{Appendices: Proofs of Main Results}
%\section{Appendices}
\label{mainproof}

In this section, we prove our main results: Theorem \ref{theorem2} (Section~\ref{proof2}) and Theorem \ref{theorem2a}---\ref{theorem2d} (Section~\ref{proof3}). In Section~\ref{discountedA}, we describe a discounted problem that helps to solve average problem \eqref{avg}. In Section~\ref{proof2}, we introduce Proposition~\ref{lemma2} which plays an important role in proving Theorem \ref{theorem2}. Section~\ref{proof3} provides the proofs of Theorem \ref{theorem2a}---\ref{theorem2d}. %we firstly show that the optimal policy belongs to a restricted action space where one always transmit the packet if the server is idle. Then  %Section \ref{proof3bc} provides the proof of Theorem \ref{theorem3} (b) and (c) in Section \ref{mainresults}. The difference between Section \ref{proof3} and \ref{proof3bc} is that Section \ref{proof3} utilizes similar technique with Section \ref{proof2} and shows that always choosing Channel $1$ is optimal if $(p,q,d)\in \textbf{B}_1 \cup \textbf{B}_4$, while Section \ref{proof3bc} provides the exact solution to computing optimal age and thresholds, since in general no constant policy can achieve optimality.        

\subsection{Preliminaries}\label{discountedA}

To solve Problem \eqref{avg}, we introduce a discounted problem below. The objective is to solve the discounted sum of expected cost given an initial state \textbf{s}:
\begin{equation}
J^\alpha (\textbf{s}) = \inf_{\pi\in\Pi'} \lim_{T\rightarrow \infty} \sum_{t=0}^{T} \mathbb{E}[\alpha^{t} \Delta^\pi(t)|\textbf{s}(0)=\textbf{s}], 
\label{disc}
\end{equation} where $\alpha\in (0,1)$ is the discount factor. We call $J^\alpha (\textbf{s})$ the \emph{value function} given the initial state \textbf{s}. \textcolor{black}{Recall that we use \textbf{s}$=(\delta,l_1,l_2)$ to denote the system state, where $\delta$ is the age value and $l_1,l_2$ are the state of Channel $1$ and Channel $2$. From Lemma \ref{zerowait}, we only need to consider $\pi\in \Pi'$ instead of $\pi\in \Pi$. }

\textcolor{black}{The value function $J^\alpha (\textbf{s})$ satisfies a following property:}
%Suppose $J^\alpha (\cdot):\textbf{S}\rightarrow \mathbb{R}$ is the optimal objective value of \eqref{disc}, 
\begin{lemma}
For any given $\alpha$ and \textbf{s}, $J^\alpha(\textbf{s})<\infty$. 
\label{finite}
\end{lemma}
\begin{proof}
\ifreport
See Appendix \ref{app1}.
\else
See our technical report \cite{pan2020age}.
\fi
\end{proof}

A policy $\pi$ is deterministic stationary if \textcolor{black}{$\pi(t)=Z(\textbf{s}(t))$ at any time $t$, where $Z: \textbf{S}\rightarrow \Pi'$ is a deterministic function.} According to \cite{sennott1989average}, and Lemma \ref{finite}, there is a direct result for Problem \eqref{disc}:

\begin{lemma}
(a) The value function $J^\alpha(\textbf{s})$ satisfies the Bellman equation
\begin{equation}
\begin{split}
   Q^\alpha(\textbf{s},u) & \triangleq \delta + \alpha \sum_{\textbf{s}'\in \textbf{S}}P_{\textbf{s}\textbf{s}'}(u)J^\alpha(\textbf{s}'),\\
   J^\alpha(\textbf{s}) & = \min_{u\in \Pi'} Q^\alpha(\textbf{s},u). \end{split}
\label{discounted-lemma1}
\end{equation}

(b) There exists a deterministic stationary policy ${\mu}^{\alpha,*}$ that satisfies Bellman equation \eqref{discounted-lemma1}. The policy ${\mu}^{\alpha,*}$ solves Problem \eqref{disc} for all initial state \textbf{s}. 

(c) Assume that $J^\alpha_0(\textbf{s})=0$ for all \textbf{s}. For $n\ge 1$, $J^\alpha_n$ %$(\cdot):\textbf{S}\rightarrow \mathbb{R}$ 
is defined as %the following value iteration  
\begin{equation} \begin{split} 
   Q^\alpha_{n}(\textbf{s},u) & \triangleq  \delta +\alpha\sum_{\textbf{s}'\in\textbf{S}} P_{\textbf{s}\textbf{s}'}(u) J^\alpha_{n-1}(\textbf{s}'),\\
   J^\alpha_{n}(\textbf{s}) & = \min_{u\in \Pi'} Q^\alpha_{n}(\textbf{s},u),  \label{value-iterate} \end{split}
\end{equation}
then $\lim_{n\rightarrow \infty} J^\alpha_n(\textbf{s})=J^\alpha(\textbf{s})$ for every \textbf{s}.
\label{lemma1}
\end{lemma}

Also, since the cost function is linearly increasing in age, utilizing Lemma \ref{lemma1}(c), we also have  \begin{lemma}\label{monotonicity}
For all given $l_1$ and $l_2$, $J^\alpha (\delta,l_1,l_2)$ is increasing in $\delta$. 
\end{lemma}
\begin{proof}
\ifreport
See Appendix \ref{monotonicityapp}.
\else
See our technical report \cite{pan2020age}.
\fi
 \end{proof}

Since Problem \eqref{disc} satisfies the properties in Lemma \ref{lemma1}, utilizing Lemma \ref{lemma1} and Lemma \ref{monotonicity}, the following Lemma gives the connection between Problem \eqref{avg} and Problem \eqref{disc}. 

\begin{lemma} 
(a) There exists a stationary deterministic policy that is optimal for Problem \eqref{avg}. 

(b) There exists a value $J^*$ for all initial state $\textbf{s}$ such that    
\begin{equation*}
    \lim_{\alpha\rightarrow 1^-}(1-\alpha)J^\alpha (\textbf{s}) = J^*.
\end{equation*} 
Moreover, $J^*$ is the optimal average cost for Problem \eqref{avg}.

(c) For any sequence $(\alpha_n)_n$ of discount factors that converges to $1$, there exists a subsequence $(\beta_n)_n$ such that $\lim_{n\rightarrow \infty}{\mu}^{\beta_n,*} = {\mu^*}$. Also, ${\mu^*}$ is the optimal policy for Problem \ref{avg}. 
\label{theorem1}
\end{lemma}

\begin{proof}
\ifreport
See Appendix \ref{proof1}.
\else
See our technical report \cite{pan2020age}.
\fi
\end{proof}

Lemma \ref{theorem1} provides the fact that: We can solve Problem \eqref{disc} to achieve Problem \eqref{avg}. The reason is that the optimal policy of Problem \eqref{disc} converges to the optimal policy of Problem \eqref{avg} in a limiting scenario (as $\alpha\rightarrow 1$).   

%Lemma \ref{theorem1} tells there exists an optimal $\mu^*: \bm{S} \rightarrow \{0,1,2\}$.

%Note that if the state of Channel $2$ is $0$, the policy is idle. So solving $\mu^*$ is equivalent to analysing the structure of binary $\mu^*(\delta,l_1,0)$ with $\delta=1,2,...$ and $l_1 = 0,1$. 

\subsection{Proof of Theorem \ref{theorem2} }\label{proof2}

We begin with providing an optimal structural result of discounted policy $\mu^{\alpha,*}$. Then, we achieve the average optimal policy ${\mu^*}$ by letting $\alpha\rightarrow 1$.

\begin{definition}
For any discount factor $\alpha \in (0,1)$, the channel parameters $p,q\in (0,1)$ and $d\in \{2,3,...\}$, \textcolor{black}{we define}
\begin{equation}
\begin{split}
    \textbf{B}_1(\alpha) & =\{ (p,q,d): F(p,q,d,\alpha)\le 0, H(p,q,d,\alpha)\le 0 \}, \\     
    \textbf{B}_2(\alpha) & =\{ (p,q,d): F(p,q,d,\alpha)> 0, G(p,q,d,\alpha)\le 0 \}, \\
    \textbf{B}_3(\alpha) & =\{(p,q,d): F(p,q,d,\alpha)> 0, G(p,q,d,\alpha)> 0 \}, \\
    \textbf{B}_4(\alpha) & =\{(p,q,d): F(p,q,d,\alpha)\le 0, H(p,q,d,\alpha)> 0 \}, 
    \end{split}\end{equation}
where functions $F(\cdot),G(\cdot),H(\cdot):\Theta\times(0,1)\rightarrow \mathbb{R}$ are defined as:  \begin{equation}\begin{split}
    F(p,q,d,\alpha)&=\sum_{i=0}^{\infty} (\alpha p)^i-\sum_{i=0}^{d-1}\alpha^i, \\
    %G(p,q,d,\alpha)&=1-(1-\alpha(1-q))\sum_{i=0}^{d}\alpha^i,\\
    G(p,q,d,\alpha)&=1+\alpha(1-q)\sum_{i=0}^{d-1}\alpha^i-\sum_{i=0}^{d-1}\alpha^i,\\
    H(p,q,d,\alpha)&=1+\alpha(1-q)\sum_{i=0}^{\infty}(\alpha p)^i-\sum_{i=0}^{d-1}\alpha^i. 
    \end{split}\end{equation}%\label{def1}
\end{definition}
 Observe that all four regions $\textbf{B}_i(\alpha)$ converge to $\textbf{B}_i$ as the discount factor $\alpha\rightarrow 1$, where the regions $\textbf{B}_i$ are described in Definition \ref{def1}.  

The optimal structural result of Problem \eqref{disc} with a discount factor $\alpha$ is provided in the following proposition:

\begin{proposition}

There exists a threshold type policy $\mu^{\alpha,*}(\delta,l_1,0)$ on age $\delta$ that is the solution to Problem \eqref{disc} such that:

(a) If $l_1 = 0$ and $(p,q,d)\in \textbf{B}_1(\alpha)\cup \textbf{B}_4(\alpha)$, then $\mu^{\alpha,*}(\delta,l_1,0)$ is non-increasing in the age $\delta$. 

(b) If $l_1 = 0$ and $(p,q,d)\in \textbf{B}_2(\alpha)\cup \textbf{B}_3(\alpha)$, then $\mu^{\alpha,*}(\delta,l_1,0)$ is non-decreasing in the age $\delta$. 

(c) If $l_1 = 1$ and $(p,q,d)\in \textbf{B}_1(\alpha)\cup \textbf{B}_2(\alpha)$, then $\mu^{\alpha,*}(\delta,l_1,0)$ is non-increasing in the age $\delta$. 

(d) If $l_1 = 1$ and $(p,q,d)\in \textbf{B}_3(\alpha)\cup \textbf{B}_4(\alpha)$, then $\mu^{\alpha,*}(\delta,l_1,0)$ is non-decreasing in the age $\delta$. 
%There exists a threshold type policy $\mu^\alpha(\delta,l_1,0)$ on age $\delta$ that is the solution to Problem \eqref{disc}. The detailed structure type is as follows in (a) and (b):

%(a) The policy $\mu^\alpha(\delta,0,0)$ is non-increasing in $\delta$ if $(p,q,d)\in \textbf{B}_1(\alpha)\cup \textbf{B}_4(\alpha)$, non-decreasing if $(p,q,d)\in \textbf{B}_2(\alpha)\cup \textbf{B}_3(\alpha)$.

%(b) The policy $\mu^\alpha(\delta,1,0)$ is non-increasing in $\delta$ if $(p,q,d)\in \textbf{B}_1(\alpha)\cup \textbf{B}_2(\alpha)$, non-decreasing if $(p,q,d)\in \textbf{B}_3(\alpha)\cup \textbf{B}_4(\alpha)$.
\label{lemma2}
\end{proposition}
\textcolor{black}{Note that Theorem \ref{theorem2} can be immediately shown from Proposition \ref{lemma2}, Lemma \ref{theorem1} and the convergence of the regions $\textbf{B}_i(\alpha)$ to $\textbf{B}_i$ (for $i=1,2,3,4$) as $\alpha\rightarrow 1$.} \textcolor{black}{\emph{The rest of Section \ref{proof2} provides the proof for Proposition \ref{lemma2}}}.
%\begin{proof}
% Please see Section \ref{proof2} for details. 
%\end{proof} 

%So Theorem \ref{theorem2} can be immediately shown from Proposition \ref{lemma2}, Lemma \ref{theorem1} and the convergence of the regions $\textbf{B}_i(\alpha)$ to $\textbf{B}_i$ (for $i=1,2,3,4$) as $\alpha\rightarrow 1$. Thus, the remainder of Section \ref{proof2} provides the main proof of Proposition \ref{lemma2}. This is the start of the proof of Proposition \ref{lemma2}. 

%In the remaining sections, for ease of description, we omit the subscription $\alpha$ in $J^\alpha (\cdot)$, $J^\alpha_n (\cdot)$, $Q^\alpha (\cdot)$, $Q^\alpha_n (\cdot)$, $\mu^\alpha (\cdot)$ and $\mu^\alpha_n (\cdot)$.

Since Channel $1$ and Channel $2$ have different delays, we are not able to show that the optimal policy is threshold type by directly observing the Bellman equation like \cite{altman2019forever}. Thus, we will use the concept of \emph{super-modularity} \cite[Theorem 2.8.2]{topkis1998supermodularity}. The domain of age set and decision set in the Q-function is $\{1,2,...\}\times \{1,2\}$, which is a \emph{lattice}. Given a positive $s$, the subset $\{s,s+1,...\}\times \{1,2\}$ is a \emph{sublattice} of $\{1,2,...\}\times \{1,2\}$. Thus, if the following holds for all $\delta>s$: 
\begin{equation}\begin{split}
& Q^\alpha (\delta,l_1,0,1)-Q^\alpha(\delta-1,l_1,0,1) \\ \le  & Q^\alpha(\delta,l_1,0,2)-Q^\alpha(\delta-1,l_1,0,2), \end{split} \label{modular}  
\end{equation}
then the Q-function $Q^\alpha(\delta,l_1,0,u)$ is super-modular in $(\delta,u)$ for $\delta>s$, which means the optimal decision \begin{equation}
    \mu^{\alpha,*}(\delta,l_1,0) =\mathrm{argmin}_{u\in \{1,2\}}Q^\alpha(\delta,l_1,0,u)
\end{equation} is non-increasing in $\delta$ for $\delta\ge s$. If the inequality of \eqref{modular} is inversed, then we call $Q^\alpha(\delta,l_1,0)$ is sub-modular in $(\delta,u)$ for  $\delta>s$, and $\mu^{\alpha,*}(\delta,l_1,0)$ is non-decreasing in $\delta$ for $\delta\ge s$.

%(\text{or}\le)

For ease of notations, we give Definition \ref{def2}:

\begin{definition} Given $l_1\in \{0,1\}$, $u\in \{1,2\}$, %$n\ge 0$ 
\begin{equation}
L^\alpha(\delta,l_1,u)\triangleq Q^\alpha(\delta,l_1,0,u)-Q^\alpha(\delta-1,l_1,0,u).
\end{equation}  
\label{def2}
\end{definition} \vspace{-8pt} %L_n(\delta,l_1,u)\triangleq Q_n(\delta,l_1,0,u)-Q_n(\delta-1,l_1,0,u)
Note that $L^\alpha(\delta,l_1,1)$ is the left hand side of \eqref{modular}, and $L^\alpha(\delta,l_1,2)$ is the right hand side of \eqref{modular}. %Also, $L(\delta,l_1,u)$ is also called the age difference function under the action $u$. Thus, the Q-function $Q(\delta,l_1,0,u)$ is supermodluar (or submodular) in $(\delta,u)$ for $\delta>s$, if $L(\delta,l_1,1)\le($or$ \ge) L(\delta,l_1,2)$ for $\delta>s$. 

%However, because of the mismatch of delays in our problem, most of the well-known techniques to show supermodularity (e.g., \cite{puterman1990markov},\cite{ngo2009optimality},\cite{krishnamurthy2016partially} etc) do not apply in our setting. 
%Thus, we need a new approach to show the super-modularity. 
\textcolor{black}{Our high-level idea to show Proposition \ref{lemma2} is as follows:} First, we show that $L^\alpha(\delta,l_1,2)$ is a constant (see Lemma \ref{g2m} below), then we compare $L^\alpha(\delta,l_1,1)$ with the constant to check super-modularity (\textcolor{black}{see the proofs of Lemma \ref{l0m} and Lemma \ref{l1m} below}). %The proof of comparing $L^\alpha(\delta,l_1,1)$ with the constant is relegated to the proofs of Lemma \ref{l0m} and Lemma \ref{l1m}.} 
 %The right hand side $L(\delta,l_1,2)$ has the following nice property, which is a constant:
 
Suppose that $m\triangleq \sum_{i=0}^{d-1} \alpha^i$, and we have:
 \begin{lemma}\label{g2m}
 For all $\delta\ge 2$ and $l_1\in \{0,1 \}$, $L^\alpha(\delta,l_1,2)=m$. 
% \begin{equation}\begin{split}
%   L(\delta,l_1,2)= m.\end{split}
% \end{equation}
 \end{lemma}
 
 \begin{proof}
 \ifreport
See Appendix \ref{g2mapp}.
\else
See our technical report \cite{pan2020age}.
\fi
 \end{proof}
% Lemma \ref{g2m} tells us that the right hand side function $L(\delta,l_1,2)$, is the discounted sum of constant $1$ and the length is the delay of Channel $2$. %Note that Lemma \ref{g2m} only holds if the cost function is linear in age $\delta$.
 
%In the following sections we may use $m$ or $\sum_{i=0}^{d} \alpha^i$ interchangeably.
%Remember that $F(p,q,d,\alpha)=\sum_{i=0}^{\infty} (\alpha p)^i-m$. 
%From Lemma \ref{g2m}, judging $L(\delta,l_1,1)\le( \text{or} \ge) L(\delta,l_1,2)$ turns into whether $L(\delta,l_1,1)\le (\text{or}\ge)$  $m$.

%Also, the constant $m$ has the following property: 

%Thus, benefiting from Lemma \ref{g2m}, we next show the relation between $L(\delta,0,1)$ and the constant $m$.
Also, we have
\begin{lemma}\label{l0m}
(a) If $l_1=0$ and $(p,q,d)\in \textbf{B}_1(\alpha)\cup \textbf{B}_4(\alpha)$, then $Q^\alpha(\delta,l_1,0,u)$ is super-modular in $(\delta,u)$ for $\delta\ge 2$.

(b) If $l_1=0$ and $(p,q,d)\in \textbf{B}_2(\alpha)\cup \textbf{B}_3(\alpha)$, then $Q^\alpha(\delta,l_1,0,u)$ is sub-modular in $(\delta,u)$ for $\delta\ge 2$.
\end{lemma}

\begin{proof}
\ifreport
See Appendix \ref{l0mapp}.
\else
See our technical report \cite{pan2020age}.
\fi
\end{proof}

%From Lemma \ref{g2m} and Lemma \ref{l0m}, we directly find that $L(\delta,0,1)\le L(\delta,0,2)$ if the system parameters $(p,q,d)\in \textbf{B}_1(\alpha)\cup \textbf{B}_4(\alpha)$. Also, $L(\delta,0,1)> L(\delta,0,2)$ if the system parameters $(p,q,d)\in \textbf{B}_2(\alpha)\cup \textbf{B}_3(\alpha)$.   
%Thus, the Q-function $Q(\delta,0,0,u)$ is supermodular in $(\delta,u)$ and

Lemma \ref{l0m}(a) implies that $\mu^{\alpha,*}(\delta,0,0)$ is non-increasing in $\delta$ if $(p,q,d)\in \textbf{B}_1(\alpha)\cup \textbf{B}_4(\alpha)$. Lemma \ref{l0m}(b) implies that $\mu^{\alpha,*}(\delta,0,0)$ is non-decreasing in $\delta$ if $(p,q,d)\in \textbf{B}_2(\alpha)\cup \textbf{B}_3(\alpha)$.
%Also, the Q-function $Q(\delta,0,0,u)$ is submodular in $(\delta,u)$ and $\mu(\delta,0,0)$ is non-decreasing threshold type if the system parameters $(p,q,d)\in \textbf{B}_2(\alpha)\cup \textbf{B}_3(\alpha)$.  
Thus, Proposition \ref{lemma2}(a),(b) hold.  

Lemma \ref{l0m} gives the result when the previous state of Channel $1$ is $0$. We then need to solve when the previous state of Channel $1$ is $1$. Different from $Q^\alpha(\delta,0,0,u)$, the Q-function $Q^\alpha(\delta,1,0,u)$ does not satisfy super-modular (or sub-modular) in ($\delta,u$) for all the age value $\delta$. Thus, we give a weakened condition: we can find out a value $s$, such that the Q-function $Q^\alpha(\delta,1,0,u)$ is super-modular (or sub-modular) for \emph{a partial age set} $s,s+1,...$ and $\mu^{\alpha,*}(\delta,1,0)$ is a constant on the set $1,2,...,s$. Then, $\mu^{\alpha,*}(\delta,1,0)$ is still non-increasing (or non-decreasing). Note that super-/sub-modularity is the sufficient but not necessary condition to the monotonicity of $\mu^{\alpha,*}(\delta,l_1,0)$ in $\delta$.     

Thus, to solve Proposition \ref{lemma2}(c),(d), we provide the following lemma:   
%Thus, according to \eqref{A10} and \eqref{A1022}, the focus is comparing $J(\delta+1,0,0)-J(\delta,0,0)$ with $m$.
%We firstly provide a Lemma: 

%Lemma \ref{11d110} helps for introducing a property for optimality of switch type policy: 

%Now we start proving (b)

%Lemma \ref{lemma-switch} is useful for solving the following lemma which is the sufficient condition for theorem \ref{lemma2} (b):

\begin{lemma}\label{l1m}
(a) If $l_1=1$ and $(p,q,d)\in \textbf{B}_1(\alpha)\cup \textbf{B}_2(\alpha)$, 
then there exists a positive integer $s$, such that $Q^\alpha(\delta,l_1,0,u)$ is super-modular in $(\delta,u)$ for $\delta > s$, and $\mu^{\alpha,*}(\delta,l_1,0)$ is always $1$ or always $2$ for all $\delta \le s$. %Thus, $\mu(\delta,1,0)$ is nonincreasing.  

(b) If $l_1=1$ and $(p,q,d)\in \textbf{B}_3(\alpha)\cup  \textbf{B}_4(\alpha)$, then there exists a positive integer $s$, such that $Q^\alpha(\delta,l_1,0,u)$ is sub-modular in $(\delta,u)$ for $\delta > s$, and $\mu^{\alpha,*}(\delta,l_1,0)$ is always $1$ or always $2$ for all $\delta \le s$.

%either of the two cases holds: ($1$) The policy $\mu(\delta,1,0)$ is always $1$ or always $2$. ($2$) there exists a positive integer $s$, such that the left hand side $L(\delta,1,1) > m$ for all $\delta > s$ and $\mu(\delta,1,0)$ is always $1$ or always $2$ for all $\delta \le s$.

%Thus, $\mu(1,1,0)$ is nondecreasing. 

%(c) If $(p,q,d)\in \textbf{B}_2(\alpha)$, then there exists $s\ge 0$, such that $L(\delta,1,1) \le L(\delta,1,2)$ for all $\delta > s$ and $\mu(\delta,1,0)=1$ for all $\delta \le s$. Thus $\mu(1,1,0)$ is nonincreasing. 

%(d) If $(p,q,d)\in \textbf{B}_3(\alpha)$, then there exists $s\ge 0$, such that $L(\delta,1,1) > L(\delta,1,2)$ for all $\delta > s$ and $\mu(\delta,1,0)=1$ for all $\delta \le s$. Thus $\mu(1,1,0)$ is nondecreasing. 
\end{lemma}

\begin{proof}
\ifreport
See Appendix \ref{l1mapp}.
\else
See our technical report \cite{pan2020age}.
\fi
\end{proof}

%There is one difference between Lemma \ref{l1m} and Lemma \ref{l0m}. Lemma \ref{l1m} does not provide a strong result e.g., the left hand side $L(\delta,1,1)\le$ or $>m$ for all the age $\delta>1$. Compared with $L(\delta,0,1)$, the left hand side $L(\delta,1,1)$ may not satisfy this condition. For example, take the system parameters $(p,q,d)=(0.8,0.3,2)$ and the discounted factor $\alpha=0.9$. Despite $\mu(\delta,1,0)=1$ is constant on age $\delta$, the difference $L(\delta,1,1)-m$ is not all positive on age $\delta$. However, Lemma \ref{l1m} is sufficient to show the policy $\mu(\delta,1,0)$ is threshold type.  

% We look at the two cases in Lemma \ref{l1m} (a). In case ($1$), the constant policy is a non-increasing policy. In case ($2$), 
 
% Lemma \ref{g2m} implies that the right hand side $L(\delta,1,2)=m$ for all $\delta$. Thus, in Lemma \ref{l1m}, we get $L(\delta,1,1)\le L(\delta,1,2)$ for age $\delta>s$. 
 
 Lemma \ref{l1m}(a) implies that  $\mu^{\alpha,*}(\delta,1,0)$ is non-increasing for $\delta\ge s$ and is constant for for $\delta\le s$. Thus, $\mu^{\alpha,*}(\delta,1,0)$ is non-increasing in $\delta$. Similarly, Lemma \ref{l1m}(b) implies that $\mu^{\alpha,*}(\delta,1,0)$ is non-decreasing for $\delta>0$. Thus, we have shown Proposition \ref{lemma2}(c),(d). \textcolor{black}{Showing the threshold structure of $\mu^{\alpha,*}(\delta,l_1,0)$ even if super-modularity does not hold is one of the key technical contributions in this paper.}   

Overall, Lemma \ref{g2m} and Lemma \ref{l0m} shows Proposition \ref{lemma2}(a),(b). Lemma \ref{g2m} and Lemma \ref{l1m} shows Proposition \ref{lemma2}(c),(d). Thus we have completed the proof of Proposition \ref{lemma2}.
 
To summarize Section \ref{proof2}, Proposition \ref{lemma2}, Lemma \ref{theorem1} and the convergence of  $\textbf{B}_1(\alpha),\ldots, \textbf{B}_4(\alpha)$ to $\textbf{B}_1,\ldots,\textbf{B}_4$ show Theorem \ref{theorem2}. 

\subsection{Proofs of Theorem \ref{theorem2a}---Theorem \ref{theorem2d}}\label{proof3}
\textcolor{black}{In this section, we prove Theorem \ref{theorem2a}$\sim$Theorem \ref{theorem2d} with $(p,q,d)\in \textbf{B}_1$---$(p,q,d)\in \textbf{B}_4$, respectively for efficiently deriving an optimal threshold-type solution.}
\subsubsection{Proof of Theorem \ref{theorem2a}}

For $(p,q,d)\in \textbf{B}_1$, we firstly prove that $\mu^*(\delta,0,0)=1$ and then show that $\mu^*(\delta,1,0)=1$.  

\begin{lemma}\label{mu00in1and4}
If $(p,q,d)\in \textbf{B}_1 \cup \textbf{B}_4$, then the optimal decisions at states $(\delta,0,0)$ for all $\delta$ are $1$.
\end{lemma}

\begin{proof}
See Appendix \ref{mu00in1and4app}.
\end{proof}
In addition, when $l_1=1$, we have the following:

\begin{lemma}\label{mu110for1}
If $(p,q,d)\in \textbf{B}_1 $,  then the optimal decisions at states $(\delta,1,0)$ for all $\delta$ are $1$. 
\end{lemma}
\begin{proof}
See Appendix \ref{mu110for1app}.
\end{proof}

%\textcolor{black}{Theorem \ref{theorem3}(a) follows directly from Lemma \ref{mu00in1and4} and Lemma \ref{mu110for1}. Due to the space limit, the proof of Theorem \ref{theorem3}(b) when $(p,q,d)\in \textbf{B}_2$ is relegated to our technical report \cite{pan2020age}.}

\ifreport
Since $\mu^*(\delta,1,0)$ is non-increasing in the region $\textbf{B}_1$ by Theorem \ref{theorem2}, Lemma \ref{mu110for1} implies that $\mu^*(\delta,1,0)=1$ for all $\delta$. Besides, Lemma \ref{mu00in1and4} implies that $\mu^*(\delta,0,0)=1$ for all $\delta$.
Thus, Theorem \ref{theorem2a} follows directly from Lemma \ref{mu00in1and4} and Lemma \ref{mu110for1}. The optimal policy for $(p,q,d)\in \textbf{B}_1$ is always choosing Channel $1$. % and Lemma \ref{theorem1}. 
 
%Before we further move to the other regions $\textbf{B}_2,\textbf{B}_3,\textbf{B}_4$, we need to emphasize that the optimal decision $\mu^*(1,1,0)$ and $\mu^*(d,1,0)$ is \emph{sufficient} for the optimal decision set $\mu^*(\delta,1,0)$ with $\delta\ge1$. This is because any binary policy cannot reach the states $(2,1,0),...,(d-1,1,0)$,$(d+1,1,0)$,... In addition, 
%Due to the space limit, the proof of Theorem \ref{theorem3} (b) when $(p,q,d)\in \textbf{B}_2$ is relegated to our technical report \cite{pan2020age}.
%We know turn to $(p,q,d)\in \textbf{B}_2$. Due to the space limit, we only provide a part of the proof and relegate the complete proof to our technical report \cite{pan2020age}.
%Because the optimal policy of other case $(p,q,d)\in \textbf{B}_2$ is not constant,  
%While the result of case $(p,q,d)\in \textbf{B}_1$ is simple to describe, the result of cases $(p,q,d)\in \textbf{B}_2$, $\textbf{B}_3$ and $\textbf{B}_4$ are not.   
%This is because for other cases $(p,q,d)\in \textbf{B}_2$, $\textbf{B}_3$ and $\textbf{B}_4$, the optimal decision may not be constant. 

\subsubsection{Proof of Theorem \ref{theorem2b}}

%As discussed in Section \ref{insight-theorem2}, 

\textcolor{black}{In \eqref{help}, we have stated that the MDP problem \eqref{avg} is reduced to deriving the steady-state distributions of the DTMCs. 
Note that Channel $1$ is Markovian $(l_1=0$ or $1)$. When $l_1 = 1$, we observe that only the states $(1,1,0)$ and $(d,1,0)$ can be reached with positive probability for any policy in $\Pi'$. As a result, \eqref{help} can be reduced to a number of the steady-state distributions of the DTMCs with different actions at $(1,1,0)$ and $(d,1,0)$.  
In addition, we observe that the state transition matrices of the DTMCs in \eqref{help} are significantly different depending on the action at $(d,0,0)$. 
Thus, we conclude that there are at most $2^3$ different steady-state distributions of DTMCs based on the actions at three system states: $(1,1,0),(d,1,0)$ with $l_1 = 1$ and $(d,0,0)$ with $l_1 = 0$.} %in regard to the threshold $\lambda_0$ of $\mu^*(\delta,0,0)$. 
 %we give an exact solution based on solving corresponding Markov Chains. 
%Since Channel $1$ is Markovian, there are two non-symmetric states $(\delta,1,0),(\delta,0,0)$ for the same age $\delta$. 
%This complicates the problem: the Markov chain varies based on the restarting states $(1,1,0)$, $(d,1,0)$ and $(d,0,0)$. %This is because any given state may be recurrent in one scenario (the decisions of restarting states) but transient in other scenario. 
%Therefore, we need to enumerate all the feasible cases based on the decisions $(1,1,0)$, $(d,1,0)$ and and $(d,0,0)$. 
Despite that there are totally $2^3$ cases \textcolor{black}{to enumerate}, we manage to reduce to only $4$ cases as in \eqref{discuss-delta1} (for $(p,q,d)\in \textbf{B}_2$). The \textcolor{black}{reason is that the remaining cases are impossible to occur due to the two following restrictions}: (1) the monotonicity is known by Theorem \ref{theorem2}, and (2) the following lemma:
\begin{lemma}\label{lemma-switch-avg}
If Channel $1$ is positive-correlated, i.e., $p+q\ge 1$, and $\mu^{*}(\delta,0,0)=1$, then $\mu^{*}(\delta,1,0)=1$.
Conversely, if Channel $1$ is negative-correlated, i.e. $p+q\le 1$, and $\mu^{*}(\delta,0,0)=2$, then $\mu^{*}(\delta,1,0)=2$.
\end{lemma}
\begin{proof}
See 
Appendix \ref{lemma-switchapp}.
\end{proof}
\else
%our technical report \cite{pan2020age}.
\fi
%are feasible according to 
%\ifreport
%Lemma \ref{lemma-switch} in Appendix %\ref{l1mapp}.
%\else
%our technical report \cite{pan2020age}.
%\fi

\textcolor{black}{Since our optimal policy is of threshold-type,} the action at $(d,0,0)$ is equivalent to whether the threshold of $\mu^*(\delta,0,0)$ is larger or smaller than $d$. Thus, we use $s$ to denote the possible threshold of $\mu^*(\delta,0,0)$.

%Thus, $u(1,1,0)=1$ is sufficient to show the optimal policy is always choosing Channel $1$. And the optimal non-decreasing threshold $\lambda_1^*\in \{ 2,3,...\}$. Thus, we provide the following lemma:  
%\begin{lemma}\label{mu110for4}
%If $(p,q,d)\in \textbf{B}_4 $, then the policy $\mu(1,1,0)=1$ 
%\end{lemma}
%\begin{proof} Please see Appendix \ref{mu110for4app} for details. 
%\end{proof}
%\subsection{Proof of Theorem \ref{theorem3} (b) and (c)}\label{proof3bc}

 For $(p,q,d)\in \textbf{B}_2 $, $\mu^*(\delta,1,0)$ is non-increasing, and $\mu^*(\delta,0,0)$ is non-decreasing. Note that $(p,q,d)\in \textbf{B}_2 $ implies $p+q\ge 1$.
 According to Lemma \ref{lemma-switch-avg}, if $\mu^*(1,1,0)=2$, then $\mu^*(1,0,0)=2$, hence $\mu^*(\delta,0,0)=2$ for all $\delta$. Thus, there are two possible types of DTMCs regarding $\mu^*(d,1,0)=1$ or $\mu^*(d,1,0)=2$.
  If $\mu^*(1,1,0)=1$, then $\mu^*(\delta,1,0)=1$ for all $\delta$, there are thus two possible types of DTMCs regarding the threshold
  \ifreport
   $s>d$ or $s\le d$.  
  \else
  $s>d$ or $s\le d$, where $s$ is the threshold of $\mu^*(\delta,0,0)$. 
  \fi
  %$\lambda_0^*>d$ or $\lambda_0^* \le d$.
 %In Section \ref{C}, we have discussed that the Markov chain varies based on the restarting states $(1,1,0)$, $(d,1,0)$, and $(d,0,0)$. Thus, we conclude that
 Thus, for $(p,q,d)\in \textbf{B}_2 $, there are four possible ways to represent the DTMC diagram of the threshold policy based on the value of the threshold $s$ and the actions at states $(d,l_1,0)$ and $(1,1,0)$
(see
% \ifreport
Appendix \ref{markovapp}
%\else
%our technical report \cite{pan2020age}
%\fi
for the corresponding DTMCs and derivations): 
\begin{itemize}
    \item 
 The threshold $s>d$ and $\mu^*(1,1,0)=\mu^*(d,1,0)=1$ ($\lambda^*_1=1$). \textcolor{black}{Note that we have mentioned $\bar{\Delta}_2(\lambda_0,\lambda_1)$ as the average age of the DTMC with thresholds $(\lambda_0,\lambda_1)$ when $(p,q,d)\in \textbf{B}_2 $. Then, the average age is derived as $\bar{\Delta}_2(s,1)= f_{1}(s)/g_{1}(s)$, which is shown in Appendix \ref{markov1app}}. The functions $f_{1}(s),g_{1}(s)$ are described in 
% \ifreport
Table \ref{fg}.
%\else
%our technical report \cite{pan2020age}.
%\fi
\textcolor{black}{As is shown later, $\beta_1$ described in Definition \ref{defthm} is the minimum of $f_1(s)/g_1(s)$}. 
  \item 
 The threshold $s\le d$ and $\mu^*(1,1,0)=\mu^*(d,1,0)=1$ ($\lambda^*_1=1$). Then the average age is $\bar{\Delta}_2(s,1)=f_{2}(s)/g_{2}(s)$, which is shown in Appendix \ref{markov2app}. The functions $f_{2}(s),g_{2}(s)$ are described in 
% \ifreport
Table \ref{fg}.
%\else
%our technical report \cite{pan2020age}.
%\fi
As is shown later, $\beta_2$ described in Definition \ref{defthm} is the minimum of $f_2(s)/g_2(s)$.
  \item 
 The threshold $s=1$, $\mu^*(1,1,0)=2$ and $\mu^*(d,1,0)=1$ ($\lambda^*_1\in \{2,3,...,d\}$). %The markov chain is described in
% \ifreport
% Appendix \ref{markov3app}.
% \else
% our technical report \cite{pan2020age}.
 %\fi
 The average age is the constant $f_0/g_0$, which is shown in Appendix \ref{markov3app}. Note that $f_0/g_0$ is described in 
% \ifreport
\textcolor{black}{Theorem \ref{theorem2b}}.
%\else
%Table \ref{fg-conference}.
%\fi
  \item 
 The threshold $s=1$ and $\mu^*(1,1,0)=\mu^*(d,1,0)=2$ ($\lambda^*_1\in \{d+1,d+2,...\}$). This policy means that we always choose Channel $2$. So the average age is $(3/2) d-1/2$.
\end{itemize}  
% Recall that $(p,q,d)\in \textbf{B}_2$ implies the Channel $1$ is positive-correlated, i.e., $p+q\ge 1$. Thus, from 
% \ifreport
% Lemma \ref{lemma-switch} in Appendix \ref{l1mapp},
%\else
%Lemma 16 in 
%our technical report \cite{pan2020age},
%\fi
% if $\mu^*(1,1,0)=2$, then $\mu^*(1,0,0)=2$. Since $\mu^*(\delta,0,0)$ is non-decreasing in age $\delta$, so $\mu^*(\delta,0,0)=2$.

\noindent The listed statements illustrated above directly provides the following property:

\begin{proposition}

If $(p,q,d)\in \textbf{B}_2$, then the optimal scheduling policy is 
\begin{align}
    \mu^*(\delta,0,0)=\left\{
\begin{array}{lll}
  1 & \text{if }~ \delta< \lambda_0^*;\\
  2  & \text{if }~ \delta\ge\lambda_0^*, 
\end{array}
\right. \\
    \mu^*(\delta,1,0)=\left\{
\begin{array}{lll}
  2 & \text{if }~ \delta< \lambda_1^*;\\
  1  & \text{if }~ \delta\ge\lambda_1^*, 
\end{array}
\right. 
\end{align} 
where $\lambda_0^*$ and $\lambda_1^*$ are given by  
\begin{equation}\label{proposition2bb}
 \!\!\!\!\!   \left\{
\begin{array}{lll}
  \lambda_0^*=\underset{s\in \{ d+1,\ldots\} }{\argmin} f_1(s)/g_1(s), \lambda_1^*=1 \ \ \   \text{if }~ \bar{\Delta}_{\text{opt}} =  \beta_{1}',\\
  \lambda_0^*=\underset{s\in \{ 1,\ldots, d\} }{\argmin} f_2(s)/g_2(s), \ \lambda_1^*=1 \  \ \   \text{if }~ \bar{\Delta}_{\text{opt}} =  \beta_{2}',\\
 \lambda_0^*=1, \ \ \  \lambda_1^*\in \{2,3,\ldots,d 
\} \ \ \ \   \text{if }~ \bar{\Delta}_{\text{opt}} = f_0/g_0,\\
\lambda_0^*=1, \ \ \  \lambda_1^*\in \{d+1,\ldots 
\} \ \ \ \ \   \text{if }~\bar{\Delta}_{\text{opt}} =  (3/2)d-1/2,
\end{array}
\right. \!\!\!\!\!
\end{equation}
 $\bar{\Delta}_{\text{opt}}$ is the optimal objective value of \eqref{avg}, determined by
\begin{equation}
   \bar{\Delta}_{\text{opt}} =  \min \Big{\{}   \beta_{1}', \beta_{2}', \frac{f_0}{g_0},\frac{3}{2}d-\frac{1}{2} \Big{\} },
\end{equation} 
 $\beta_1'$, $\beta_2'$ are given in \eqref{discuss-beta1}, \eqref{discuss-beta2}, respectively.

\label{propositionb}
\end{proposition}

%  Since $f_i(\lambda_0)$ and $g_i(\lambda_0)$ have a complicated structure, optimizing $f_i(\lambda_0)/g_i(\lambda_0)$ in \eqref{discuss-beta1} and \eqref{discuss-beta2} is challenging. However, 
By using Dinkelbach's method \cite{dinkelbach1967nonlinear}, we can change the minimization problem \eqref{discuss-beta1}, \eqref{discuss-beta2} into a two-layer problem. The inner-layer problem is shown to be unimodal and we derive an exact solution. Thus, we only need a bisection algorithm for the outer-layer, i.e., solving the roots of the equations $h_1(\beta)=0,h_2(\beta)=0$ in \eqref{beta}. \textcolor{black}{To show this}, we introduce the following lemma: 
\begin{lemma}\label{fractional} Suppose that $i\in \{1,2,3,4 \}$.
%If $(p,q,d)\in \textbf{B}_2\cup \textbf{B}_3$, suppose that \begin{equation}\beta'_{i}=\min_{s\in N(i)}\frac{f_{i}(s)}{g_{i}(s)}. \label{algr}\end{equation}
Define
\begin{align}%\label{inlayer}
h'_{i}(c) & =\min_{s\in \{ d+1,...\}} f_{i}(s)-cg_{i}(s), \ i\in \{1,3,4\}, \label{h-func-1} \\  
h'_2(c) & = \min_{s\in \{ 2,...d \}}  f_{2}(s)-cg_{2}(s), \label{h-func-2}
\end{align}
then for all $i=1,2,3,4$, $h'_{i}(c)\lesseqqgtr 0$ if and only if $c\gtreqqless \beta'_{i}$.  
\end{lemma}
\begin{proof} 
        %\ifreport
See Appendix \ref{fractionalapp}.
        %\else
%See our technical report \cite{pan2020age}.
        % \fi
\end{proof}
The solution \textcolor{black}{to} $h'_{i}(c)$ in Lemma \ref{fractional} is shown in the following lemma:

\begin{lemma}\label{fracsolution} Suppose that $i\in \{1,2,3,4 \}$.
If $(p,q,d)\in \textbf{B}_2\cup \textbf{B}_3$, then the threshold $s_{i}(c)$ defined in \eqref{threshold1} is the solution to \eqref{h-func-1} and \eqref{h-func-2}, i.e., $h_i(c)=h'_i(c)$.  
\end{lemma}
\begin{proof}
       %\ifreport
See Appendix \ref{app2}.
      % \else
%See our technical report \cite{pan2020age}.
      % \fi
\end{proof}

%From Lemma \ref{fractional}, $h_{i}(c)=0$ is the solution of \eqref{algr}. So from Lemma \ref{fracsolution}, 
Therefore, we can immediately conclude that  
for all $i\in \{1,2,3,4\}$:
\begin{equation}\label{final}
%    \beta_{i}=\min_{s\in N(i)}\frac{f_{i}(s)}{g_{i}(s)}.
\beta'_{i}=\beta_i,
\end{equation} where $\beta'_i$ is defined in \eqref{discuss-beta1}, \eqref{discuss-beta2} and $\beta_{i}$ is derived in Definition \ref{defthm} with low complexity algorithm.
\textcolor{black}{In addition, 
\begin{align}
 s_i(\beta_i) & = \underset{s\in \{ d+1,\ldots\} }{\argmin} f_i(s)/g_i(s), i\in \{1,3,4\}, \label{final1}\\
 s_2(\beta_2) & = \underset{s\in \{ 2,\ldots,d \} }{\argmin} f_2(s)/g_2(s). \label{final2}
\end{align} }
The studies in \cite{sun2019samplings,ornee2019samplings,sun2019sampling} also derive an exact solution to their inner-layer problem. However, their technique is using optimal stopping rules \cite{ornee2019samplings,sun2019samplings} or stochastic convex optimization \cite{sun2019sampling}, which is different with our study. In conclusion, \eqref{final} and Proposition \ref{propositionb} shows Theorem \ref{theorem2b}.

\subsubsection{Proof of Theorem \ref{theorem2c}}

 When $(p,q,d)\in \textbf{B}_3 $ 
 %\ifreport
 , $\mu^*(\delta,0,0)$ and $\mu^*(\delta,1,0)$ are non-decreasing. \textcolor{black}{Then, the two cases are removed: $\mu^*(\delta,0,0)=2$, $\mu^*(\delta,1,0)=1$, $s\leq d$ or $s>d$}. Since $(p,q,d)\in \textbf{B}_3 $ does not imply $p+q\le 1$ or $p+q\ge 1$, we will enumerate all of the five possible ways to represent the DTMCs of the threshold policy based on the value of the threshold $s$ and the optimal decision at states $(d,1,0)$ and $(1,1,0)$ (see 
     %\ifreport
Appendix \ref{markovapp}
     %\else
     %our technical report \cite{pan2020age}
     %\fi
 for the corresponding DTMCs): 
\begin{itemize}
    \item 
 The threshold $s>d$ and $\mu^*(1,1,0)=\mu^*(d,1,0)=1$ ($\lambda^*_1\in \{d+1,d+2,...\}$). The average age is derived as $f_{1}(s)/g_{1}(s)$.
   \item 
  The threshold $s>d$, $\mu^*(1,1,0)=1$ and $\mu^*(d,1,0)=2$ ($\lambda^*_1\in \{2,...,d\}$). Then, the average age is $f_{3}(s)/g_{3}(s)$, \textcolor{black}{which is shown in Appendix \ref{markov4app}.}
   \item 
 The threshold $s>d$ and $\mu^*(1,1,0)= \mu^*(d,1,0)=2$ ($\lambda^*_1\in \{2,...,d\}$) with average age $f_{4}(s)/g_{4}(s)$, \textcolor{black}{which is shown in Appendix \ref{markov5app}}.
   \item 
  The threshold $s\le d$ and $\mu^*(1,1,0)=\mu^*(d,1,0)=1$ ($\lambda^*_1\in \{d+1,d+2,...\}$), with average age $f_{2}(s)/g_{2}(s)$.
   \item 
  The threshold  $s\le d$ and $\mu^*(d,1,0)=2$. Then, regardless of $ \mu^*(1,1,0)$ ($\lambda^*_1\in \{1,2,...,d\}$), the \textcolor{black}{DTMC corresponds} to always choosing $2$, with average age $(3/2)d-1/2$. 
\end{itemize} 

\noindent Then, we directly have the following result:

\begin{proposition}\label{propositionc}
If $(p,q,d)\in \textbf{B}_3$, then the optimal scheduling policy is 
\begin{align}
    \mu^*(\delta,0,0)=\left\{
\begin{array}{lll}
  1 & \text{if }~ \delta< \lambda_0^*;\\
  2  & \text{if }~ \delta\ge\lambda_0^*, 
\end{array}
\right. \\
    \mu^*(\delta,1,0)=\left\{
\begin{array}{lll}
  1 & \text{if }~ \delta< \lambda_1^*;\\
  2  & \text{if }~ \delta\ge\lambda_1^*, 
\end{array}
\right. 
\end{align} 
where $\lambda_0^*$ and $\lambda_1^*$ are given by 
%The thresholds $\lambda_0^*$ and $\lambda_1^*$ are listed as follows:
%The thresholds $\lambda_0^*$ is a unique value and $\lambda_1^*$ may take multiple values, which are listed as follows:
\begin{equation}\label{proposition2cc}
\!\!\!\!    \left\{
\begin{array}{lll}
  \lambda_0^*=\underset{s\in \{ d+1,\ldots\} }{\argmin} f_1(s)/g_1(s), \lambda_1^*\in \{d+1,\ldots 
\}  \   \text{if~} \bar{\Delta}_{\text{opt}} =  \beta_{1}',\\
  \lambda_0^*=\underset{s\in \{ 2,\ldots, d\} }{\argmin} f_2(s)/g_2(s), \  \lambda_1^*\in \{d+1,\ldots 
\}   \    \text{if~} \bar{\Delta}_{\text{opt}} =  \beta_{2}',\\
  \lambda_0^*=\underset{s\in \{ d+1,\ldots\} }{\argmin} f_3(s)/g_3(s), \lambda_1^*\in \{2,\ldots,d 
\} \ \ \   \text{if~} \bar{\Delta}_{\text{opt}} =  \beta_{3}',\\
 \lambda_0^*=\underset{s\in \{ d+1,\ldots\} }{\argmin} f_4(s)/g_4(s),  \lambda_1^*\in \{2,\ldots,d 
\} \ \ \   \text{if~}\bar{\Delta}_{\text{opt}} =  \beta_{4}',\\ \lambda_0^*=1, \ \ \  \lambda_1^*\in \{1,2,\ldots,d\} \ \ \  \text{if~} \bar{\Delta}_{\text{opt}} =  (3/2)d-1/2, 
\end{array}
\right. \!\!\!\!
\end{equation}
 $\bar{\Delta}_{\text{opt}}$ is the optimal objective value of \eqref{avg}, determined by
\begin{equation}
   \bar{\Delta}_{\text{opt}} =  \min \Big{\{}   \beta_{1}', \beta_{2}', \beta_{3}', \beta_{4}', \frac{3}{2}d-\frac{1}{2} \Big{\} }.
\end{equation}
\end{proposition}

\textcolor{black}{According to \eqref{final}, \eqref{final1} and \eqref{final2}, Theorem \ref{theorem2c} is shown directly from Proposition \ref{propositionc}.}

\subsubsection{Proof of Theorem \ref{theorem2d}}

For $(p,q,d)\in \textbf{B}_4 $, $\mu^*(\delta,1,0) $ is non-decreasing in $\delta$ from Theorem \ref{theorem2}. Also, $\mu^*(\delta,0,0)=1$ by Lemma \ref{mu00in1and4}.

If $\mu^*(1,1,0)=1$, the policy becomes always choosing Channel $1$ \textcolor{black}{(since $(d,1,0)$ is not reached at any time slot with probability $1$)}. If $\mu^*(1,1,0)=2$, then $\mu^*(\delta,1,0)=2$ for all $\delta$.   
Thus, the solution to the optimal threshold-type policy when $(p,q,d)\in \textbf{B}_4 $ may contain two possible steady-state DTMCs which directly gives Theorem \ref{theorem2d}: \begin{itemize}
    \item The optimal decision $\mu^*(\delta,0,0)=1$ \textcolor{black}{for all $\delta\geq1$} and $\mu^*(1,1,0)=1$. Then, the optimal policy is always choosing Channel $1$. The average age of always choosing Channel $1$ is $((1-q)(2-p)+(1-p)^2)/((2-q-p)(1-p))$ as in \eqref{always1}.
    \item The optimal decision $\mu^*(\delta,0,0)=1$ and  $\mu^*(\delta,1,0)=2$ for all \textcolor{black}{$\delta\geq1$}. See
    % \ifreport
Appendix \ref{markovapp}
      %  \else
     %our technical report \cite{pan2020age}
     %\fi
for the corresponding DTMC. \textcolor{black}{The average age by analyzing the steady-state distribution of this DTMC is $f'_0/g'_0$ which is shown in Appendix \ref{markov6app}.}\end{itemize}

Therefore, the listed items directly proves Theorem \ref{theorem2d}.

From our analysis in Section \ref{proof3}, we have the following conclusion for the proof of Theorem \ref{theorem2a}-\ref{theorem2d}: (i) If $(p,q,d)\in \textbf{B}_1$, the optimal decision is always choosing Channel $1$; 
\textcolor{black}{(ii) If $(p,q,d)\in \textbf{B}_2$, $ \textbf{B}_3$ or $\textbf{B}_4$,
there are a couple of possible cases ($4$ cases for $(p,q,d)\in \textbf{B}_2$, $5$ cases for $(p,q,d)\in \textbf{B}_3$ and $2$ cases for $(p,q,d)\in \textbf{B}_4$, respectively). Each case corresponds to analyzing the steady-state distribution of a single DTMC or a collection of DTMCs over the threshold $s$; in the latter case, the optimal threshold can be computed efficiently using bisection search.     
% the problem turns into a couple of minimization of average age over the possible threshold values. 
 The optimal objective value in \eqref{avg} is the minimum of the derived ages in each cases and the optimal thresholds are determined by the case that achieves the minimum.} %Thus, we have the following proposition coupled with a definition:

%% file: sections/appendix.tex
\appendices

    \begin{table}  \caption{Notations for $f_i(s),g_i(s),l_i,o_i$ $(i=1,2,3,4)$ in Definition \ref{defthm}  } 
\begin{center}  
\begin{tabular}{| p{2.5em} |l|}  
\hline  
Name & Expression  \\ \hline
 & $\begin{bmatrix} a'_d & b'_d \\ a_d & b_d \end{bmatrix}= \begin{bmatrix} q & 1-q \\ 1-p & p \end{bmatrix}^d$ \\
 \hline
$f_0$ & $q\sum_{i=1}^{d}i+(1-q)\sum_{i=d+1}^{2d}i+(\frac{b'_d q+b_d}{1-b_d}+1)\sum_{i=d+1}^{2d}i$  \\
$g_0$ & $\frac{b'_d q+b_d}{1-b_d}d+d+1$\\
\hline
$f'_0$ & $\sum_{i=1}^{d}i+a'_d/b'_d\times \sum_{i=d}^{2d-1}i+\sum_{i=d}^{\infty}ip^{i-d}$  \\
$g'_0$ & $d/b'_d+1/(1-p)$\\
\hline
$c_1(s)$ & $ 1-b_d p^{s-d}-(1-q)a_d p^{s-d-1} $\\
 $f_{1}(s)$     & $c_1(s) (p/(1-q)+\sum_{i=2}^{d}i p^{i-1})+d p^d$\\ & $+\sum_{i=d+1}^{s-1}ip^{i-1}+\sum_{i=s}^{s+d-1}ip^{s-1}$\\ \hline
 $g_{1}(s)$ & $c_1(s) (p/(1-q)+\sum_{i=2}^{d} p^{i-1})+p^{d}$ \\  &  $+\sum_{i=d+1}^{s-1}p^{i-1}+\sum_{i=s}^{s+d-1}p^{s-1}$ \\ \hline 
 $c_2$ & $b_d/(a_d q)$ \\ 
$f_{2}(s)$  &  $(p/(1-q)+\sum_{i=2}^{s}i p^{i-1})+\sum_{i=s+1}^{s+d-1}i p^{s-1}$ \\ & $+c_2\sum_{i=d-1}^{2d-1}ip^{s-1} +(\sum_{i=d+1}^{2d}(1-q)i+d)p^{s-1}/q$ \\
\hline
$g_{2}(s)$ & $(p/(1-q)+\sum_{i=2}^{s}p^{i-1})+(d-1)p^{s-1}+c_2 d p^{s-1}$\\ & $+(d(1-q)+1)p^{s-1}/q$   \\
      \hline  
$f_{3}(s)$ & $(1-p^{s-d})(p/(1-q)+\sum_{i=2}^{d-1}ip^{i-1})$ \\
& $+\sum_{i=d}^{s}ip^{i-1} +\frac{a_d}{b'_d}p^{s-1}\sum_{i=d}^{2d-1}i+\sum_{i=s+1}^{s+d-1}ip^{s-1}$\\  \hline
$g_{3}(s)$ & $(1-p^{s-d})(p/(1-q)+\sum_{i=2}^{d-1}p^{i-1})$\\ & $+\sum_{i=d}^{s}p^{i-1}+ \frac{a_d}{b'_d}p^{s-1}d+(d-1)p^{s-1}$\\  \hline 
$c_4(s)$ & $(a'_d+(a_d-a'_d)p^{s-d})/(b'_d)$\\
$f_{4}(s)$ &  $(1-p^{s-d}) \sum_{i=1}^{d} i +c_4(s) \sum_{i=d}^{2d-1} i $\\ & $+\sum_{i=s+1}^{s+d-1} ip^{s-d}+\sum_{i=d}^{s}ip^{i-d}$  \\\hline
$g_{4}(s)$ &   $(1-p^{s-d})d+c_4(s)d+(d-1)p^{s-d}$\\ & $+\sum_{i=d}^{s}p^{i-d}$\\ \hline
$l'_{1}$ & $-p^{-d+1}(b_d+(1-q)a_d)(1-p)(p/(1-q)+\sum_{i=2}^{d}ip^{i-1})$\\ & $+d(p-(1-p)(d-1)/2)$ \\
 $o_{1}$ & $-p^{-d+1}(b_d+(1-q)a_d)(1-p)(p/(1-q)+\sum_{i=2}^{d}p^{i-1})$\\ & $+1-(1-p)d$ \\
 \hline
 $l'_{2}$ & $-c_2(1-p) \sum_{i=d-1}^{2d-1}-(\sum_{i=d+1}^{2d}(1-q)i+d+1)(1-p)/q$ \\ & $+ d(p-(1-p)(d-1)/2)$\\
 $o_{2}$ & $p-(1-p)(1+d(1-q))/q$ \\ & $-(1-p)d(1+c_2)$ \\
 \hline $l'_{3}$ & $(p/(1-q)+\sum_{i=2}^{d-1} i p^{i-d}(1-p)) - \sum_{i=d}^{2d-1}i(1-p)a_d/b'_d$\\& $+(d-1)(p-d(1-p)/2)$\\
 $o_{3}$ & $(p/(1-q)+\sum_{i=2}^{d-1}  p^{i-d}(1-p))$ \\ & $+1-(1-p)(d-1+d a_d/b'_d)$ \\ \hline $p^{d-1} l'_{4}$ & $-(1-p)\sum_{i=1}^{d-1} i -\sum_{i=d}^{2d-1}(1-p)/(1-a_d)$\\ & $+(d-1)(p-d(1-p)/2)$\\
 $p^{d-1} o_{4}$ & $-(1-p)d-(1-p)d/(1-a_d)-(d-1)(1-p)+1$ \\ \hline
\end{tabular}  
\end{center} 
\label{fg}
\end{table}

\section{Proof of Lemma \ref{zerowait}}\label{zerowaitapp}

Suppose that the age at initial time $0$ is the same for any policy.
For any given policy $\pi\in \Pi$, we construct a policy $\pi'$: whenever both channels are idle and $\pi$ chooses \emph{none}, $\pi'$ chooses Channel $1$, and at other time $\pi$ and $\pi'$ are the same.
The equivalent expression of $\pi'$ is given as follows:   
\begin{equation}
    \pi'(t) = \left\{
\begin{array}{lll}
  1 & \text{if } l_2(t)=0 \text{ and } \pi(t)=\emph{none}  ;\\
  \pi(t)  & \text{else.} 
\end{array}
\right. 
\end{equation}
The policy $\pi$ and $\pi'$ are coupled given a sample path of Channel $1$: $\mathcal{I}=\{l_1(0), l_1(1),\ldots \}$. For any $\mathcal{I}$, we want to show that the age of policy $\pi'$ is smaller or equal to that of $\pi$.

 For simplicity, we use $\Delta_{\pi}(t)$ and $l_{2}(t)$ to be the age and the state of Channel $2$, respectively, with a policy $\pi$ and $\mathcal{I}$. Compared with $\pi$, $\pi'$ only replaces \emph{none} by $1$. Thus, the state of Channel $2$ of $\pi'$ is still $l_2(t)$. 
 %We denote $V_\pi=\limsup_{T\rightarrow \infty} 1/T \sum_{t=1}^{T} \mathbb{E}[\Delta_\pi(t)]$ to be the average expected age of policy $\pi$, and $V_{\pi} | \mathcal{I} =\limsup_{T\rightarrow \infty} 1/T \sum_{t=1}^{T} \Delta_\pi(t) | \mathcal{I}$ to be the average age of policy $\pi$ given $\mathcal{I}$. 
 
 Then, we will show that for all time $t$ and any $\mathcal{I}$, the age $\Delta_{\pi'}(t)\le \Delta_{\pi}(t)$. %(ii) the state $l_{2,\pi'}(t)=l_{2,\pi}(t)$.
 %Thus, $V_{\pi'} | \mathcal{I}\le V_{\pi} | \mathcal{I}$ for any $\mathcal{I}$. Thus, $V_{\pi'}\le V_{\pi}$, which proves our result.  
 We prove by using induction. 
 
If $t=0$, then according to our assumption, the hypothesis trivially holds.

Suppose that the hypothesis holds for $t=k$. We will show for $t=k+1$.% From hypothesis, $l_{2,\pi}(k)=l_{2,\pi'}(k)$.
We divide the proof into two different conditions: 
(i) If $l_{2}(k)>0$, then $\pi(k)=\pi'(k)=\emph{none}$. Thus,
\begin{equation}\begin{split}
   & \Delta_{\pi}(k+1)= \left\{
\begin{array}{lll}
   \Delta_{\pi}(k)+1 & \text{if } l_2(k)\ge 2 ;\\
  d  & \text{if } l_2(k)= 1, 
\end{array}
\right.  \\ &  \Delta_{\pi'}(k+1)= \left\{
\begin{array}{lll}
   \Delta_{\pi'}(k)+1 & \text{if } l_2(k)\ge 2 ;\\
  d  & \text{if } l_2(k)= 1. 
\end{array}
\right. \end{split}
\end{equation}
Thus, $\Delta_{\pi'}(k+1)\le \Delta_{\pi}(k+1)$. %and $l_{2,\pi'}(k+1)=l_{2,\pi}(k+1)=l_{2,\pi}(k)-1$.

\noindent (ii) If $l_{2}(k)=0$, then $\pi(k)$ may take \emph{none}, $1$, or $2$. If $\pi(k)=1$ or $2$, then $\pi'(k)=\pi(k)$. Thus, the hypothesis directly gives $\Delta_{\pi'}(k+1)\le \Delta_{\pi}(k+1)$. %and $l_{2,\pi'}(k+1)=l_{2,\pi}(k+1)=d-1$. 
If $\pi(k)=\emph{none}$, then $\pi'(k)=1$. Then,
\begin{equation}
    \begin{split}
  & \Delta_{\pi'}(k+1)\le \Delta_{\pi'}(k)+1, \\
  & \Delta_{\pi}(k+1)= \Delta_{\pi}(k)+1. 
    \end{split}
\end{equation}
Thus, $\Delta_{\pi'}(k+1)\le \Delta_{\pi}(k+1)$.
%$\Delta_{\pi'}(k+1)= \Delta_{\pi'}(k)+1$ or $1$ depending on $l_1(k)$. Thus, $\Delta_{\pi'}(k+1)\le \Delta_{\pi'}(k)+1$ and $\Delta_{\pi}(k+1)= \Delta_{\pi}(k)+1$, which implies $\Delta_{\pi'}(k+1)\le \Delta_{\pi}(k+1)$. Also, $l_{2,\pi'}(k+1)=l_{2,\pi}(k+1)=0$. 
From (i) and (ii), we complete the proof of induction. 
%Thus, we show (i) and (ii) when $t=k+1$. Finally we have completed our proof. 

%The reason is that for a policy $\pi'\notin \Pi'$ that allows an inactive choice ($u(t)=0$) when idle, we can always find out an alternative policy $\pi\in\Pi'$ such that  for all $t$. Thus, the average age $V_{\pi}\le V_{\pi'}$. The policy $\pi$ is designed as below: 

\section{Proof of Lemma \ref{concave}}\label{concaveapp}
Similar techniques were also used recently in \cite{ornee2019samplings}. 

(1) According to Lemma \ref{fracsolution}, the function $h_{i}(\beta)$ in \eqref{beta} also satisfies
%\begin{equation}\label{lemma01}
 \begin{align}%\label{inlayer}
h_{i}(\beta) & =\min_{s\in \{ d+1,...\}} f_{i}(s)-\beta g_{i}(s), \ i\in \{1,3,4\}, \\  
h_2(\beta) & = \min_{s\in \{ 1,...d \}}  f_{2}(s)-\beta g_{2}(s), 
\end{align}
%\end{equation} %where the set $N(i)$ is defined in Definition \ref{def6}. 
The function $f_{i}(s)-\beta g_{i}(s)$ is linearly decreasing, which is concave and continuous. Since the minimization preserves the concavity and continuity, $h_i(\beta)$ is still concave. From Table \ref{fg}, it is easy to show that there exists a positive $d'$ such that $f_{i}(s)>d'$ and $g_{i}(s)>d'$ for all $i\in \{1,2,3,4\}$. So, for all $s$ and any $\beta_1<\beta_2$, $f_{i}(s)-\beta_1 g_{i}(s)>f_{i}(s)-\beta_2 g_{i}(s)$. Thus, $h_{i}(\beta)$ is strictly decreasing. 

(2) Since $f_{i}(s)>d'$ and $g_{i}(s)>d'$, so $h_{i}(0)>0$. Moreover, since $h_i(\beta)$ is strictly decreasing, we have $\lim_{\beta\rightarrow \infty}h_{i}(\beta) =-\infty$. 

\section{proof of Lemma \ref{finite}}\label{app1}

Consider the policy that always idles at every time slot (i.e., $\pi(t)=\emph{none}$ for all $t$). Under this policy, the age increases linearly with time. The discounted cost under the aforementioned policy acts as an upper bound on the optimal value function $J^\alpha(\textbf{s})$. Thus, for any initial state $s=(\delta,l_1,l_2)$, $J^\alpha(\textbf{s})$ satisfies  
%We consider the extreme case such that no packet is delivered at any time. In that case Age increases linearly in time. For an initial state $\textbf{s} = (\delta_0,l_1,l_2)$, the value function satisfies
\begin{equation}
    J^\alpha(\textbf{s}) \le \delta+ \alpha(\delta+1)+ \alpha^2(\delta+2)... 
     =\frac{(\delta+\frac{\alpha}{1-\alpha})}{1-\alpha}<\infty, \label{finiteproof}
\end{equation} which proves the result.

\section{proof of Lemma \ref{monotonicity}}\label{monotonicityapp}
We show Lemma \ref{monotonicity} by using induction in value iteration \eqref{value-iterate}. We want to show that $J^\alpha_n(\textbf{s})=J^\alpha_n(\delta,l_1,l_2)$ is increasing in age $\delta$ for all iteration number $n$.
%$J^\alpha_n(\delta_1,l_1,l_2)\ge J^\alpha_n(\delta_2,l_1,l_2)$ for all $\delta_1\ge\delta_2$. 

If $n=0$, $J^\alpha_0 (\delta,l_1,l_2)=0$, so the hypothesis holds.  
Suppose the hypothesis holds for $n=k$, then we will show that it also holds for $n=k+1$.  
First, note that in \eqref{value-iterate}, the immediate cost of any state $\textbf{s}=(\delta,l_1,l_2)$ is $\delta$, which is increasing in age. Second, by our hypothesis and the evolution of age in Section \ref{systemod} , $\sum_{\textbf{s}'\in \textbf{S}}P_{\textbf{s}\textbf{s}'}(u)J^\alpha_k(\textbf{s}')$ is increasing in age $\delta$. Thus, $Q^\alpha_{k+1}(\textbf{s},u)$ is increasing in age $\delta$. Thus, $J^\alpha_{k+1}(\textbf{s})$ is increasing in age $\delta$ and we have completed the induction.  

%, it is easy to find that for all $\textbf{s}'\in \textbf{S}$,  

%For induction part, we take $l_1=l_2=0$ for example, from \eqref{value-iterate},
%\begin{equation*}\begin{split}
%    &Q^{\alpha}_{n+1}(\delta,0,0,1)\\=& \delta+\alpha pJ^{\alpha}_n(\delta+1,0,0)+\alpha (1-p)J^{\alpha}_n(1,1,0), \\
%    & Q^{\alpha}_{n+1}(\delta,0,0,2)\\=& \delta+\alpha pJ^{\alpha}_n(\delta+1,0,d-1)+\alpha (1-p)J^{\alpha}_n(\delta+1,1,d-1).
%    \end{split}
%\end{equation*} We get if $l_1=l_2=0$, for $u=1,2$ 
%\begin{equation}\begin{split}
% &   Q^{\alpha}_{n+1}(\delta,l_1,l_2,u)\ge Q^{\alpha}_{n+1}(\delta-1,l_1,l_2,u).\label{lemmaincreasing}\end{split}
%\end{equation}
%So $J^{\alpha}_{n+1}(\delta,0,0)\ge J^{\alpha}_{n+1}(\delta-1,0,0)$. For other values $l_1,l_2$, we use \eqref{value-iterate} similarly. And  \eqref{lemmaincreasing} holds for all $l_1,l_2$. Overall, by induction, for all $n,\delta,l_1,l_2$, we have $J^{\alpha}_{n+1}(\delta,l_1,l_2)\ge J^{\alpha}_{n+1}(\delta-1,l_1,l_2)$. Thus, the limit of $J^{\alpha}_{n}(\delta,l_1,l_2)$, which is the value function $J^{\alpha}(\delta,l_1,l_2)$, is monotone increasing in age $\delta$. 

\section{Proof of Lemma \ref{theorem1}} \label{proof1}
Similar techniques were also used recently in \cite{hsu2019scheduling}.

According to \cite{sennott1989average} and Lemma \ref{finite}, it is sufficient to show that Problem \eqref{avg} satisfies the following two conditions: 

(a) There exists a non-negative function $M({\textbf{s}})$ such that $h^\alpha (\textbf{s})\le M({\textbf{s}}) $ for all $\textbf{s}$ and $\alpha$, where the relative function $h^\alpha (\textbf{s}) = J^\alpha(\textbf{s})-J^\alpha(1,1,0)$. 

(b) There exists a non-negative $N$ such that $-N\le h^\alpha(\textbf{s})$ for all the state $\textbf{s}$ and $\alpha$. 

For (a), we first consider a stationary deterministic policy $f$ that always chooses Channel $1$. The states $(1,1,0)$, $(\delta,0,0)$ ($\delta\ge 2$) are referred as \emph{recurrent} states. The remaining states in the state space are referred as \emph{transient} states. %For a transient state $\textbf{s}=(\delta,l_1,l_2)$, it takes at most $d$ time slots to back to a recurrent state $\textbf{s}'$ (i.e., at most $d$ time slots with $l_2>0$ and at most $1$ time slot with $l_2=0$). So the cost of returning to a recurrent state is finite, and upper bounded by $d(\delta+d-1)$.
Define $e(\textbf{s}_1,\textbf{s}_2)$ to be the average cost of the first passage from $\textbf{s}_1$ to $\textbf{s}_2$ under the policy $f$ where $\textbf{s}_1$ and $\textbf{s}_2$ are recurrent states. The recurrent states of $f$ form an aperiodic, recurrent and irreducible Markov chain. So, from Proposition $4$ in \cite{sennott1989average}, for any recurrent state $\textbf{s}'$, $e({\textbf{s}',\textbf{s}_0})$ is finite (where $\textbf{s}_0=(1,1,0)$).

Now, we pick $M({\textbf{s}})$. We let $M({\textbf{s}})=e({\textbf{s}',\textbf{s}_0})+d(\delta+d-1)$ for the transient state $\textbf{s}$, and let $M({\textbf{s}})=e({\textbf{s},\textbf{s}_0})$ for the recurrent state $\textbf{s}$. 
Then, replacing "$BE_\theta (T)$" by $M(\textbf{s})$ in the proof of Proposition $1$ in \cite{sennott1986new}, we have $h^\alpha (\textbf{s})\le M({\textbf{s}})$ for all the state $\textbf{s}$. Overall, there exists $M({\textbf{s}})$ such that $h^\alpha (\textbf{s})\le M({\textbf{s}})$.

We start to show (b). According to Lemma \ref{monotonicity}, the value function is increasing in age. Thus, we only need to show that there exists $N$ such that $-N\le h^\alpha(1,l_1,l_2)$ for all $l_1$ and $l_2$. In order to prove this, we will show that there exists $E'(1,l_1,l_2)$ such that $-E'(1,l_1,l_2)\le h^\alpha(1,l_1,l_2)$ for all $l_1$ and $l_2$. Thus, we take $N =\max_{l_1\in \{ 0,1\},l_2\in \{ 0,1,...,d-1 \}}E'(1,l_1,l_2)$, which is still finite, and condition (b) is shown. 

Now, we start to find out $E'(1,l_1,l_2)$. 

We split the states $(1,l_1,l_2)$ into three different cases. 

$(b1)$ If $l_1=1$ and $l_2=0$, then $h^\alpha(1,l_1,l_2)=0$. Thus, we take $E'(1,l_1,l_2)=0$.

$(b2)$ If $l_1=0$ and $l_2=0$, then we take $T=1$ if the optimal decision of $(1,0,0)$ is $1$ and take $T=d+1$ if the optimal decision is $2$. Therefore, the definition of $J^\alpha (1,0,0)$ tells that there exists $T\in \{1,d+1\}$ and $a\in [0,1]$ such that 
\begin{equation}\begin{split}
   & J^\alpha(1,0,0)\\  \ge & \sum_{i=1}^{T} i\alpha^{i-1}+\alpha^{T} \left( a J^\alpha(1,1,0)+(1-a) J^\alpha(1,0,0) \right)\\  \ge & \alpha^{T}(a J^\alpha(1,1,0)+(1-a) J^\alpha(1,0,0)) \\
     = & \alpha^{T} \Big{(}a J^\alpha(1,1,0)+(1-a) (h^\alpha(1,0,0)+J^\alpha(1,1,0))\Big{)}\\
 = & \alpha^T J^\alpha(1,1,0) + \alpha^T (1-a) h^\alpha(1,0,0)     \\
  \ge & \alpha^T J^\alpha(1,1,0) + \alpha^T (1-a) h^\alpha(1,0,0). 
\end{split}\label{sennot1}
\end{equation}
Notice that $J^\alpha (1,1,0)$ is smaller or equal to the $\alpha$-discounted cost of always choosing Channel $1$ with initial state $(1,1,0)$. Consider a special case of MDP consisting of states $(1,1,0),(2,0,0),(3,0,0),\ldots,$ where there is only one policy that always chooses Channel $1$.  
Therefore, applying Lemma A2 in appendix of \cite{sennott1989average} to this MDP, $(1-\alpha)J^{\alpha}(1,1,0)$ is upper bounded by a constant $c'$ that is not a function of $\alpha$. Note that 
\begin{equation}
    1-\alpha^T=(1-\alpha)(1+\alpha+...+\alpha^{T-1}) \le (1-\alpha) T.
\end{equation}
Then from \eqref{sennot1},  we get
\begin{equation}
    \begin{split}
 & h^\alpha (1,0,0) \\ = & J^\alpha(1,0,0)- J^\alpha(1,1,0)\\
 \ge & -(1-\alpha^T)J^\alpha(1,1,0) + \alpha^T (1-a) h^\alpha(1,0,0) \\
  \ge & -Tc' +\alpha^T (1-a) h^\alpha(1,0,0). \\      
    \end{split}
\end{equation} 
Therefore, 
\begin{equation}
h^\alpha (1,0,0) \ge -\frac{Tc'}{1-\alpha^T(1-a)} \triangleq E'(1,0,0).
\end{equation}

$(b3)$ If $l_2>0$, similar with $(b2)$, we take $T' = l_2$ and there exists $a'\in [0,1]$ such that 
\begin{align}
h^\alpha (1,l_1,l_2) & \ge -T' c' - \alpha^{T'} (1-a') \frac{Tc'}{1-\alpha^{T}(1-a)}\\ & \triangleq E'(1,l_1,l_2).
\end{align}
Therefore, we have found out all the values of $E'(1,l_1,l_2)$. 
Overall, by proving (a) and (b), we complete the proof of Lemma \ref{theorem1}.     

%Then from (a), under the deterministic and stationary policy $f$, $C_{\textbf{s}\textbf{s}_0}(f)<\infty$ for all $\alpha$, where $\textbf{s}_0=(1,1,0)$. Then from proof of proposition $5$ in \cite{sennott1989average}, $h^\alpha (\textbf{s})\ge -C_{\textbf{s}\textbf{s}_0}$

\section{proof of Lemma \ref{g2m}}\label{g2mapp}
 Recall that we use $l_1$ (which is $0$ or $1$) to denote the state of Channel $1$ and that \begin{equation}
     L^{\alpha}(\delta,l_1,2)=Q^{\alpha}(\delta,l_1,0,2)-Q^{\alpha}(\delta-1,l_1,0,2)
 \end{equation}.
%The key idea of the proof is that any state $(\delta,l_1,0)$ after the action $2$ will restart to the state $(d,0,0)$ or $(d,1,0)$ after $d$ time slots. Thus, the value functions coming from the Q-function $Q^{\alpha}(\delta,l,0,2)$ and $Q^{\alpha}(\delta-1,l_1,0,2)$ will be \emph{cancelled out}. Thus, the function $L(\delta,l_1,2)$ becomes a constant, $m$. 
We define the sequences $a_k$, $a'_k$, $b_k$, $b'_k$ with the non-negative index $k$ as
\begin{equation}
    [a_k,b_k] = [0,1]\times P^{k}, \ \ \  [a'_k,b'_k] = [1,0]\times P^{k}, \label{matrix1}
\end{equation} 
where $P$ is the transition probability matrix of Channel $1$, given by $\begin{bmatrix}
q & 1-q \\
1-p & p 
\end{bmatrix}$. 
%Note that $a_k$, or $b_k$ is the probability that after $k$ slots, the state of Channel $1$ is $1$ or $0$ given that the initial state is $0$. $a'_k$, or $b'_k$ is the probability that after $k$ slots, the state of Channel $1$ is $1$ or $0$ given the initial state is $1$. 
Note that \eqref{matrix1} implies $a_k+b_k=a'_k+b'_k=1$ for all the index $k$.

 %The Q-function $Q^{\alpha}(\delta,l_1,0,2)$ satisfies the Bellman equation  \eqref{discounted-lemma1}. Also, for action $2$, 
 
 By using the Bellman equation \eqref{discounted-lemma1} iteratively, 
 $L^{\alpha}(\delta,0,2)$ and $L^{\alpha}(\delta,1,2)$ satisfy the following lemma: 
    \begin{lemma}\label{rhs2}
 %Then for $k=1,2,...,d-1$
    The values $L^{\alpha}(\delta,0,2)$ and $L^{\alpha}(\delta,1,2)$ satisfy: 
     \begin{equation}\begin{split}  & L^{\alpha}(\delta,0,2)  = \sum_{i=0}^{d-2}\alpha^i \\ & +  \alpha^{d-1} a_{d-1} \Big{(}J^{\alpha}(\delta+d-1,1,1)-J^{\alpha}(\delta+d-2,1,1)\Big{)} \\ & +  \alpha^{d-1} b_{d-1}\Big{(}J^{\alpha}(\delta+d-1,0,1)-J^{\alpha}(\delta+d-2,0,1)\Big{)},\\
     & L^{\alpha}(\delta,1,2)  = \sum_{i=0}^{d-2}\alpha^i \\ & +  \alpha^{d-1} a'_{d-1} \Big{(}J^{\alpha}(\delta+d-1,1,1)-J^{\alpha}(\delta+d-2,1,1)\Big{)} \\ & +  \alpha^{d-1} b'_{d-1}\Big{(}J^{\alpha}(\delta+d-1,0,1)-J^{\alpha}(\delta+d-2,0,1)\Big{)},\label{equation02}\end{split}
\end{equation}
    where $a_{d-1}$, $a'_{d-1}$, $b_{d-1}$, $b'_{d-1}$ are defined in \eqref{matrix1}.
%     \begin{equation}\begin{split}  & L(\delta,0,2)  = \sum_{i=0}^{k-1}\alpha^i \\&+\alpha^k a_k \Big{(}J^{\alpha}(\delta+k,1,d-k)-J^{\alpha}(\delta+k-1,1,d-k)\Big{)} \\ & +\alpha^k b_k\Big{(}J^{\alpha}(\delta+k,0,d-k)-J^{\alpha}(\delta+k-1,0,d-k)\Big{)},\\
%     & L(\delta,1,2)  = \sum_{i=0}^{k-1}\alpha^i \\ &+\alpha^k a'_k \Big{(}J^{\alpha}(\delta+k,1,d-k)-J^{\alpha}(\delta+k-1,1,d-k)\Big{)} \\ & +\alpha^k b'_k\Big{(}J^{\alpha}(\delta+k,0,d-k)-J^{\alpha}(\delta+k-1,0,d-k)\Big{)}.\label{equation02}\end{split}
%\end{equation}
    \end{lemma}
 \begin{proof} Please see Appendix \ref{rhs2app} for details.
 \end{proof}  
 Note that the state of Channel $2$ represents the remaining transmission time of Channel $2$. From Lemma \ref{rhs2}, $L^{\alpha}(\delta,0,2)$, $L^{\alpha}(\delta,1,2)$ can be expressed by $J^{\alpha}(\delta+d-1,0,1),J^{\alpha}(\delta+d-1,1,1),J^{\alpha}(\delta+d-2,0,1)$ and $J^{\alpha}(\delta+d-2,1,1)$. Observe that $l_2=1$ in all of these terms. Thus, we can use \eqref{discounted-lemma1} to further expand these terms and prove $L^{\alpha}(\delta,l_1,2)=m$.
 
  %  We take the last term that $k=d-1$ in \eqref{equation02}. Then the value functions $J^{\alpha}(\delta+k,l_1,d-k)$ in \eqref{equation02} are $J^{\alpha}(\delta+d-1,l_1,1)$. Also,  $J^{\alpha}(\delta+k-1,l_1,d-k)$ are $J^{\alpha}(\delta+d-2,l_1,1)$.

    Since the state of Channel $2$ is $1$, then at the next time slot, the state of Channel $2$ is $0$, and the age drops to $d$. So, from \eqref{discounted-lemma1}, for all age value $\delta_0>d-1$, we have 
\begin{equation}\label{equation04}\begin{split}
  & J^{\alpha}(\delta_0,0,1) \\ = & \delta_0 + \alpha p J^{\alpha}(d,0,0)+ \alpha (1-p) J^{\alpha}(d,1,0),  \\
  & J^{\alpha}(\delta_0,1,1) \\ = & \delta_0 + \alpha (1-q) J^{\alpha}(d,0,0)+ \alpha q J^{\alpha}(d,1,0).\end{split}\end{equation} 
  Then, we replace $\delta_0$ by $\delta+d-1$, and $\delta+d-2$ in \eqref{equation04}.
Recall that $[a_{d},b_{d}]=[a_{d-1},b_{d-1}]P$, $[a'_{d},b'_{d}]=[a'_{d-1},b'_{d-1}]P$ and $a_{d-1}+b_{d-1}=a'_{d-1}+b'_{d-1}=1$. Then \eqref{equation02} becomes
\begin{equation}\begin{split}
    L^{\alpha}(\delta,0,2)
    =& \sum_{i=0}^{d-2}\alpha^i + \alpha^{d-1} (a_{d-1}+b_{d-1}) \\ +&\alpha^{d}a_{d}  \Big{(}J^{\alpha}(d,1,0)-J^{\alpha}(d,1,0)\Big{)} \\ +&\alpha^{d}b_{d}\Big{(}J^{\alpha}(d,0,0)-J^{\alpha}(d,0,0)\Big{)}\\ =& \sum_{i=0}^{d-1}\alpha^i \triangleq m. \label{A002}
    \end{split}
\end{equation} Also, 
\begin{equation}\begin{split}\label{A0002}
    L^{\alpha}(\delta,1,2)
    =& \sum_{i=0}^{d-2}\alpha^i + \alpha^{d-1} (a'_{d-1}+b'_{d-1})\\ +&\alpha^{d}a'_{d}  \Big{(}J^{\alpha}(d,1,0)-J^{\alpha}(d,1,0)\Big{)} \\ +&\alpha^{d}b'_{d}\Big{(}J^{\alpha}(d,0,0)-J^{\alpha}(d,0,0)\Big{)}\\ =& \sum_{i=0}^{d-1}\alpha^i \triangleq m. 
    \end{split}
\end{equation}

\section{proof of Lemma \ref{rhs2}}\label{rhs2app}

 We show Lemma \ref{rhs2} by using recursion. The state $(\delta,0,0)$ has a probability of $p$ to increase to $(\delta+1,0,d-1)$, and a probability of $1-p$ to $(\delta+1,1,d-1)$. Thus, \eqref{discounted-lemma1} implies   \begin{equation}\begin{split}\label{qdelta02}
 & Q^{\alpha}(\delta,0,0,2) \\  = & \delta+\alpha p J^{\alpha}(\delta+1,0,d-1)+\alpha (1-p)J^{\alpha}(\delta+1,1,d-1),\end{split}\end{equation}  thus,
 \begin{equation}\begin{split}  L^{\alpha}(\delta,0,2)&= 1 +\alpha p\Big{(}J^{\alpha}(\delta+1,0,d-1)-J^{\alpha}(\delta,0,d-1)\Big{)} \\  +& \ \alpha (1-p)\Big{(}J^{\alpha}(\delta+1,1,d-1)-J^{\alpha}(\delta,1,d-1)\Big{)}.\label{delta02}\end{split}
\end{equation}
 Using similar idea when $l_1=1$, \begin{equation}\begin{split}
 & Q^{\alpha}(\delta,1,0,2) \\  = & \delta+\alpha (1-q)J^{\alpha}(\delta+1,0,d-1)+\alpha q J^{\alpha}(\delta+1,1,d-1),\end{split}\end{equation} 
  Thus, \begin{equation}\begin{split}  & L^{\alpha}(\delta,1,2)\\ = & 1+\alpha (1-q)\Big{(}J^{\alpha}(\delta+1,0,d-1)-J^{\alpha}(\delta,0,d-1)\Big{)} \\ & + \ \alpha q\Big{(}J^{\alpha}(\delta+1,1,d-1)-J^{\alpha}(\delta,1,d-1)\Big{)}.\label{delta12}\end{split}
\end{equation} Observe that, from \eqref{delta02} and \eqref{delta12}, we can express $L^{\alpha}(\delta,l_1,2)$ in terms of $J^{\alpha}(\delta+1,l_1,d-1)$ and $J^{\alpha}(\delta,l_1,d-1)$. Also, the optimal decision is \emph{none} when $l_2>0$. Then, we can iteratively expand $J^{\alpha}(\delta+1,l_1,d-1)$ and $J^{\alpha}(\delta,l_1,d-1)$ using \eqref{discounted-lemma1}. For all the age $\delta_0$:   
\begin{equation}\label{equation06}\begin{split}
  & J^{\alpha}(\delta_0,0,d-1) \\ = & \delta_0 + \alpha p J^{\alpha}(\delta_0+1,0,d-2)+ \alpha (1-p) J^{\alpha}(\delta_0+1,1,d-2),  \\
  & J^{\alpha}(\delta_0,1,d-1) \\ = & \delta_0 + \alpha (1-q) J^{\alpha}(\delta_0+1,0,d-2)+ \alpha q J^{\alpha}(\delta_0+1,1,d-2).\end{split}\end{equation}
Applying \eqref{equation06} into \eqref{delta02} and \eqref{delta12}:
     \begin{equation}\begin{split}   & L^{\alpha}(\delta,0,2)\\  = & 1 + \alpha (a_1+b_1) +\alpha^{2} (q a_1+(1-p)b_1) \\ &  \Big{(}J^{\alpha}(\delta+2,1,d-2)-J^{\alpha}(\delta+1,1,d-2)\Big{)} \\ & +\alpha^{2} (p b_1+(1-q)a_1) \\ &  \Big{(}J^{\alpha}(\delta+2,0,d-2)-J^{\alpha}(\delta+1,0,d-2)\Big{)},\\
     & L^{\alpha}(\delta,1,2) \\ =  & 1+\alpha (a'_{1}+b'_{1})+\alpha^{2} (q a'_{1}+(1-p)b'_{1})\\ & \Big{(}J^{\alpha}(\delta+2,1,d-2)-J^{\alpha}(\delta+1,1,d-2)\Big{)} \\ & +\alpha^{2} (p b'_{1}+(1-q) a'_{1}) \\ & \Big{(}J^{\alpha}(\delta+2,0,d-2)-J^{\alpha}(\delta+1,0,d-2)\Big{)}, \label{equation03}\end{split}
\end{equation} where $a_1=1-p, b_1=p,a'_1=q,$ and $b'_1=1-q$.

%Thus,  give an expression of  when the state of Channel $2$ is $d-1$. We can expand these value functions iteratively, so that the state of Channel $2$ becomes $d-2$,...,$1$. From \eqref{matrix1}, we know $1-p=a_1, p=b_1,q=a'_1,$ and $1-q=b'_1$. Thus, \eqref{delta02} and \eqref{delta12} satisfy the following equation with $k=1$:
   
  % \begin{equation}\begin{split}  & L(\delta,0,2)  = \sum_{i=0}^{k-1}\alpha^i \\&+\alpha^k a_k \Big{(}J^{\alpha}(\delta+k,1,d-k)-J^{\alpha}(\delta+k-1,1,d-k)\Big{)} \\ & +\alpha^k b_k\Big{(}J^{\alpha}(\delta+k,0,d-k)-J^{\alpha}(\delta+k-1,0,d-k)\Big{)},\\
%     & L(\delta,1,2)  = \sum_{i=0}^{k-1}\alpha^i \\ &+\alpha^k a'_k \Big{(}J^{\alpha}(\delta+k,1,d-k)-J^{\alpha}(\delta+k-1,1,d-k)\Big{)} \\ & +\alpha^k b'_k\Big{(}J^{\alpha}(\delta+k,0,d-k)-J^{\alpha}(\delta+k-1,0,d-k)\Big{)}.\label{equation022}\end{split}
%\end{equation}
%Now, we will show \eqref{equation022} holds for all $k=1,2,...,d-1$. Suppose \eqref{equation022} holds for a value $k\le d-2$, then we will show \eqref{equation022} holds for $k+1$. Then,
% Lemma \ref{rhs2} is shown directly by taking $k=d-1$ in\eqref{equation022}.
From \eqref{matrix1}, we have 
\begin{equation}\label{matrixx}
    \begin{split}
&  qa_{1}+(1-p)b_{1}=a_{2},\ pb_{1}+(1-q)a_{1}=b_{2},\\  & qa'_{1}+(1-p)b'_{1}=a'_{2},\  pb'_{1}+(1-q)a'_{1}=b'_{2},\\ & a_{1}+b_{1}=a'_{1}+b_{1}'=1.     
    \end{split}
\end{equation} Applying \eqref{matrixx} in \eqref{equation03}, we get
  \begin{equation}\begin{split}  & L^{\alpha}(\delta,0,2)  = \sum_{i=0}^{1}\alpha^i \\&+\alpha^2 a_2 \Big{(}J^{\alpha}(\delta+2,1,d-2)-J^{\alpha}(\delta+1,1,d-2)\Big{)} \\ & +\alpha^2 b_2\Big{(}J^{\alpha}(\delta+2,0,d-2)-J^{\alpha}(\delta+1,0,d-2)\Big{)},\\
     & L^{\alpha}(\delta,1,2)  = \sum_{i=0}^{1}\alpha^i \\ &+\alpha^2 a'_2 \Big{(}J^{\alpha}(\delta+2,1,d-2)-J^{\alpha}(\delta+1,1,d-2)\Big{)} \\ & +\alpha^2 b'_2\Big{(}J^{\alpha}(\delta+2,0,d-2)-J^{\alpha}(\delta+1,0,d-2)\Big{)}.\label{equation022}\end{split}
\end{equation}
We use \eqref{equation022} iteratively for $d-3$ times, and we get \eqref{equation02} in Lemma \ref{rhs2} (note that if $d=2$, we have proved \eqref{equation02} in \eqref{delta02} and \eqref{delta12}).

\section{proof of Lemma \ref{l0m}}\label{l0mapp}

Frist of all, we observe that $\textbf{B}_1(\alpha)\cup \textbf{B}_4(\alpha)$ implies that $\sum_{i=0}^{\infty}(\alpha p)^i-m \le 0$, while $\textbf{B}_2(\alpha)\cup \textbf{B}_3(\alpha)$ implies that $\sum_{i=0}^{\infty}(\alpha p)^i-m > 0$. Thus, we will need the following lemma:  
\begin{lemma} For any real number $m'$ that satisfies $m'\lesseqqgtr \sum_{i=0}^{\infty}(\alpha p)^i$, we have $m' \lesseqqgtr \sum_{i=0}^{k-1}(\alpha p)^i+(\alpha p)^k m'$ for all $k\in \{0,1,2...\}$.
%  (a) If $m'\ge \sum_{i=0}^{\infty}(\alpha p)^i$, then for all positive integer $k$, the inequality holds: $m' \ge \sum_{i=0}^{k-1}(\alpha p)^i+(\alpha p)^k m'$.  
%(b) If $m'< \sum_{i=0}^{\infty}(\alpha p)^i$, then for all positive integer $k$, the inequality holds: $m' < \sum_{i=0}^{k-1}(\alpha p)^i+(\alpha p)^k m'$.  
%\begin{equation*}
%   m'\lesseqqgtr \sum_{i=0}^{\infty} (\alpha p)^i \Longleftrightarrow \forall k\in \mathbb{N}^+, \  m'\lesseqqgtr\sum_{i=0}^{k-1}(\alpha p)^i+(\alpha p)^k m',  
%\end{equation*} 
\label{smalllemma}
\end{lemma} %\vspace{-12pt}
\begin{proof} Please see Appendix \ref{smalllemmaapp} for details.  
\end{proof}
%Since we have shown from Lemma \ref{g2m} that $L(\delta,0,2)=m$, thus we need to show $L(\delta,0,1)\le m$ if $(p,q,d)\in \textbf{B}_1(\alpha)\cup \textbf{B}_4(\alpha)$, and $L(\delta,0,1)> m$ if $(p,q,d)\in \textbf{B}_2(\alpha)\cup \textbf{B}_3(\alpha)$. 

Next, we need to know an alternative expression of $L^{\alpha}(\delta,0,1)$. %We give the detailed expression of Bellman Equation \eqref{discounted-lemma1} for the Q-function $Q^{\alpha}(\delta,0,0,1)$. With probability $p$, the transmission fails and the state will go to $(\delta+1,0,0)$. With probability $1-p$, the transmission is successful and the state will go to $(1,1,0)$. Thus,     
\begin{equation}\label{A01}\begin{split}
    &Q^{\alpha}(\delta,0,0,1)\\=& \delta+\alpha pJ^{\alpha}(\delta+1,0,0)+\alpha (1-p)J^{\alpha}(1,1,0), \end{split}\end{equation} Thus,\begin{equation}\begin{split}
    & L^{\alpha}(\delta,0,1)= 1+\alpha p\Big{(}J^{\alpha}(\delta+1,0,0)-J^{\alpha}(\delta,0,0)\Big{)} \label{A00}\end{split}.\end{equation}
    Now, we start to prove Lemma \ref{l0m}. From Lemma \ref{g2m}, it is sufficient to show that:
    
    (a) If $(p,q,d)\in \textbf{B}_1(\alpha)\cup \textbf{B}_4(\alpha)$, then $L^{\alpha}(\delta,0,1)\le m$ for $\delta\ge 2$.

(b) If $(p,q,d)\in \textbf{B}_2(\alpha)\cup \textbf{B}_3(\alpha)$, then $L^{\alpha}(\delta,0,1)> m$ for $\delta\ge 2$.

(a) If $(p,q,d)\in \textbf{B}_1(\alpha)\cup \textbf{B}_4(\alpha)$, then the function $F(p,q,d,\alpha)\le 0$ i.e., $m\ge \sum_{i=0}^{\infty} (\alpha p)^i$. We want to show that $L^{\alpha}(\delta,0,1)\le m$. 

%We prove part (a) according to the decision information of the state $(\delta,0,0)$. 
Suppose that $u_\delta$ is the optimal decision of state $(\delta,0,0)$, i.e., the value function $J^{\alpha}(\delta,0,0)=Q^{\alpha}(\delta,0,0,u_\delta)$.
For all given $\delta$, 
\begin{equation}
    \begin{split}
    &   J^{\alpha}(\delta+1,0,0)-J^{\alpha}(\delta,0,0)\\ =& Q^{\alpha}(\delta+1,0,0,u_{\delta+1})-Q^{\alpha}(\delta,0,0,u_\delta)\\
    =& \underbrace{Q^{\alpha}(\delta+1,0,0,u_{\delta+1}) - Q^{\alpha}(\delta+1,0,0,u_\delta)}_{\le 0 (\text{, by optimality})}\\ & +  Q^{\alpha}(\delta+1,0,0,u_{\delta}) - Q^{\alpha}(\delta,0,0,u_{\delta})\\
    \le & Q^{\alpha}(\delta+1,0,0,u_{\delta}) - Q^{\alpha}(\delta,0,0,u_\delta)=L^{\alpha}(\delta+1,0,u_\delta).
    \end{split}\label{l01}
\end{equation} Thus, \eqref{A00} and \eqref{l01} gives
\begin{equation}
    L^{\alpha}(\delta,0,1)\le 1+\alpha p L^{\alpha}(\delta+1,0,u_\delta).\label{l011}
\end{equation}
Given age $\delta_0$, there are two possible cases for the optimal decision when $\delta>\delta_0$.  

Case (a$1$) For some non-negative integer $l$, we have $u_{\delta_0}=u_{\delta_0+1}=...=u_{\delta_0+l-1}=1$ and $u_{\delta_0+l}=2$.

%There exists a non-negative $l$, such that $a_{\delta_0+l}=2$. We suppose that $l$ is the smallest value such that $a_{\delta+l}=2$. 

In this case, if $l=0$, then $u_{\delta_0}=2$. From Lemma \ref{g2m}, we get $L^{\alpha}(\delta_0+1,0,2)=m$. Also, $(p,q,d)\in \textbf{B}_1(\alpha)\cup \textbf{B}_4(\alpha)$ implies that $m\ge \sum_{i=0}^{\infty} (\alpha p)^i$. From Lemma \ref{smalllemma}, if $m\ge \sum_{i=0}^{\infty} (\alpha p)^i$, then we have $1+(\alpha p)m\le m$. Combining these with \eqref{l011}, we get \begin{equation}
  L^{\alpha}(\delta_0,0,1)\le 1+(\alpha p)m\le m.   
\end{equation}
If $l>0$, then $u_{\delta_0}=...=u_{\delta_0+l-1}=1$. Thus, we can expand $L^{\alpha}(\delta_0+1,0,u_{\delta_0}),...L^{\alpha}(\delta_0+l,0,u_{\delta_0+l-1})$ iteratively using \eqref{l011} and get
\begin{equation}\begin{split}
 L^{\alpha}(\delta_0,0,1)\le \sum_{i=0}^{l}(\alpha p)^i +(\alpha p)^{l+1} L^{\alpha}(\delta_0+l+1,0,u_{\delta_0+l}).  \end{split}\label{proof4}  
\end{equation} %Note that \eqref{proof4} also holds for $l=0$.
Since $u_{\delta_0+l}=2$, Lemma \ref{g2m} implies that $L^{\alpha}(\delta_0+l+1,0,2)=m$. By Lemma \ref{smalllemma}, we get  \begin{equation}\begin{split}
& L^{\alpha}(\delta_0,0,1)  \le \sum_{i=0}^{l}(\alpha p)^i+(\alpha p)^{l+1} m\le m.  \end{split}\label{proof7}  
\end{equation} 

Case (a$2$) For all $l\ge0$, we have $u_{\delta_0+l}=1$. Then, we can use \eqref{l011} iteratively. Thus,
\eqref{proof4} holds for all the value $l$. 

Since the optimal decision $u_{\delta_0+l}=1$, we take \eqref{l011} into \eqref{proof4}, and get
\begin{equation}\label{proof5}\begin{split}
  &  \sum_{i=0}^{l}(\alpha p)^i +(\alpha p)^{l+1} L^{\alpha}(\delta_0+l+1,0,u_{\delta_0+l}) \\  \le & \sum_{i=0}^{l+1}(\alpha p)^i +(\alpha p)^{l+2} L^{\alpha}(\delta_0+l+2,0,u_{\delta_0+l+1}).\end{split}
\end{equation}
Thus, the right hand side of \eqref{proof4} is an increasing sequence in $l$. Then in order to prove $L^{\alpha}(\delta_0,0,1)\le m$, we want to show that the supremum limit of the sequence over $l$ is less than or equal to $m$. To prove this, we will show that the tail term of \eqref{proof4}, which is $(\alpha p)^{l+1} L^{\alpha}(\delta_0+l+1,0,u_{\delta_0+l})$, vanishes.

 Lemma \ref{monotonicity} implies that the value function $J^{\alpha}(\delta,l_1,l_2)$ is increasing in $\delta$. Equation \eqref{finiteproof} in the proof of Lemma \ref{finite} gives $ J^{\alpha}(\delta,l_1,l_2) \le 
    (\delta+\alpha/(1-\alpha))/(1-\alpha)$, which is linear on the age $\delta$. Thus, we get    
\begin{equation}\begin{split}\label{proof11}
 0 &\le L^{\alpha}(\delta_0+l+1,0,1)\\ &= J^{\alpha}(\delta_0+l+1,0,0)-J^{\alpha}(\delta_0+l,0,0)\\ & \le J^{\alpha}(\delta_0+l+1,0,0) \le \frac{(\delta_0+l+1+\frac{\alpha}{1-\alpha})}{1-\alpha}. \end{split} 
\end{equation}
From \eqref{proof11} and $\alpha,p<1$, we get 
\begin{equation}\label{proof10}
    \lim_{l\rightarrow \infty} (\alpha p)^{l+1} L^{\alpha}(\delta_0+l+1,0,1)=0.
\end{equation}
 Thus, we give \begin{equation}\begin{split}
 L^{\alpha}(\delta_0,0,1) & \le   \limsup_{l\rightarrow \infty} \sum_{i=0}^{l}(\alpha p)^i +(\alpha p)^{l+1} L^{\alpha}(\delta_0+l+1,0,1) \\ & = \lim_{l\rightarrow \infty}\sum_{i=0}^{l}(\alpha p)^i = \sum_{i=0}^{\infty}(\alpha p)^i. \end{split} \label{proof9}
\end{equation}  Part (a) implies that $m\ge \sum_{i=0}^{\infty}(\alpha p)^i$.
Thus, \eqref{proof9} directly gives $L^{\alpha}(\delta_0,0,1)\le \sum_{i=0}^{\infty}(\alpha p)^i\le m$. In conclusion, for both cases (a$1$) and (a$2$), we have \begin{equation}
L^{\alpha}(\delta_0,0,1)\le m.  
\end{equation} 

(b) If $(p,q,d)\in \textbf{B}_2(\alpha)\cup \textbf{B}_3(\alpha)$, then $F(p,q,d,\alpha)>0$, i.e., $m< \sum_{i=0}^{\infty} (\alpha p)^i$. Thus, we want to show that $L^{\alpha}(\delta,0,1)>m$ for all age $\delta$.
The proof of (b) is similar to (a), by reversing the inequalities and a slight change of \eqref{101}. We use the same definition of $u_\delta$ in part (a), assuming that $J^{\alpha}(\delta,0,0)=Q^{\alpha}(\delta,0,0,u_\delta)$. We get
\begin{equation}\label{delta001}
    \begin{split}
    &   J^{\alpha}(\delta+1,0,0)-J^{\alpha}(\delta,0,0)\\ =& Q^{\alpha}(\delta+1,0,0,u_{\delta+1})-Q^{\alpha}(\delta,0,0,u_\delta)\\
    =& Q^{\alpha}(\delta+1,0,0,u_{\delta+1}) - Q^{\alpha}(\delta,0,0,u_{\delta+1})\\ & + \underbrace{Q^{\alpha}(\delta,0,0,u_{\delta+1}) - Q^{\alpha}(\delta,0,0,u_\delta)}_{\ge 0 \text{, by optimality}}\\
    \ge & Q^{\alpha}(\delta+1,0,0,u_{\delta+1}) - Q^{\alpha}(\delta,0,0,u_{\delta+1})\\ = & L^{\alpha}(\delta+1,0,u_{\delta+1}).
    \end{split}
\end{equation} From \eqref{delta001} and \eqref{A00}, we can directly get  
\begin{equation}\label{100}
    \begin{split}
        L^{\alpha}(\delta,0,1)\ge 1+\alpha p L^{\alpha}(\delta+1,0,u_{\delta+1}).
    \end{split}
\end{equation}
Like in part (a), we split part (b) into two different cases:

Case (b$1$) For some non-negative integer $l$, we have $u_{\delta_0+1}=...=u_{\delta_0+l}=1$ and $u_{\delta_0+l+1}=2$.
%Suppose that $l$ is the smallest value such that $a_{\delta_0+l+1}=2$. 
Similar to \eqref{proof4}, by the iteration in \eqref{100}, 
\begin{equation}\label{101}\begin{split}
 L^{\alpha}(\delta_0,0,1)\ge \sum_{i=0}^{l}(\alpha p)^i +(\alpha p)^{l+1} L^{\alpha}(\delta_0+l+1,0,u_{\delta_0+l+1}).  \end{split}
\end{equation} Using Lemma \ref{smalllemma} (b), and $u_{\delta_0+l+1}=2$, we can get $L^{\alpha}(\delta_0,0,1)\ge \sum_{i=0}^{l}(\alpha p)^i +(\alpha p)^{l+1} m > m$.

Case (b$2$) The optimal decision $u_{\delta_0+l+1}=1$ for all $l\ge 0$. Then \eqref{101} holds for all non-negative $l$. Also, similar to \eqref{proof5}, the right hand side of \eqref{101} is decreasing in $l$. Thus, $L^{\alpha}(\delta_0,0,1)$ is larger than the infimum limit of the sequence over $l$. From \eqref{proof10}, and $m<\sum_{i=0}^{\infty}(\alpha p)^i$, we get \begin{equation}\begin{split}
 L^{\alpha}(\delta_0,0,1) & \ge   \liminf_{l\rightarrow \infty} \sum_{i=0}^{l}(\alpha p)^i +(\alpha p)^{l+1} L^{\alpha}(\delta_0+l+1,0,1) \\ & = \lim_{l\rightarrow \infty}\sum_{i=0}^{l}(\alpha p)^i = \sum_{i=0}^{\infty}(\alpha p)^i>m. \end{split} 
\end{equation}   
 Thus, the proof of Lemma \ref{l0m} is complete.

\section{proof of Lemma \ref{smalllemma}}\label{smalllemmaapp} 

(a) If $m'> \sum_{i=0}^{\infty} (\alpha p)^i$, we will show that $m'> \sum_{i=0}^{k-1}(\alpha p)^i+(\alpha p)^k m'$ for all $k\in \{1,2...\}$. We prove by using induction.

Suppose that $k=1$. Since $m'> \sum_{i=0}^{\infty} (\alpha p)^i=1/(1-\alpha p)$, then $(1-\alpha p)m'\ge 1$, and we get $m'> 1+(\alpha p)m'$. So, the condition holds for $k=1$. 

Suppose that the condition holds for $k=n$, then we will show that it holds for $k=n+1$. Since we have shown that $m'> 1+(\alpha p) m'$, the hypothesis inequality becomes
\begin{equation}
  m' > \sum_{i=0}^{n-1}(\alpha p)^i+(\alpha p)^n (1+(\alpha p) m')=\sum_{i=0}^{n}(\alpha p)^i+(\alpha p)^{n+1} m'.  
\end{equation}
 Thus, the condition holds for $k=n+1$.  

(b) If $m'< \sum_{i=0}^{\infty} (\alpha p)^i$, the proof is same with that of (a) except replacing notation '$>$' by '$<$'.

(c) If $m'= \sum_{i=0}^{\infty} (\alpha p)^i$, then we have for all $k\in\{1,2,...\}$, 
\begin{equation}\begin{split}
  m' & =\sum_{i=0}^{k-1} (\alpha p)^i +  \sum_{i=k}^{\infty} (\alpha p)^i  
  = \sum_{i=0}^{k-1} (\alpha p)^i + (\alpha p)^{k} \sum_{i=0}^{\infty} (\alpha p)^i \\
 & = \sum_{i=0}^{k} (\alpha p)^i + (\alpha p)^{k+1} m'. \end{split}
\end{equation}

Thus, we complete the proof of Lemma \ref{smalllemma}.

\section{proof of Lemma \ref{l1m}}\label{l1mapp}

%The main idea of the proof includes two steps: $(1)$, comparing $L^{\alpha}(\delta,1,1)$ with $m$ for $\delta>s$, where $s$ is a given age value, and $(2)$, showing the policy $\mu(\delta,1,0)$ is a constant for $\delta\le s$. 

%To compare the increasing difference function $L(\delta,1,1)$ with $m$, 
Lemma \ref{g2m} implies that: Showing that $L^{\alpha}(\delta,l_1,1)\le m$ for $\delta>s$ is sufficient to show that $Q^{\alpha}(\delta,l_1,0,u)$ is supermodular in $(\delta,u)$ for $\delta>s$. Conversely, showing that $L^{\alpha}(\delta,l_1,1)> m$ for $\delta>s$  is sufficient to show that $Q^{\alpha}(\delta,l_1,0,u)$ is supermodular in $(\delta,u)$ for $\delta>s$. Thus, it remains to prove the following statements:

($s1$) If $(p,q,d)\in \textbf{B}_1(\alpha)\cup \textbf{B}_2(\alpha)$, 
then there exists a positive integer $s$, such that $L^{\alpha}(\delta,1,1)\le m$ for $\delta > s$, and $\mu^{\alpha,*}(\delta,l_1,0)$ is constant for all $\delta \le s$. %Thus, $\mu(\delta,1,0)$ is nonincreasing.  

($s2$) If $(p,q,d)\in \textbf{B}_3(\alpha)\cup  \textbf{B}_4(\alpha)$, then there exists a positive integer $s$, such that $L^{\alpha}(\delta,1,1)>m$ for $\delta > s$, and $\mu^{\alpha,*}(\delta,l_1,0)$ is constant for all $\delta \le s$.

We first need to give three preliminary statements before the proof.  

$(1)$ We first need to give an expression of $L^{\alpha}(\delta,1,1)$.

 The state $(\delta,1,0)$ has a probability $q$ to decrease to state $(1,1,0)$ and a probability $1-q$ to be $(\delta+1,0,0)$. According to \eqref{discounted-lemma1}, we get 
\begin{equation}\label{A11}\begin{split}
    &Q^{\alpha}(\delta,1,0,1)\\=& \delta+\alpha (1-q) J^{\alpha}(\delta+1,0,0)+\alpha q J^{\alpha}(1,1,0), \end{split}\end{equation} Thus,
\begin{equation}
    \begin{split}
        L^{\alpha}(\delta,1,1) =1+\alpha (1-q)\Big{(}J^{\alpha}(\delta+1,0,0)-J^{\alpha}(\delta,0,0)\Big{)}.\label{A10}
    \end{split}
\end{equation} 
%According to \eqref{A10}, showing $L(\delta,1,1)\le m$ (or $>m$) is equivalent to showing\begin{equation}
%        1+\alpha (1-q)\Big{(}J^{\alpha}(\delta+1,0,0)-J^{\alpha}(\delta,0,0)\Big{)} \le m (\text{or}> m).
%\end{equation}.
%We know $L(\delta,1,1)$ is the age difference function with the state of Channel $1$ to be $1$. Equation \eqref{A10} tells that $L(\delta,1,1)$ depends on the decision information with the state of Channel $1$ to be $0$.

(2) We consider a special case when $J^{\alpha}(\delta+l,0,0)=Q^{\alpha}(\delta+l,0,0,1)$ for all non-negative $l$. Then, we have \begin{equation}
    J^{\alpha}(\delta+1,0,0)-J^{\alpha}(\delta,0,0)=L^{\alpha}(\delta+1,0,1)
\end{equation} . 
Recall that \begin{equation}\begin{split}
    & L^{\alpha}(\delta,0,1)= 1+\alpha p\Big{(}J^{\alpha}(\delta+1,0,0)-J^{\alpha}(\delta,0,0)\Big{)} \label{A00prime}\end{split}.\end{equation}
Then, \eqref{A00prime} gets
\begin{equation}\label{proof13}
L^{\alpha}(\delta,0,1)=1+\alpha p L^{\alpha}(\delta+1,0,1). 
\end{equation}
By iterating the \eqref{proof13} on $\delta+1,\delta+2,...$, we get for all non-negative $l$,  
\begin{equation}\label{proof13_1}
   L^{\alpha}(\delta,0,1) = \sum_{i=0}^{l-1} (\alpha p)^i + (\alpha p)^{l} L^{\alpha}(\delta+l+1,0,1). 
\end{equation} Equation \eqref{proof10} implies that $(\alpha p)^{l} L^{\alpha}(\delta+l+1,0,1)$ vanishes as $l$ goes to infinity. After taking the limit of $l$, our conclusion is that if $J^{\alpha}(\delta+l,0,0)=Q^{\alpha}(\delta+l,0,0,1)$ for all non-negative $l$, for all age $\delta$, \begin{equation}\label{proof09}
   L^{\alpha}(\delta,0,1) = \sum_{i=0}^{\infty} (\alpha p)^i. 
\end{equation}  
 %Next, we want to show the policy $\mu(\delta,1,0)$ is a constant for age $\delta\le s$. One method is to build the relationship between $\mu(\delta,1,0)$ and $\mu(\delta,0,0)$. The relationship
$(3)$ The threshold $s$ mentioned in Lemma \ref{l1m} depends on whether Channel $1$ is positive-correlated or negative-correlated. So, we will utilize Lemma \ref{lemma-switch} in Appendix \ref{lemma-switchapp}.
%depends on whether the Channel $1$ is positive-correlated or negative-correlated. So we need to introduce a lemma for the switch-type criteria:

%\begin{lemma}\label{lemma-switch}
%If Channel $1$ is positive-correlated, i.e., $p+q\ge 1$, and $\mu^{\alpha,*}(\delta,0,0)=1$, then $\mu^{\alpha,*}(\delta,1,0)=1$.
%Conversely, if Channel $1$ is negative-correlated, i.e. $p+q\le 1$, and $\mu^{\alpha,*}(\delta,0,0)=2$, then $\mu^{\alpha,*}(\delta,1,0)=2$.
%\end{lemma}
%\begin{proof}
%Please see Appendix \ref{lemma-switchapp} for details.
%\end{proof}

After introducing the three statements, we start our proof of Lemma \ref{l1m}. The proof is divided into four parts: (a), (b), (c) and (d). Parts (a) and (b) are dedicated to prove part ($s1$) that gives Lemma \ref{l1m} (a), and parts (c) and (d) are dedicated to prove part ($s2$) that gives Lemma \ref{l1m} (b).

(a) If $(p,q,d)\in \textbf{B}_1(\alpha) $, then we have $1+\alpha (1-q)\sum_{i=0}^{\infty} (\alpha p)^i \le m$ and $\sum_{i=0}^{\infty} (\alpha p)^i\le m$. Our objective is: there \emph{exists} a value $s$, such that the function $L^{\alpha}(\delta,1,1)\le m$ for $\delta>s$, and the optimal decisions $\mu^{\alpha,*}(\delta,1,0)$ is a constant for $\delta\le s$.
The choice of $s$ depends on two cases: $p+q\ge 1$ or $p+q< 1$. If $p+q\ge 1$, we will take $s=1$. If $p+q< 1$, We will take $s$ to be the threshold of $\mu^{\alpha,*}(\delta,0,0)$.  %Also, from Lemma \ref{l0m}, $\mu(\delta,0,0)$ is decreasing. 

Case (a1) Suppose that $p+q\ge 1$. Thus, by comparing \eqref{A00} with \eqref{A10}, we get 
$L^{\alpha}(\delta,1,1)\le L^{\alpha}(\delta,0,1)$. Lemma \ref{l0m} (a) implies that $L^{\alpha}(\delta,0,1)\le m$. Thus, $  L^{\alpha}(\delta,1,1) \le m$ for all the age $\delta>1$. Thus, we take $s=1$, and our objective holds.

% Thus, $Q^{\alpha}(\delta,1,0)$ is  supermodularity, so, $\mu(\delta,1,0)$ is nonincreasing.  

Case (a2) Suppose that $p+q< 1$. Lemma \ref{l0m} (a) implies that $\mu^{\alpha,*}(\delta,0,0)$ is non-increasing. Then we take $s$ to be the threshold of $\mu^{\alpha,*}(\delta,0,0)$.
Then, $\mu^{\alpha,*}(\delta,0,0)=2$ for $\delta\le s$. Lemma \ref{lemma-switch} implies that $\mu^{\alpha,*}(\delta,1,0)=2$ for $\delta\le s$. 
Also, $\mu^{\alpha,*}(\delta,0,0)= 1$ for $\delta> s$. So, \eqref{proof09} implies that 
  $ L^{\alpha}(\delta+1,0,1)=\sum_{i=0}^{\infty} (\alpha p)^i$ for $\delta\ge s$. 
 From \eqref{A10}, for $\delta > s$, 
\begin{equation}
 L^{\alpha}(\delta,1,1)=1+\alpha (1-q)L^{\alpha}(\delta+1,0,1)=1+\alpha (1-q)\sum_{i=0}^{\infty} (\alpha p)^i.   \label{proof12}
\end{equation}
Thus, the first condition in part (a) implies that
$L^{\alpha}(\delta,1,1)\le m$ for $\delta>s$. %Thus, $\mu(\delta,1,0)$ is nonincreasing in $\delta\ge s$. Since the decision $\mu(\delta,1,0)=2$ for $\delta\le s$, so overall $\mu(\delta,1,0)$ is nonincreasing for all $\delta\ge 1$. 
By combining both $p+q\ge 1$ and $p+q< 1$ in Case (a1) and Case (a2) respectively, we complete the proof when $(p,q,d)\in \textbf{B}_1(\alpha) $.

%by utilizing \eqref{proof4} and Lemma \ref{smalllemma}
(b) Suppose that $(p,q,d)\in \textbf{B}_2(\alpha)$. Similar to (a), our objective is to show that there exists a value $s$ such that $L^{\alpha}(\delta,1,1)\le m$ for $\delta>s$, and $\mu^{\alpha,*}(\delta,1,0)$ is a constant for $\delta\le s$.

Since the system parameters  $(p,q,d)\in \textbf{B}_2(\alpha) $, we have $1+\alpha (1-q)m \le m$ and $\sum_{k=0}^{\infty} (\alpha p)^k> m$. This implies $p+q\ge 1$. Also, Lemma \ref{l0m} (b) implies that $\mu^{\alpha,*}(\delta,0,0)$ is non-decreasing. Then we take $s$ to be the threshold of $\mu^{\alpha,*}(\delta,0,0)$. Then, $\mu^{\alpha,*}(\delta,0,0)=1$ for $\delta\le s$, and Lemma \ref{lemma-switch} implies that $\mu^{\alpha,*}(\delta,1,0)=1$ for $\delta\le s$. Also, $\mu^{\alpha,*}(\delta,0,0)=2$ for $\delta> s$. 
 Thus, $J^{\alpha}(\delta+1,0,0)-J^{\alpha}(\delta,0,0)=L^{\alpha}(\delta+1,0,2)$ for $\delta>s$. Lemma \ref{g2m} implies that  $L^{\alpha}(\delta+1,0,2)=m$ for $\delta>s$. Thus, from \eqref{A10},
we get $L^{\alpha}(\delta,1,1)= 1+\alpha (1-q)m$. From the condition in part (b), $1+\alpha(1-q)m\le m$. Thus, $L^{\alpha}(\delta,1,1)\le m$ for $\delta>s$, and we complete the proof of our objective when $(p,q,d)\in \textbf{B}_2(\alpha)$. %this is supermodularity on $Q^{\alpha}(\delta,1,0,u)$ for all $\delta> s$, thus $\mu(\delta,1,0)$ is nonincreasing in $\delta\ge s$. Overall, $\mu(\delta,1,0)$ is nonincreasing for all $\delta\ge 1$. 

(c) The case $(p,q,d)\in \textbf{B}_3(\alpha) $ has a similar proof to part (a) where $(p,q,d)\in \textbf{B}_1(\alpha)$. Our objective is to show that there exists a value $s$ such that $L^{\alpha}(\delta,1,1)> m$ for $\delta>s$, and $\mu^{\alpha,*}(\delta,1,0)$ is a constant for $\delta\le s$.
We will take $s=1$ if $p+q<1$.  Lemma \ref{l0m} (b) implies that $\mu^{\alpha,*}(\delta,0,0)$ is non-decreasing threshold type. So, we will take $s$ to be the threshold of $\mu^{\alpha,*}(\delta,0,0)$ if $p+q\ge 1$.

Note that the system parameters $(p,q,d)\in \textbf{B}_3(\alpha)$ implies $1+\alpha (1-q)m > m$ and $\sum_{i=0}^{\infty} (\alpha p)^i> m$.

Case (c$1$) Suppose that $p+q< 1$. Similar to the proof of part $(a1)$, we compare \eqref{A00} with \eqref{A10}, and we get 
$L^{\alpha}(\delta,1,1)> L^{\alpha}(\delta,0,1)$. Lemma \ref{l0m} (b) implies that $L^{\alpha}(\delta,0,1)> m$. Thus, $  L^{\alpha}(\delta,1,1) > m$ for $\delta>1$. Thus, we take $s=1$, and our objective holds.    

Case (c$2$) Suppose that $p+q\ge 1$. We take $s$ to be the threshold of non-decreasing $\mu^{\alpha,*}(\delta,0,0)$. Then, $\mu^{\alpha,*}(\delta,0,0)=1$ for $\delta\le s$. Thus, Lemma \ref{lemma-switch} implies that $\mu^{\alpha,*}(\delta,1,0)=1$.
Also, $\mu^{\alpha,*}(\delta,0,0)=2$ for $\delta>s$, same with part (b), $L^{\alpha}(\delta,1,1)=1+\alpha (1-q)m \ge m$, which proves our objective. By combining both Case (c$1$) and Case (c$2$) respectively, we complete the proof when $(p,q,d)\in \textbf{B}_3(\alpha) $.

(d) The case $(p,q,d)\in \textbf{B}_4(\alpha) $ has a similar proof to part (b) where $(p,q,d)\in \textbf{B}_2(\alpha)$. 
Our objective is to show that there exists a value $s$ such that $L^{\alpha}(\delta,1,1)> m$ for $\delta>s$, and $\mu^{\alpha,*}(\delta,1,0)$ is a constant for $\delta\le s$.

The case $(p,q,d)\in \textbf{B}_4(\alpha)$ gives $1+\alpha (1-q)\sum_{i=0}^{\infty} (\alpha p)^i > m$ and $\sum_{i=0}^{\infty} (\alpha p)^i \le m$. These 2 conditions imply that $p+q\le 1$. Lemma \ref{l0m} (a) implies that $\mu^{\alpha,*}(\delta,0,0)$ is non-increasing threshold type. Then we take $s$ to be the threshold of $\mu^{\alpha,*}(\delta,0,0)$.
So, $\mu^{\alpha,*}(\delta,0,0)=2$ for $\delta\le s$, and Lemma \ref{lemma-switch} implies that $\mu^{\alpha,*}(\delta,1,0)=2$ for $\delta\le s$.
Also, $\mu^{\alpha,*}(\delta,0,0)=1$ for $\delta>s$.
Thus, \eqref{proof12} in proof of (a2) still holds for $\delta>s$. Since $1+\alpha (1-q)\sum_{i=0}^{\infty} (\alpha p)^i > m$, \eqref{proof12} directly implies that $L^{\alpha}(\delta,1,1)> m$ for all $\delta>s$. Thus, we complete the proof of our objective when $(p,q,d)\in \textbf{B}_4(\alpha)$.

\section{proof of Lemma \ref{lemma-switch-avg}}\label{lemma-switchapp}

According to Lemma \ref{theorem1}, it is sufficient to show that for all $\alpha<1$, the following lemma holds.  
\begin{lemma}\label{lemma-switch}
If Channel $1$ is positive-correlated, i.e., $p+q\ge 1$, and $\mu^{\alpha,*}(\delta,0,0)=1$, then $\mu^{\alpha,*}(\delta,1,0)=1$.
Conversely, if Channel $1$ is negative-correlated, i.e. $p+q\le 1$, and $\mu^{\alpha,*}(\delta,0,0)=2$, then $\mu^{\alpha,*}(\delta,1,0)=2$.
\end{lemma}

We start the proof of Lemma \ref{lemma-switch}.

First of all, since both  $Q^{\alpha}(\delta,0,0,2)-Q^{\alpha}(\delta,0,0,1)$ and $Q^{\alpha}(\delta,1,0,2)-Q^{\alpha}(\delta,1,0,1)$ will induce a term $J^{\alpha}(\delta+1,1,d-1)-J^{\alpha}(1,1,0)$, we need to provide a lemma: 
\begin{lemma} We have 
$J^{\alpha}(\delta+1,1,d-1)\ge J^{\alpha}(1,1,0)$. \label{11d110}
\end{lemma}
\begin{proof}
Please see Appendix \ref{appendix11d110} for details.
\end{proof}
Then we start the proof. 

(a) Suppose that $p+q\ge 1$ and $\mu^{\alpha,*}(\delta,0,0)=1$. Thus,
\begin{equation}
  Q^{\alpha}(\delta,0,0,2)\ge Q^{\alpha}(\delta,0,0,1).  
\end{equation}
Recall that \eqref{A01}, \eqref{qdelta02} give the expression of $Q^{\alpha}(\delta,0,0,1),\\ Q^{\alpha}(\delta,0,0,2)$ respectively. We get 
\begin{equation}\begin{split}
   & Q^{\alpha}(\delta,0,0,2)-Q^{\alpha}(\delta,0,0,1)\\ = &\alpha p \Big{(}J^{\alpha}(\delta+1,0,d-1)-J^{\alpha}(\delta+1,0,0)\Big{)}\\ & + \alpha (1-p)\Big{(}J^{\alpha}(\delta+1,1,d-1)-J^{\alpha}(1,1,0)\Big{)}.\label{switchtype0}\end{split}
\end{equation} Then we want to show that $Q^{\alpha}(\delta,1,0,2)\ge Q^{\alpha}(\delta,1,0,1)$. Note that 
\begin{equation}
    \begin{split}
          & Q^{\alpha}(\delta,1,0,2)-Q^{\alpha}(\delta,1,0,1)\\ = &\alpha (1-q) \Big{(}J^{\alpha}(\delta+1,0,d-1)-J^{\alpha}(\delta+1,0,0)\Big{)}\\ & + \alpha q\Big{(}J^{\alpha}(\delta+1,1,d-1)-J^{\alpha}(1,1,0)\Big{)}.\label{switchtype1}
    \end{split}
\end{equation}
%From Lemma \ref{monotonicity},  $J^{\alpha}(\delta+1,1,d)\ge J^{\alpha}(2,1,d)$. Also 
%From Lemma \ref{11d110}, we have $J^{\alpha}(\delta+1,1,d-1)\ge J^{\alpha}(1,1,0)$. 
For the first terms in \eqref{switchtype1}, we have two possible cases:

Case (a$1$) Suppose that $J^{\alpha}(\delta+1,0,d-1)-J^{\alpha}(\delta+1,0,0)\ge 0$. From Lemma \ref{11d110}, we have $J^{\alpha}(\delta+1,1,d-1)\ge J^{\alpha}(1,1,0)$. Thus, \eqref{switchtype1} implies that $Q^{\alpha}(\delta,1,0,2)-Q^{\alpha}(\delta,1,0,1)\ge 0$. 

Case (a$2$) Suppose that $J^{\alpha}(\delta+1,0,d-1)-J^{\alpha}(\delta+1,0,0)< 0$. Since $q\ge 1-p$, then, \eqref{switchtype0} and \eqref{switchtype1} imply that 
\begin{equation}
    \begin{split}
 & Q^{\alpha}(\delta,1,0,2)-Q^{\alpha}(\delta,1,0,1)  \\
 \ge & Q^{\alpha}(\delta,0,0,2)-Q^{\alpha}(\delta,0,0,1) \ge 0.
    \end{split}
\end{equation}

%compared with \eqref{switchtype0}, Equation \eqref{switchtype1} is adding more positive term and less negative term, thus, \eqref{switchtype1} is still $\ge 0$.

(b) Suppose that $p+q\le 1$ and $\mu^{\alpha,*}(\delta,0,0)=2$. Then \eqref{switchtype0} is negative. Therefore, $J^{\alpha}(\delta+1,0,d-1)-J^{\alpha}(\delta+1,0,0)$ must be negative. Then, \eqref{switchtype0} and \eqref{switchtype1} imply that \begin{equation}
    \begin{split}
 & Q^{\alpha}(\delta,1,0,2)-Q^{\alpha}(\delta,1,0,1)  \\
 \le & Q^{\alpha}(\delta,0,0,2)-Q^{\alpha}(\delta,0,0,1) \le 0.
    \end{split}
\end{equation}

By considering (a) and (b), we have completed the proof.

\section{proof of Lemma \ref{11d110}}\label{appendix11d110}
%Recall that \eqref{matrix1} defines the two sequences $a'_k,b'_k$: $[a'_k,b'_k]=[1,0]\times P^k$, where $P$ is the transition matrix of Channel $1$.  
First, when $d=2$, $J^{\alpha}(\delta+1,1,1)$ is expanded according to \eqref{equation04}, and we have 
\begin{equation}\begin{split}
 & J^{\alpha}(\delta+1,1,1)-J^{\alpha}(1,1,0) \\ \ge & J^{\alpha}(\delta+1,1,1)-Q^{\alpha}(1,1,0,1) \\ = & \delta+ \alpha a'_1 \Big{(}J^{\alpha}(d,1,0)-J^{\alpha}(1,1,0)\Big{)}\\  + & \alpha b'_1\Big{(}J^{\alpha}(d,0,0)-J^{\alpha}(d,0,0)\Big{)} \ge 0
  \end{split}
\end{equation} Thus, we only need to consider $d\ge 3$ in this proof.

Then, we will use the similar technique that is used in the proof of Lemma \ref{g2m}, to show the following inequality holds: %for $k=1,2,...,d-1$, 
\begin{equation}\label{add}
    \begin{split}
   & J^{\alpha}(\delta+1,1,d-1)-J^{\alpha}(1,1,0)\\ \ge & 
  \alpha^{d-2} a'_{d-2}\Big{(}J^{\alpha}(\delta+d-1,1,1)-J^{\alpha}(1,1,0)\Big{)}\\  & +  \alpha^{d-2} b'_{d-2}\Big{(}J^{\alpha}(\delta+d-1,0,1)-J^{\alpha}(d-1,0,0)\Big{)},
    \end{split}
\end{equation}
%\begin{equation}\label{switch01}
%    \begin{split}
%       & J^{\alpha}(\delta+1,1,d-1)-J^{\alpha}(1,1,0)\\ \ge &  a'_{k-1}\Big{(}J^{\alpha}(\delta+k,1,d-k)-J^{\alpha}(1,1,0)\Big{)}\\  + & b'_{k-1}\Big{(}J^{\alpha}(\delta+k,0,d-k)-J^{\alpha}(k,0,0)\Big{)}.
%    \end{split}
%\end{equation}
where $a'_{d-2}$, $b'_{d-2}$ are defined in \eqref{matrix1}.
\begin{proof}

Note that the optimal decision of $(\delta+1,l_1,d-1)$ is \emph{none} and $J^{\alpha}(\delta+1,l_1,d-1)$ is expanded according to \eqref{equation06}. Also, $J^{\alpha}(1,1,0)\le Q^{\alpha}(1,1,0,1)$ and $Q^{\alpha}(1,1,0,1)$ is expanded according to \eqref{A11}. We get     
\begin{equation}\label{equation110}
    \begin{split}
       & J^{\alpha}(\delta+1,1,d-1)-J^{\alpha}(1,1,0) \\
       \ge & J^{\alpha}(\delta+1,1,d-1)-Q^{\alpha}(1,1,0,1)
       \\ = & \delta+\alpha a'_{1}\Big{(}J^{\alpha}(\delta+2,1,d-2)-J^{\alpha}(1,1,0)\Big{)}\\  & + \alpha b'_{1}\Big{(}J^{\alpha}(\delta+2,0,d-2)-J^{\alpha}(2,0,0)\Big{)}\\
        \ge & \alpha a'_{1}\Big{(}J^{\alpha}(\delta+2,1,d-2)-J^{\alpha}(1,1,0)\Big{)}\\   & + \alpha b'_{1}\Big{(}J^{\alpha}(\delta+2,0,d-2)-J^{\alpha}(2,0,0)\Big{)},
    \end{split}
\end{equation} where $a'_1=q$ and $b'_1=1-q$ as defined in \eqref{matrix1}.
The optimal decision of $(\delta+2,l_1,d-2)$ in \eqref{equation110} is \emph{none} and $J^{\alpha}(\delta+2,l_1,d-2)$ is expanded similar to \eqref{equation06} according to the following: 
\begin{equation}\label{equation06add}\begin{split}
  & J^{\alpha}(\delta_0,0,d-2) \\ = & \delta_0 + \alpha p J^{\alpha}(\delta_0+1,0,d-3)+ \alpha (1-p) J^{\alpha}(\delta_0+1,1,d-3),  \\
  & J^{\alpha}(\delta_0,1,d-2) \\ = & \delta_0 + \alpha (1-q) J^{\alpha}(\delta_0+1,0,d-3)+ \alpha q J^{\alpha}(\delta_0+1,1,d-3),\end{split}\end{equation}
where $\delta_0$ is arbitrary.
Also,
\begin{equation}\begin{split}\label{equation210}
& J^{\alpha}(1,1,0) \le  Q^{\alpha}(1,1,0,1) \le Q^{\alpha}(2,1,0,1) \\ = & 1+\alpha q J^{\alpha}(1,1,0)+\alpha (1-q) J^{\alpha}(3,0,0).
    \end{split}
\end{equation}
Thus, \eqref{equation110},\eqref{equation06add} and \eqref{equation210} give 
\begin{equation}
    \begin{split}
      & J^{\alpha}(\delta+1,1,d-1)-J^{\alpha}(1,1,0)\\ \ge & \alpha^2 a'_{2}\Big{(}J^{\alpha}(\delta+2,1,d-3)-J^{\alpha}(1,1,0)\Big{)}\\  & + \alpha^2 b'_{2}\Big{(}J^{\alpha}(\delta+2,0,d-3)-J^{\alpha}(3,0,0)\Big{)}.
    \end{split}
\end{equation}

By using recursion for another $d-4$ times, we can finally get \eqref{add} (note that if $d=3$, then we have already proved \eqref{add} in \eqref{equation110}). 
%\begin{equation}\label{switch01}
%    \begin{split}
%      & J^{\alpha}(\delta+1,1,d-1)-J^{\alpha}(1,1,0)\\ \ge &  a'_{k-1}\Big{(}J^{\alpha}(\delta+k,1,d-k)-J^{\alpha}(1,1,0)\Big{)}\\  + & b'_{k-1}\Big{(}J^{\alpha}(\delta+k,0,d-k)-J^{\alpha}(k,0,0)\Big{)}.
%    \end{split}
%\end{equation} 
%Then, Lemma \ref{11d110} is shown directly by taking $k=d-1$.
%The initial $k=1$ is shown directly, because $a'_0=1$, $b'_0=0$, and $J^{\alpha}(\delta,1,d-1)-J^{\alpha}(1,1,0)$ is smaller or equal to itself. Suppose that \eqref{switch01} holds at $k$. On the case of $k+1$, the idea is to expand the terms in right hand side of \eqref{switch01}. The value function $J^{\alpha}(\delta+k,1,d-k)$ and $J^{\alpha}(\delta+k,0,d-k)$ are expanded according to \eqref{equation06}. The value function $J^{\alpha}(k,0,0)$ is upper bounded by $Q^{\alpha}(k,0,0,1)$, and then expanded in \eqref{A01}. The final term $J^{\alpha}(1,1,0)$ is upper bounded by $Q^{\alpha}(1,1,0,1)$, and then expanded in \eqref{A11}. Also, we upper bound $J^{\alpha}(2,0,0)$ by $J^{\alpha}(k+1,0,0)$ in \eqref{A11}. Utilizing the definition of the sequence in \eqref{matrix1}, we get \eqref{switch01} with $k+1$.   
\end{proof}

%Take $k=d-1$, so the state of Channel $2$ is $1$.
Now, we show Lemma \ref{11d110}. 
The value function $J^{\alpha}(\delta+d-1,l_1,1)$ in \eqref{add} is expanded according to \eqref{equation04}. Also, we have $J^{\alpha}(d-1,0,0)\le Q^{\alpha}(d-1,0,0,1)$. Similar to \eqref{equation210}, 
\begin{equation}
    J^{\alpha}(1,1,0)\le Q^{\alpha}(1,1,0,1)\le Q^{\alpha}(d-1,0,0,1).
\end{equation}
Thus, \eqref{add} gives:
\begin{equation}
    \begin{split}
   & J^{\alpha}(\delta+1,1,d-1)-J^{\alpha}(1,1,0)\\ \ge & 
 \alpha^{d-1}  a'_{d-1}\Big{(}J^{\alpha}(d,1,0)-J^{\alpha}(1,1,0)\Big{)}\\   & + \alpha^{d-1} b'_{d-1}\Big{(}J^{\alpha}(d,0,0)-J^{\alpha}(d,0,0)\Big{)}\ge 0,
    \end{split}
\end{equation}
%State $(\delta+1,1,d)$ will stay inactive for consecutive $d$ time slots. Also, after $d$ slots the state $(\delta+1,1,d)$ decreases to $(d+1,0,0)$ and $(d+1,1,0)$, with probability  $(a'_d,b'_d)=(0,1)\mathbb{P}^d$. Note that $a'_d+b'_d=1$. The cost triggered in the $i^{th}$ iteration is $\alpha ^{i-1} i$, $(i=1,2,...,d)$. Thus, 
%\begin{equation}
%    J^{\alpha}(\delta+1,1,d) = \sum_{i=1}^{d} \alpha^{i-1}(\delta+1) + a'_d J^{\alpha}(d+1,0,0)+b'_d J^{\alpha}(d+1,1,0)
%\end{equation}
%For each of the $d$ stages iteration starting from the value function $J^{\alpha}(1,1,0)$, for all possible states $\textbf{s}\in (1,1,0),(2,0,0)...(d+1,0,0)$, we always upper bound  $J^{\alpha}(\textbf{s})$ by $Q^{\alpha}(\textbf{s},1)$. Then after $d$ stages, $(1,1,0)$ turns into $(1,1,0)$ with probability $a'_d$, and probability of $b'_d$ combining $(2,0,0)$, $(3,0,0)$..., and $(1+d,0,0)$. From Lemma \ref{monotonicity}, $J^{\alpha}(1+d,0,0)\ge J^{\alpha}(\delta,0,0)$ for all $\delta=2,3...,d+1$. Also, the cost in $i^{th}$ each iteration is smaller than $\alpha ^{i-1} i$ (because the states that restart to $(1,1,0)$ in $i^{th}$ iteration will only trigger cost $\alpha^{i}\cdot 1$). Thus 
%\begin{equation}\begin{split}
%    J^{\alpha}(1,0,0)%& = 1+\alpha p J^{\alpha}(2,0,0)+\alpha(1-p)J^{\alpha}(1,1,0)\\
%    &\le \sum_{i=1}^{d} \alpha^{i-1}i + a'_d J^{\alpha}(d+1,0,0)+b'_d J^{\alpha}(1,1,0)\\ & \le J^{\alpha}(\delta+1,1,d) \end{split}
%\end{equation}
where $a'_{d-1},b'_{d-1}$ is described in \eqref{matrix1}. Thus, we complete the proof of Lemma \ref{11d110}.

\section{proof of Lemma \ref{mu00in1and4}}\label{mu00in1and4app}

Recall that we use $\mu^{\alpha,*} (\cdot)$ to denote the discounted problem's optimal decisions. From Lemma \ref{theorem1}, it is sufficient to show that: for all discount factor $\alpha<1$, the optimal decisions $\mu^{\alpha,*}(\delta,0,0)=1$ if $(p,q,d)\in \textbf{B}_1(\alpha)\cup \textbf{B}_4(\alpha)$. %For simplicity, we will use $\mu(\delta,0,0)$ instead of $\mu^\alpha(\delta,0,0)$. 
We use $\mu_j^{\alpha,*}(\delta,0,0)$ to denote the optimal decision of the state $(\delta,0,0)$ at $j^{th}$ iteration according to the value iteration \eqref{value-iterate}. From Lemma \ref{lemma1}(c), to prove that $\mu^{\alpha,*}(\delta,0,0)=1$ for all $\delta$, we will show that $\mu_j^{\alpha,*}(\delta,0,0)=1$ for all $\delta$ and the iteration $j$. We show this by using induction on $j$. 

The value function $J^{\alpha}_0(\cdot)=0$ and cost function is $\delta$ for both choices. Thus, for $j=0$, we directly get $\mu^{\alpha,*}_j(\delta,0,0)=1$.
%and .

Suppose that $\mu^{\alpha,*}_j(\delta,0,0)=1$ for $j=n-1\ge 0$, we will show that $\mu^{\alpha,*}_j(\delta,0,0)=1$ for $j=n$.
To show this, we need to show: 

(i) The supermodularity holds for all $\delta\ge 2$:
\begin{equation}\begin{split}
& Q^{\alpha}_n(\delta,0,0,1)-Q^{\alpha}_n(\delta-1,0,0,1) \\ \le  & Q^{\alpha}_n(\delta,0,0,2)-Q^{\alpha}_n(\delta-1,0,0,2). \end{split} \label{modularn}  
\end{equation}
Thus, from \cite{topkis1998supermodularity}, $\mu^{\alpha,*}_n(\delta,0,0)$ is non-increasing in $\delta$.

(ii) The optimal decision $\mu^{\alpha,*}_n(1,0,0)=1$, i.e., $Q^{\alpha}_n(1,0,0,2)\ge Q^{\alpha}_n(1,0,0,1)$. From (i) and (ii), the optimal decision $\mu^{\alpha,*}_n(\delta,0,0)$ is $1$ for all $\delta$.  

We first show (i). For simplicity we define the age difference function:\begin{equation}
    \begin{split}
&  L^{\alpha}_n(\delta,0,1)=Q^{\alpha}_n(\delta,0,0,1)-Q^{\alpha}_n(\delta-1,0,0,1), \\   & L^{\alpha}_n(\delta,0,2)=Q^{\alpha}_n(\delta,0,0,2)-Q^{\alpha}_n(\delta-1,0,0,2).  
    \end{split}
\end{equation}
We want to show that 
\begin{equation}\label{wantoshow}
    L^{\alpha}_n(\delta,0,1)\le L^{\alpha}_n(\delta,0,2).
\end{equation}
First, we derive    \begin{equation}\label{ln01}
         L^{\alpha}_n(\delta,0,1)=\sum_{i=0}^{n-1} (\alpha p)^i.
     \end{equation}
%We derive $L_n(\delta,0,1)=\sum_{i=0}^{n-1} (\alpha p)^i$ in \eqref{wantoshow} by providing the following lemma: 
%\begin{lemma}\label{l00_n}
%Suppose that $J^{\alpha}_0(\delta,l_1,l_2)=0$ for all the state $(\delta,l_1,l_2)$. For all $k$, if $\mu_0(\delta,0,0)=$ $...$, $=\mu_{k-1}(\delta,0,0)=1$ for all $\delta\ge 1$, then, for all $\delta\ge 2$,
%\begin{equation}
%   L_{k}(\delta,0,1) = \sum_{i=0}^{k-1} (\alpha p)^i. 
%\end{equation} 
%\end{lemma}
\begin{proof}
%With probability $p$, the transmission fails and the state $(\delta,0,0)$ will go to $(\delta+1,0,0)$. With probability $1-p$, the transmission is successful and the state will go to $(1,1,0)$. 
Similar to \eqref{A01}, we can get 
\begin{equation}\label{b00_n1}
    \begin{split}
     &Q^{\alpha}_{n}(\delta,0,0,1)\\=& \delta+\alpha pJ^{\alpha}_{n-1}(\delta+1,0,0)+\alpha (1-p)J^{\alpha}_{n-1}(1,1,0),      
    \end{split}
\end{equation} thus, 
\begin{equation}\label{modular1}
    \begin{split}
     L^{\alpha}_n(\delta,0,1) & = 1+\alpha p \big{(} J^{\alpha}_{n-1}(\delta+1,0,0)-J^{\alpha}_{n-1}(\delta,0,0) \big{)}.\end{split}\end{equation} Since  $\mu^{\alpha,*}_0(\delta,0,0)=...=\mu^{\alpha,*}_{n-1}(\delta,0,0)=1$ for all $\delta$, we have
\begin{equation}\label{modular2}
    L^{\alpha}_n(\delta,0,1)   =1+\alpha p L^{\alpha}_{n-1}(\delta+1,0,1). 
\end{equation} Since \eqref{modular2} holds for all $\delta$, we can iteratively use \eqref{modular2}, similar to \eqref{proof13_1}, \eqref{proof09}, to get \begin{equation}\label{b00_n} \begin{split}   
     L^{\alpha}_n(\delta,0,1) = \sum_{i=0}^{n-1} (\alpha p)^i+(\alpha p)^{n} L^{\alpha}_{0}(\delta+n,0,1).\end{split}\end{equation} 
Since $L^{\alpha}_0(\delta+n,0,1)=0$, \eqref{b00_n} directly gives \eqref{ln01}. \end{proof}
%We let $k=n$ in Lemma \ref{l00_n} and we get 
%     \begin{equation}\label{ln01}
%         L_n(\delta,0,1)=\sum_{i=0}^{n-1} (\alpha p)^i.
%     \end{equation}
Then, we derive $L^{\alpha}_n(\delta,0,2)$ in \eqref{wantoshow}. Following the same steps that are used in Lemma \ref{g2m}, we can show that: 
%Lemma \ref{g2m} implies that $\\ L^{\alpha}(\delta,0,2)=m$. we will get a similar result. %Recall that in the proof of Lemma \ref{g2m}, Lemma \ref{rhs2} gives the relationship between $L^{\alpha}(\delta,0,2)$ and the value functions $J^{\alpha}(\delta+k,l_1,d-k)$, $J^{\alpha}(\delta+k-1,l_1,d-k)$ for $k=1,2,...,d-1$. 
%By only adding the iterations numbers compared with Lemma \ref{rhs2}, we can also get:
%\begin{lemma}\label{rhs2_n}
%Suppose that $k=\min\{n,d-1\}$. The following equation holds for all $\delta$: 
\begin{equation}\label{b002}
    \begin{split}
     &   L^{\alpha}_n(\delta,0,2)  = \sum_{i=0}^{k-1}\alpha^i \\&+\alpha^k a_k \Big{(}J^{\alpha}_{n-k}(\delta+k,1,d-k)-J^{\alpha}_{n-k}(\delta+k-1,1,d-k)\Big{)} \\ & +\alpha^k b_k\Big{(}J^{\alpha}_{n-k}(\delta+k,0,d-k)-J^{\alpha}_{n-k}(\delta+k-1,0,d-k)\Big{)},
    \end{split}
\end{equation} where $k=\min\{n,d-1\}$, and $a_k,b_k$ are defined in \eqref{matrix1}.
%$[a_k,b_k]=[0,1]\times P^k$, and $P$ is the transition matrix of Channel $1$ and $P = \begin{bmatrix}
%q & 1-q \\
%1-p & p 
%\end{bmatrix}$.
%\end{lemma} 
%\begin{proof}
%See Appendix \ref{rhs2_napp}. \textcolor{red}{modify}
%\end{proof}

If $n\le d-1$, then $k=n$ and the value functions $J^{\alpha}_{n-k}(\cdot)$ inside \eqref{b002} are $0$. Thus, $L_n(\delta,0,2)=\sum_{i=0}^{n-1}\alpha^i$.  

If $n>d-1$, then $k=d-1$. We will expand all the value functions in \eqref{b002}. Recall that for all age value $\delta_0>d-1$, we have the same equation as \eqref{equation04} except adding a subscription:
\begin{equation}\label{equation04_n}\begin{split}
  & J^{\alpha}_{n-d+1}(\delta_0,0,1) \\ = & \delta_0 + \alpha p J^{\alpha}_{n-d}(d,0,0)+ \alpha (1-p) J^{\alpha}_{n-d}(d,1,0),  \\
  & J^{\alpha}_{n-d+1}(\delta_0,1,1) \\ = & \delta_0 + \alpha (1-q) J^{\alpha}_{n-d}(d,0,0)+ \alpha q J^{\alpha}_{n-d}(d,1,0).\end{split}\end{equation} 
 Applying \eqref{equation04_n} and \eqref{matrixx} into \eqref{b002}, we get the following equation which is the same as \eqref{A002}, except adding a subscription: \begin{equation}\begin{split}
    L^{\alpha}_n(\delta,0,2)
    =& \sum_{i=0}^{d-2}\alpha^i + \alpha^{d-1} (a_{d-1}+b_{d-1}) \\ +&\alpha^{d}a_{d}  \Big{(}J^{\alpha}_{n-d}(d,1,0)-J^{\alpha}_{n-d}(d,1,0)\Big{)} \\ +&\alpha^{d}b_{d}\Big{(}J^{\alpha}_{n-d}(d,0,0)-J^{\alpha}_{n-d}(d,0,0)\Big{)}\\ =& \sum_{i=0}^{d-1}\alpha^i \triangleq m.
    \end{split}
\end{equation}Thus, 
\begin{equation}\label{ln02}
    L^{\alpha}_n(\delta,0,2)=\left\{
\begin{array}{lll}
  m  & \text{if } n\ge d;\\
  \sum_{i=0}^{n-1} \alpha^i  & \text{if } n< d. 
\end{array}
\right. 
\end{equation} Since $(p,q,d)\in \textbf{B}_1 \cup \textbf{B}_4$, we have $\sum_{i=0}^{\infty} (\alpha p)^i \le m$. Thus, from \eqref{ln01} and \eqref{ln02}, we get $L_n(\delta,0,1)\le L_n(\delta,0,2)$, which proves condition (i).

%The proof has a slight difference for $n< d+1$ and $n\ge d+1$. 

 We next show (ii). %Since $Q^{\alpha}_1(\delta,0,0,1)=Q^{\alpha}_1(\delta,0,0,2)=\delta$, so $\mu_1(\delta,0,0)=1$ and we default that $n\ge 2$.
 We have a following statement:
 
 \begin{lemma}\label{q11}
Suppose that $k=\min \{ n,d-1 \}$. Then, we have:
\begin{equation}\label{q11eq}
    \begin{split}
      &  Q^{\alpha}_n(1,0,0,2)-Q^{\alpha}_n(1,0,0,1)
     \ge  \sum_{i=0}^{k-1}\alpha^i(1-p^i) \\  & + \alpha^k a_k \big{(} J^{\alpha}_{n-k}(k+1,1,d-k)-J^{\alpha}_{n-k}(1,1,0) \big{)} \\  & + \alpha^k(b_k-p^k)\big{(} J^{\alpha}_{n-k}(k+1,0,d-k)-J^{\alpha}_{n-k}(k,0,0) \big{)}\\  & + \alpha^k p^k\big{(} J^{\alpha}_{n-k}(k+1,0,d-k)-J^{\alpha}_{n-k}(k+1,0,0) \big{)},
    \end{split}
\end{equation}
where $a_k,b_k$ are defined in \eqref{matrix1}.
%$[a_k,b_k]=[0,1]\times P^k$, and $P$ is the transition matrix of Channel $1$ and  $P=\begin{bmatrix}
%q & 1-q \\
%1-p & p 
%\end{bmatrix}$. %Then the following inequality holds for $k=1,2,...,\min \{ n,d-1 \} $: 
 \end{lemma}
 \begin{proof}
 See Appendix \ref{q11app}. Note that $a_k+b_k=1$.
 \end{proof}
If $n\le d-1$, then $k=n$. In this case, all the value functions in \eqref{q11eq} (of Lemma \ref{q11}) are $0$. Then,
\begin{align}
\nonumber & Q^{\alpha}_n(1,0,0,2)-Q^{\alpha}_n(1,0,0,1) \\
     \ge &  \alpha(1-p)+...+\alpha^{n-1}(1-p^{n-1})\ge 0.    
\end{align}
Thus, $\mu^{\alpha,*}_n(1,0,0)=1$ and (ii) holds.  

If $n> d-1$, then $k=d-1$. In \eqref{q11eq}, we expand $J^{\alpha}_{n-d+1}(d,1,1)$, $J^{\alpha}_{n-d+1}(d,0,1)$, $J^{\alpha}_{n-d+1}(1,1,0)$, $J^{\alpha}_{n-d+1}(d-1,0,0)$ and  $J^{\alpha}_{n-d+1}(d,0,0)$ respectively.

The expansions of $J^{\alpha}_{n-d+1}(d,1,1), J^{\alpha}_{n-d+1}(d,0,1)$ follow from \eqref{equation04_n}: 

\begin{equation}\label{equation04_n1}\begin{split}
  & J^{\alpha}_{n-d+1}(d,0,1) \\ = & d + \alpha p J^{\alpha}_{n-d}(d,0,0)+ \alpha (1-p) J^{\alpha}_{n-d}(d,1,0),  \\
  & J^{\alpha}_{n-d+1}(d,1,1) \\ = & d + \alpha (1-q) J^{\alpha}_{n-d}(d,0,0)+ \alpha q J^{\alpha}_{n-d}(d,1,0).\end{split}\end{equation}

The value functions $J^{\alpha}_{n-d+1}(d,0,0)$, and $J^{\alpha}_{n-d+1}(d+1,0,0)$ are expanded as following:
\begin{equation}\begin{split}\label{equation04n_2}
    &J^{\alpha}_{n-d+1}(\delta_0,0,0)\le Q^{\alpha}_{n-d+1}(\delta_0,0,0,1)\\=& \delta_0+\alpha p J^{\alpha}_{n-d}(2,0,0)+\alpha (1-p) J^{\alpha}_{n-d}(1,1,0), \end{split}\end{equation} where $\delta_0\ge 0$.

%upper bounded by $Q^{\alpha}_{n-d+1}(d,0,0,1)$ and $Q^{\alpha}_{n-d+1}(d+1,0,0,1)$, which follow from \eqref{equation04_n}. 

Also, $J^{\alpha}_{n-d+1}(1,1,0)$ are expanded as follows:
\begin{equation}\begin{split}\label{equation04n_3}
    &J^{\alpha}_{n-d+1}(1,1,0)\le Q^{\alpha}_{n-d+1}(1,1,0,1) \\ \le & Q^{\alpha}_{n-d+1}(d-1,1,0,1)\\=& d-1+\alpha (1-q) J^{\alpha}_{n-d}(d,0,0)+\alpha q J^{\alpha}_{n-d}(1,1,0). \end{split}\end{equation}
Applying \eqref{equation04_n1},\eqref{equation04n_2} and \eqref{equation04n_3} into \eqref{q11eq}, we get   
\begin{equation}\label{q11eq_2}
    \begin{split}
      &  Q^{\alpha}_n(1,0,0,2)-Q^{\alpha}_n(1,0,0,1)\\
     \ge & \alpha(1-p)+...+\alpha^{d-1}(1-p^{d-1})\\  & + \alpha^d a_d \big{(} J^{\alpha}_{n-d}(d,1,0)-J^{\alpha}_{n-d}(1,1,0) \big{)} \\  & + \alpha^d(b_d-p^d)\big{(} J^{\alpha}_{n-d}(d,0,0)-J^{\alpha}_{n-d}(d,0,0) \big{)}\\  & + (\alpha p)^d\big{(} J^{\alpha}_{n-d}(d,0,0)-J^{\alpha}_{n-d}(d+1,0,0) \big{)}.
    \end{split}
\end{equation} Because value function is increasing in age, 
\begin{equation}\label{tacotuesday}
    J^{\alpha}_{n-d}(d,1,0)-J^{\alpha}_{n-d}(1,1,0)\ge 0.
\end{equation}
Thus, \eqref{q11eq_2} gives \begin{equation}\label{q11eq_3}
    \begin{split}
      &  Q^{\alpha}_n(1,0,0,2)-Q^{\alpha}_n(1,0,0,1)\\
     \ge & \alpha(1-p)+...+\alpha^{d-1}(1-p^{d-1}) \\  & +  (\alpha p)^d\big{(} J^{\alpha}_{n-d}(d,0,0)-J^{\alpha}_{n-d}(d+1,0,0) \big{)}.
    \end{split}
\end{equation} 

Since by the hypothesis, $\mu^{\alpha,*}_{0}(\delta,0,0)=...=\mu^{\alpha,*}_{n-d}(\delta,0,0)$ $=1$ for all $\delta\ge 0$, \eqref{ln01} implies that 
\begin{equation}\label{q11eq_3pre}
   J^{\alpha}_{n-d}(d,0,0)-J^{\alpha}_{n-d}(d+1,0,0)= -\sum_{i=0}^{n-d-1} (\alpha p)^i \ge - \sum_{i=0}^{\infty} (\alpha p)^i. 
\end{equation}
 Thus, \eqref{q11eq_3} and \eqref{q11eq_3pre} give \begin{equation}
    \begin{split}
    & Q^{\alpha}_n(1,0,0,2)-Q^{\alpha}_n(1,0,0,1) \\
  \ge &  \sum_{i=1}^{d-1}\alpha^{i}- \sum_{i=1}^{d-1} (\alpha p)^i  -(\alpha p)^d\sum_{i=0}^{\infty} (\alpha p)^i \\ = & \sum_{i=0}^{d-1}\alpha^{i}-\sum_{i=0}^{\infty} (\alpha p)^i=m-\sum_{i=0}^{\infty} (\alpha p)^i\ge 0,  
\label{1A00}    \end{split}
\end{equation} where the last inequality is because $(p,q,d)\in \textbf{B}_1 \cup \textbf{B}_4$.

Thus, (ii) holds. We complete the proof.

\section{Proof of Lemma \ref{q11}}\label{q11app}

We show Lemma \ref{q11} by using recursion. %We will show that \eqref{q11eq} holds for $k=1,2,...,\min \{ n,d-1 \} $.

First of all, same with \eqref{switchtype0} and \eqref{switchtype1}, we have 
\begin{equation}\begin{split}
  &  Q^{\alpha}_n(1,0,0,2)-Q^{\alpha}_n(1,0,0,1)\\
  =& \alpha p \big{(} J^{\alpha}_{n-1}(2,0,d-1)-J^{\alpha}_{n-1}(2,0,0) \big{)} \\
   & + \alpha (1-p) \big{(} J^{\alpha}_{n-1}(2,1,d-1)-J^{\alpha}_{n-1}(1,1,0) \big{)}.
\end{split}\label{q11_1_2}
\end{equation} Note that $1-p^0=0$. Thus, \eqref{q11_1_2} is the same with \eqref{q11eq} if $k=1$. 

Note that $J^{\alpha}_{n-1}(1,1,0)\le Q^{\alpha}_{n-1}(1,1,0,1)$,  and $J^{\alpha}_{n-1}(2,0,0)\le Q^{\alpha}_{n-1}(2,0,0,1)$. Then, the second term of \eqref{q11_1_2} is as follows:
%Suppose that \eqref{q11eq} holds for a value $k<\min\{n,d-1\}$. For the step of $k+1$, note that we expand all of the value functions in \eqref{q11eq} (so that the subscriptions change from $n-k$ to $n-k-1$),  
\begin{equation}\label{q11eq_proof}
    \begin{split}
 &  J^{\alpha}_{n-1}(2,1,d-1)-J^{\alpha}_{n-1}(1,1,0) \\ \ge   &  1+ \alpha q \big{(} J^{\alpha}_{n-2}(3,1,d-2)-J^{\alpha}_{n-2}(1,1,0) \big{)} \\ & + \alpha (1-q)\big{(} J^{\alpha}_{n-2}(3,0,d-2)-J^{\alpha}_{n-2}(2,0,0) \big{)}.  \end{split}
\end{equation} The first term of \eqref{q11_1_2} is as follows:
 \begin{equation}\label{q11eq_proof2}
    \begin{split}
 &  J^{\alpha}_{n-1}(2,0,d-1)-J^{\alpha}_{n-1}(2,0,0) \\ \ge   &  1+ \alpha (1-p) \big{(} J^{\alpha}_{n-2}(3,1,d-2)-J^{\alpha}_{n-2}(1,1,0) \big{)} \\ & + \alpha p \big{(} J^{\alpha}_{n-2}(3,0,d-2)-J^{\alpha}_{n-2}(3,0,0) \big{)}.
    \end{split}
\end{equation} Thus, applying \eqref{q11eq_proof} and \eqref{q11eq_proof2} into \eqref{q11eq} with $k=1$, we get \eqref{q11eq} when $k=2$. By using \eqref{q11eq_proof} and \eqref{q11eq_proof2} iteratively for $\min \{ n,d-1 \}-2$ times, we finally derive \eqref{q11eq} when $k=\min \{ n,d-1 \}$ (note that if $\min \{ n,d-1 \}=1$, we have proved \eqref{q11eq} in \eqref{q11_1_2}).
%\begin{equation}\label{q11eq_proof}
%    \begin{split}
% &  J^{\alpha}_{n-k}(k+1,1,d-k)-J^{\alpha}_{n-k}(1,1,0) \\ \ge   &  1+ \alpha q \big{(} J^{\alpha}_{n-k-1}(k+2,1,d-k-1)-J^{\alpha}_{n-k-1}(1,1,0) \big{)} \\ & + \alpha (1-q)\big{(} J^{\alpha}_{n-k-1}(k+2,0,d-k-1)-J^{\alpha}_{n-k-1}(2,0,0) \big{)} \\ \ge &  1+ \alpha q \big{(} J^{\alpha}_{n-k-1}(k+2,1,d-k-1)-J^{\alpha}_{n-k-1}(1,1,0) \big{)} \\ & + \alpha (1-q)\big{(} J^{\alpha}_{n-k-1}(k+2,0,d-k-1)-J^{\alpha}_{n-k-1}(k+1,0,0) \big{)}, \\  &  J^{\alpha}_{n-k}(k+1,0,d-k)-J^{\alpha}_{n-k}(k,0,0) \\    \ge & 1+ \alpha (1-p) \big{(} J^{\alpha}_{n-k-1}(k+2,1,d-k-1)-J^{\alpha}_{n-k-1}(1,1,0) \big{)} \\ & + \alpha p \big{(} J^{\alpha}_{n-k-1}(k+2,0,d-k-1)-J^{\alpha}_{n-k-1}(k+1,0,0) \big{)},\\  &  J^{\alpha}_{n-k}(k+1,0,d-k)-J^{\alpha}_{n-k}(k+1,0,0) \\   \ge & \alpha (1-p) \big{(} J^{\alpha}_{n-k-1}(k+2,1,d-k-1)-J^{\alpha}_{n-k-1}(1,1,0) \big{)} \\ & + \alpha p \big{(} J^{\alpha}_{n-k-1}(k+2,0,d-k-1)-J^{\alpha}_{n-k-1}(k+2,0,0) \big{)}.
%    \end{split}
%\end{equation} Equation \eqref{q11eq_proof} provides the expansions of all the three value functions pairs in \eqref{q11eq}. We know the definition of sequences $a_k$ and $b_k$: $b_{k+1}=p b_k+(1-q)a_k$ and $a_{k+1}=(1-p) b_k+ q a_k$. Also, $a_k+b_k=1$. Through \eqref{q11eq_proof} and properties of $a_k,b_k$, we can finally get \eqref{q11eq} for the step $k+1$. Thus, we complete the proof by taking $k=\min \{ n,d-1 \} $ in \eqref{q11eq}. 

\section{proof of Lemma \ref{mu110for1}}\label{mu110for1app}

Recall that we use $\mu^{\alpha,*}(\cdot)$ to denote the optimal policy of the discounted problem. From Lemma \ref{theorem1}, it is sufficient to show that: for all discount factor $\alpha$, $\mu^{\alpha,*}(1,1,0)=1$ if $(p,q,d)\in \textbf{B}_1(\alpha)$. %For simplicity we use $\mu$ to replace $\mu^\alpha$.

The condition $(p,q,d)\in \textbf{B}_1(\alpha)$ implies that $H(p,q,d,\alpha)\le 0$ and $F(p,q,d,\alpha)\le 0$. From Theorem \ref{theorem2}, $\mu^{\alpha,*}(\delta,1,0)$ is non-increasing in $\delta$. We want to show that $Q^{\alpha}(1,1,0,2)\ge Q^{\alpha}(1,1,0,1)$. Then, $\mu^{\alpha,*}(\delta,1,0)=1$ for all $\delta$. 

Using the same technique with the proof of Lemma \ref{q11}, we get:
\begin{equation}\label{q11eq2}
    \begin{split}
      &  Q^{\alpha}(1,1,0,2)-Q^{\alpha}(1,1,0,1)
  \\   \ge & h'(d-1) +  \alpha^{d-1} a'_{d-1} \big{(} J^{\alpha}(d,1,1)-J^{\alpha}(1,1,0) \big{)} \\  & + \alpha^{d-1}(b'_{d-1}-(1-q)p^{d-1})\big{(} J^{\alpha}(d,0,1)-J^{\alpha}(d-1,0,0) \big{)}\\  & +  \alpha^{d-1} (1-q)p^{d-2}\big{(} J^{\alpha}(d,0,1)-J^{\alpha}(d,0,0) \big{)},
    \end{split}
\end{equation} where $a'_{d-1},b'_{d-1}$ are defined in \eqref{matrix1}, and the function $h'(k)$ for $k\in \{1,2,...\}$ is defined as follows:
\begin{equation*}
    h'(k) = \left\{
\begin{array}{lll}
  \sum_{i=1}^{k-1}\alpha^i (1-(1-q)p^{i-1}) & \text{if } k\ge 2;\\
  0  & \text{if } k=1. 
\end{array}
\right. 
\end{equation*}
Applying \eqref{equation04_n1}, \eqref{equation04n_2} and \eqref{equation04n_3} into \eqref{q11eq2} and we get
\begin{equation}\label{q11eq_4}
    \begin{split}
      &  Q^{\alpha}(1,1,0,2)-Q^{\alpha}(1,1,0,1)\\
     \ge & \alpha(1-(1-q)p^0)+...+\alpha^{d-1}(1-(1-q)p^{d-2})\\ & + \alpha^d a'_d \big{(} J^{\alpha}(d,1,0)-J^{\alpha}(1,1,0) \big{)} \\  & + \alpha^d(b'_d-(1-q)p^{d-1})\big{(} J^{\alpha}(d,0,0)-J^{\alpha}(d,0,0) \big{)}\\ & +  \alpha^d (1-q)p^{d-1}\big{(} J^{\alpha}(d,0,0)-J^{\alpha}(d+1,0,0) \big{)}\\  \ge & \alpha(1-(1-q)p^0)+...+\alpha^{d-1}(1-(1-q)p^{d-2}) \\  & + \alpha (1-q)(\alpha p)^{d-1}\big{(} J^{\alpha}(d,0,0)-J^{\alpha}(d+1,0,0) \big{)}\\ = & \sum_{i=0}^{d-1}\alpha^i-1-\alpha (1-q) \sum_{i=0}^{d-2} (\alpha p)^i  \\  & +  \alpha (1-q)(\alpha p)^{d-1}\big{(} J^{\alpha}(d,0,0)-J^{\alpha}(d+1,0,0) \big{)},
    \end{split}
\end{equation} where the second inequality is from \eqref{tacotuesday}. From Lemma \ref{mu00in1and4}, we know that $\mu^{\alpha,*}(\delta,0,0)=1$ for all $\delta$. Then, \eqref{proof09} implies that \begin{equation}
    J^{\alpha}(d,0,0)-J^{\alpha}(d+1,0,0)=-\sum_{i=0}^{\infty}(\alpha p)^i.
\end{equation}
 Thus, \eqref{q11eq_4} becomes
\begin{equation}
    \begin{split}
     & Q^{\alpha}(1,1,0,2)-Q^{\alpha}(1,1,0,1)\\
     \ge & \sum_{i=0}^{d-1}\alpha^i-1-\alpha (1-q) \sum_{i=0}^{d-2} (\alpha p)^i  -   \alpha (1-q)(\alpha p)^{d-1}\sum_{i=0}^{\infty}(\alpha p)^i\\
  = & \sum_{i=0}^{d-1}\alpha^i-1-\alpha (1-q) \sum_{i=0}^{\infty} (\alpha p)^i=-H(p,q,d,\alpha)\ge 0. \label{solution110}
    \end{split}
\end{equation} Thus, $\mu^{\alpha,*}(1,1,0)=1$.

\section{Diagrams and Derivations of Steady-State DTMCs}\label{markovapp}

This section provides the Markov chains corresponding to the cases in the proofs of Theorem \ref{theorem2b}---\ref{theorem2d} in Section \ref{proof3}. The Markov chains are described in Fig. \ref{markov1}---\ref{markov6}. The derivations of the expected age for each Markov chain are described later. %Appendix \ref{markovderapp}.
We need to remark here for the descriptions of the following Markov chains. (i) We sometimes replace two states by a new "state" in the Markov chains. For example, in Fig. \ref{markov1}, we include the two states $(s+d,0,1),(s+d,1,1)$ into one circle (the same occurs for $(s+1,0,d-1),(s+1,1,d-1)$, etc). This means that we only consider the combined probability distribution of the two states $(s+d,0,1),(s+d,1,1)$. The combination of the two states can largely simplify the Markov chains figures. Also, it does not affect the derivations of the expected age. (ii) The values $a_d,b_d,a'_d,b'_d$ are defined in \eqref{matrix1}. Suppose that we choose Channel $2$ with $l_1=0$. Then \eqref{A002} and \eqref{A0002} imply that the probabilities of returning back to $(d,1,0)$, and $(d,0,0)$ are $a_d,b_d$ respectively (e.g., see Fig. \ref{markov1}). If $l_1=1$, then the probabilities are $a'_d,b'_d$ respectively (e.g., see the left part of Fig. \ref{markov3}).  

%We also provide the derivation of each Markov chain. Note that: (i) We use $\pi(\delta,l_1,l_2)$ to denote the probability distribution of a state $(\delta,l_1,l_2)$. (ii) If $l_2>0$, we only need the combined probability $\pi(\delta,0,l_2)+\pi(\delta,1,l_2)$ rather than each probability. This is because the cost from $\pi(\delta,0,l_2)$ and $\pi(\delta,1,l_2)$ are both $\delta$.

\subsection{} \label{markov1app} %The threshold $s>d$ and the policy $\mu^*(1,1,0)=\mu^*(d,1,0)=1$}

\begin{figure}[htbp]
\centerline{\includegraphics[width=0.5\textwidth]{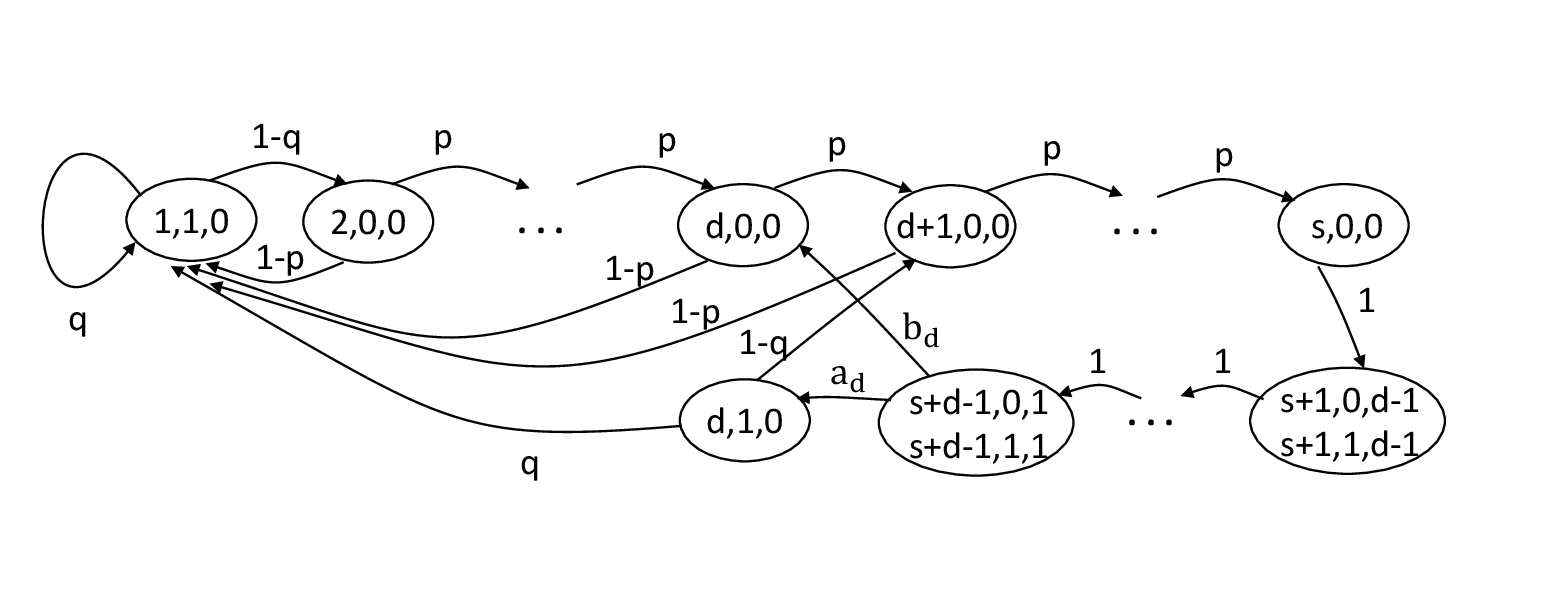}}
\caption{ The threshold $s>d$ and the optimal decisions $\mu^*(1,1,0)=\mu^*(d,1,0)=1$.}
\label{markov1}
\end{figure}

Referring to Fig. \ref{markov1}, we derive the balance equation on the states $(2,0,0),...,(d-1,0,0),(d+1,0,0),...,(s,0,0)$, and the $d-1$ combined states out of $(s,0,0)$ respectively. Then we get
\begin{equation}\begin{split}\label{markov1_1}
 & \pi(\delta,0,0)=p^{\delta-d-1}\pi(d+1,0,0) \ \ \delta=d+1,...,s, \\
 & \pi(\delta,0,0)=(1-q)p^{\delta-2}\pi(1,1,0) \ \ \delta=2,3,...,d-1,\\
   & \pi(s,0,0)=\pi(s+1,0,d-1)+\pi(s+1,1,d-1)  \\
  &...=\pi(s+d-1,0,1)+\pi(s+d-1,1,1). 
\end{split}
\end{equation} From \eqref{markov1_1}, the balance equation on the state $(d,1,0)$ implies \begin{equation}
    \pi(d,1,0)=a_d \pi(s,0,0).
\end{equation} The balance equation on the state $(d+1,0,0)$ implies \begin{equation}
    (1-q)\pi(d,1,0)+p\pi(d,0,0)=\pi(d+1,0,0).
\end{equation} The balance equation on the state $\pi(d,0,0)$ gives \begin{equation}
    p \pi(d-1,0,0)+b_d p^{s-d-1}\pi(d+1,0,0)=\pi(d,0,0).
\end{equation} 
The above equations give \begin{equation}\label{markov1_2}
\begin{split}
     & \pi(d,1,0)+\pi(d,0,0)\\ & = \big{(}\frac{p^{s-d+1}}{1-b_d p^{s-d}-(1-q)a_d p^{s-d-1}}+p \big{)} (1-q)p^{d-3} \pi(1,1,0), \\
     & \pi(d+1,0,0)=\frac{(1-q)p^{d-1}}{1-b_d p^{s-d}-(1-q)a_d p^{s-d-1}}\pi(1,1,0).
\end{split}
\end{equation} Thus, \eqref{markov1_1} and \eqref{markov1_2} directly implies that all the states in the Markov chain can be expressed in terms of $\pi(1,1,0)$. Since the summing up of all the states probabilities are $1$, we can directly get the distribution of $\pi(1,1,0)$: \begin{equation}
  \pi(1,1,0) = \frac{p c_1(s)}{(1-q) g_1(s)}.
\end{equation} Where $c_1(s),g_1(s)$ are described in Table \ref{fg}. The expected age is the summation of the probability of the state multiplied by the state's age value, which is given by 
\begin{equation}
  \pi(1,1,0)(1-q)/(p c_1(s))\times f_1(s)=f_1(s)/g_1(s)  
\end{equation}
The function $f_1(s)$ is in Table \ref{fg} as well.    

Thus, the expected age is $f_1(s)/g_1(s)$.

\subsection{} \label{markov2app} %The threshold $s\le d$ and the policy $\mu^*(1,1,0)=\mu^*(d,1,0)=1$}

\begin{figure}[htbp]
\centerline{\includegraphics[width=0.5\textwidth]{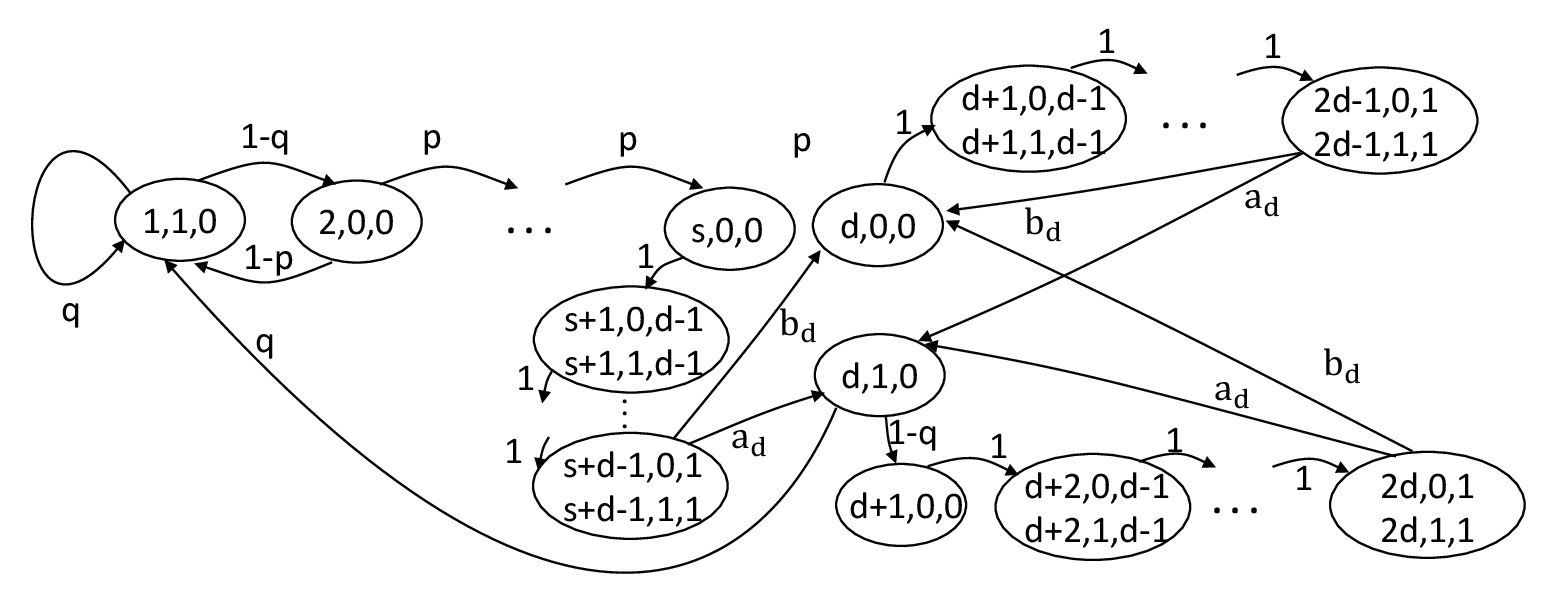}}
\caption{The threshold $s\le d$ and the optimal decisions $\mu^*(1,1,0)=\mu^*(d,1,0)=1$.}
\label{markov2}
\end{figure}

Referring to Fig. \ref{markov2}, we derive the balance equations on the states $(2,0,0),(3,0,0),...,(s,0,0)$, and the $d-1$ combined states out of $(s,0,0)$ , and get 
\begin{equation}\begin{split}\label{markov2_1}
  &  \pi(s,0,0)=\pi(s+1,0,d-1)+\pi(s+d-1,1,d-1)\\
  & ... =\pi(s+d-1,0,1)+\pi(s+d-1,1,1),\\
  & \pi(\delta,0,0)=(1-q)p^{\delta-2}\pi(1,1,0) \ \ \delta=2,3,...,s.
    \end{split}
\end{equation} We then observe the set $\{ (1,1,0),(2,0,0),...,(s,0,0)\}$: the inflow of $q\pi(d,1,0)$ equals to the outflow $(s,0,0)$. Thus, combined with \eqref{markov2_1}, \begin{equation}\label{markov2_2}
   q \pi(d,1,0) = \pi(s,0,0)=(1-q)p^{s-2}\pi(1,1,0).
\end{equation} The state $(d,1,0)$ gives 
\begin{equation}
    a_d \pi(d,0,0)=b_d\pi(s,0,0)+b_d(1-q)\pi(d,1,0),
\end{equation} thus,
\begin{equation}\label{markov2_3}
    \pi(d,0,0)=\frac{b_d}{a_d q} \pi(s,0,0)=\frac{b_d}{a_d q} (1-q)p^{s-2}\pi(1,1,0).
\end{equation} Thus, \eqref{markov2_1}, \eqref{markov2_2} and \eqref{markov2_3} imply that all the states in the Markov chain can be expressed in terms of $\pi(1,1,0)$. Also, the sums up of the probability of all the states is $1$: 
\begin{equation}\begin{split}
  & \pi(1,1,0)+ \sum_{\delta=2}^s \pi(\delta,0,0)+(d-1)\pi(s,0,0)+d \pi(d,0,0) \\ & +(d(1-q)+1)\pi(d,1,0)=1.\end{split}
\end{equation} Thus,\begin{equation}
  \pi(1,1,0) = \frac{p}{(1-q) g_2(s)}.
\end{equation}   
Thus, we give the expected age to be $f_2(s)/g_2(s)$ in Table \ref{fg}.

\subsection{} \label{markov3app} % The threshold $s=1$ and the policy $\mu(1,1,0)=2, \mu(d,1,0)=1$}

\begin{figure}[htbp]
\centerline{\includegraphics[width=0.5\textwidth]{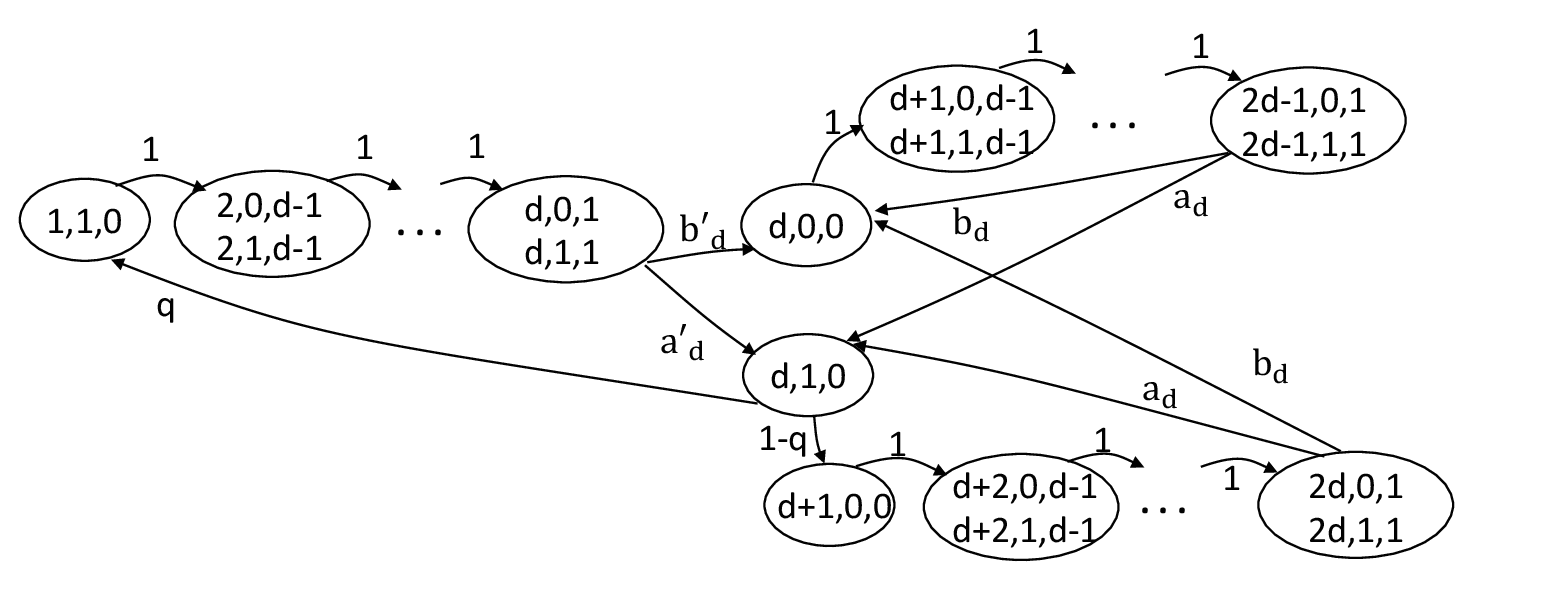}}
\caption{ The threshold $s=1$ and the optimal decisions $\mu^*(1,1,0)=2, \mu^*(d,1,0)=1$.}
\label{markov3}
\end{figure}
Referring to Fig. \ref{markov3}, the $d-1$ combinations states from $(1,1,0)$ gives
\begin{equation}\begin{split}
   & \pi(1,1,0)=\pi(2,0,d-1)+\pi(2,1,d-1) \\
   & ...=\pi(d,0,1)+\pi(d,1,1),\end{split}
\end{equation} the state $(1,1,0)$ gives \begin{equation}
    \pi(1,1,0)=q \pi(d,1,0).
\end{equation} Also, 
\begin{equation}\begin{split}
  &  \pi(d+1,0,0)=\pi(d+2,0,d-1)+\pi(d+2,1,d-1)\\
  & ... =\pi(2d,0,1)+\pi(2d,1,1)=(1-q)\pi(d,1,0).
  \end{split}
\end{equation} The state of $(d,0,0)$ gives 
\begin{equation}
    (1-b_d)\pi(d,0,0)=b'_d \pi(1,1,0)+b_d \pi(d,1,0),
\end{equation} thus, \begin{equation}
    \pi(d,0,0)=\frac{b'_d q + b_d}{1-b_d}\pi(d,1,0).
\end{equation}

Thus, all state distributions can be expressed in terms of $\pi(d,1,0)$, and
the expected age is $f_0/g_0$.

\subsection{}\label{markov4app} %The threshold $s>d$ and the policy $\mu(1,1,0)=1, \mu(d,1,0)=2$.}

\begin{figure}[htbp]
\centerline{\includegraphics[width=0.5\textwidth]{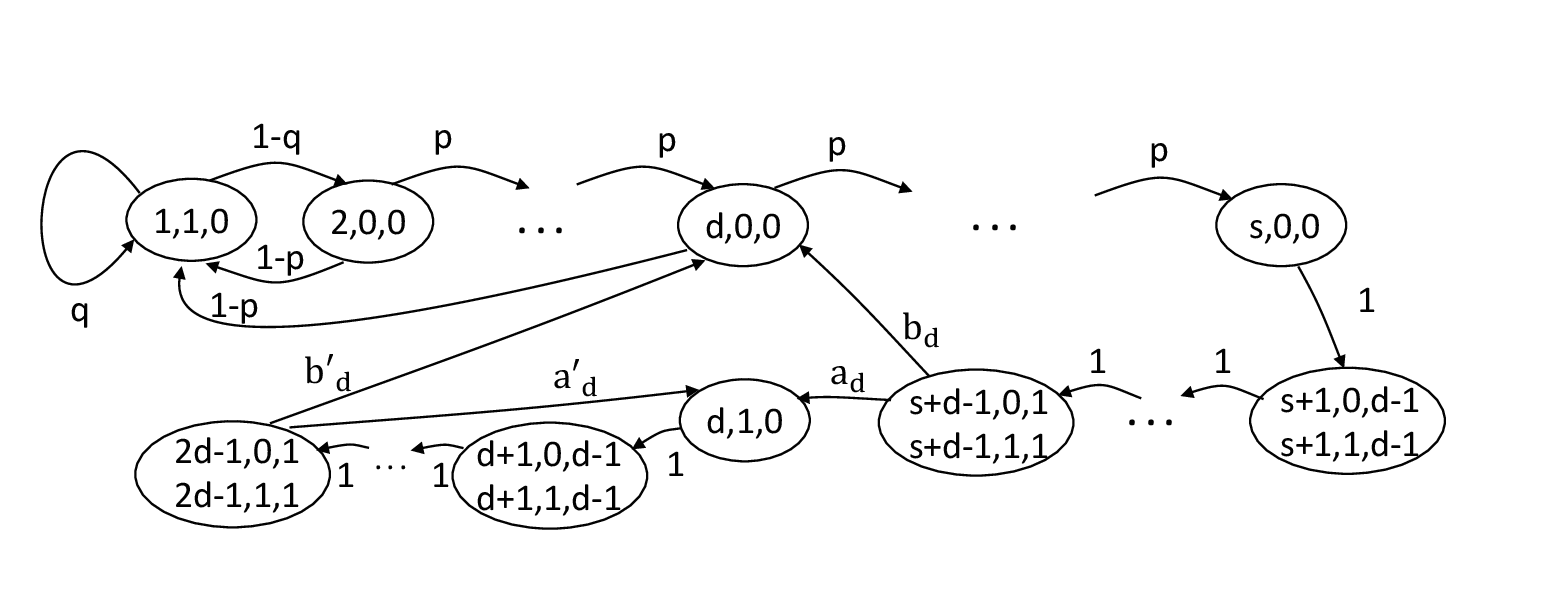}}
\caption{The threshold $s>d$ and the optimal decisions $\mu^*(1,1,0)=1, \mu^*(d,1,0)=2$.}
\label{markov4}
\end{figure}
Referring to Fig. \ref{markov4}, the states $(2,0,0),...,(d-1,0,0)$, $(d+1,0,0),...,(s,0,0)$, and $d-1$ states from $(s,0,0)$ give:
\begin{equation}\begin{split}\label{markov4_1}
  &  \pi(\delta,0,0)=p^{\delta-2}(1-q)\pi(1,1,0) \ \  \delta=2,3,...,d-1,\\
  & \pi(\delta,0,0)=p^{\delta-d}\pi(d,0,0) \ \  \delta=d+1,...,s, \\
 & \pi(s,0,0)=\pi(s+1,0,d-1)+\pi(s+1,1,d-1) \\
  & ...= \pi(s+d-1,0,1)+\pi(s+d-1,1,1).
  \end{split}
\end{equation} 
The combination of states $(d,0,0)$ and $(d,1,0)$ gives 
\begin{equation}
    \pi(s,0,0)+p\pi(d-1,0,0)=\pi(d,0,0).
\end{equation} Equation \eqref{markov4_1} implies that $\pi(d-1,0,0)=p^{d-3}(1-q)\pi(1,1,0)$, thus, 
\begin{equation}
    \pi(d,0,0)=\frac{p}{1-p^{s-d}}\pi(d-1,0,0)=\frac{p^{d-2}(1-q)}{1-p^{s-d}}\pi(1,1,0).
\end{equation} The state $(d,1,0)$ gives \begin{equation}
    \pi(d,1,0)(1-a'_d)=a_d \pi(s,0,0)=a_d p^{s-d}\pi(d,0,0).
\end{equation} Thus, all the states distributions in the Markov chain can be expressed in terms of $\pi(1,1,0)$. The expected age is $f_3(s)/g_3(s)$.

\subsection{} \label{markov5app} %The threshold $s>d$ and the policy $\mu(1,1,0)= \mu(d,1,0)=2$}

\begin{figure}[htbp]
\centerline{\includegraphics[width=0.5\textwidth]{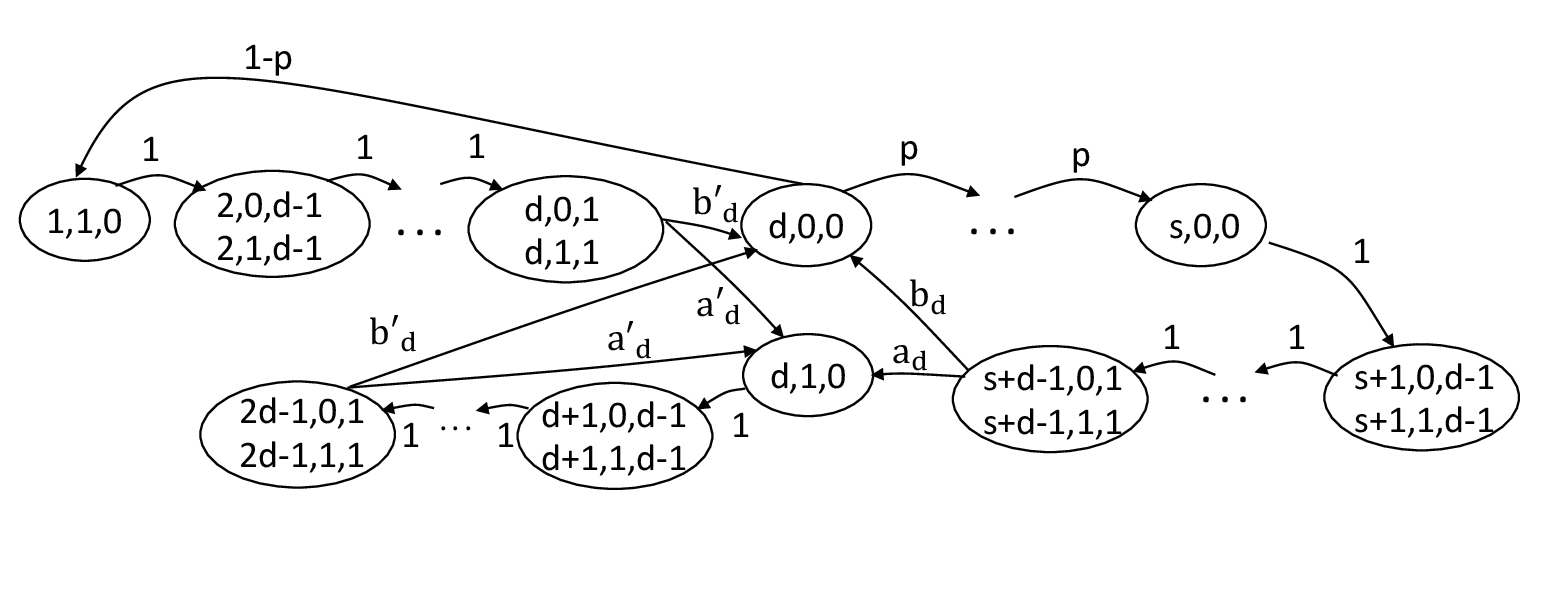}}
\caption{The threshold $s>d$ and the optimal decisions $\mu^*(1,1,0)= \mu^*(d,1,0)=2$.}

\label{markov5}
\end{figure}
Referring to Fig. \ref{markov5}, the balance equations of the $d-1$ states from $(s,0,0)$ and states $(d+1,0,0),...,(s,0,0)$ are given by:
\begin{equation}\begin{split}\label{markov5_1}
 & \pi(s,0,0) = \pi(s+1,0,d-1)+\pi(s+1,1,d-1) \\ & ...=\pi(s+d-1,0,1)+\pi(s+d-1,1,1),\\
 & \pi(\delta,0,0)=p^{\delta-d}\pi(d,0,0), \ \ \delta=d+1,...,s.  
 \end{split}
\end{equation} The combination of $(d,0,0),(d,1,0)$ gives \begin{equation}
    \pi(1,1,0)+\pi(s,0,0)=\pi(d,0,0),
\end{equation} thus, using \eqref{markov5_1}, we get 
\begin{equation}
    \pi(1,1,0)=(1-p^{s-d})\pi(d,0,0).
\end{equation} By looking at $\pi(d,1,0)$, \begin{equation}
    (1-a'_d)\pi(d,1,0)=a_d \pi(s,0,0)+a'_d \pi(1,1,0),
\end{equation} thus, \begin{equation}
    \pi(d,1,0)=\frac{a'_d+(a_d-a'_d)}{1-a'_d}\pi(d,0,0).
\end{equation} 
Thus, all the states distributions in the Markov chain can be expressed in terms of $\pi(d,0,0)$. Similar to previous sections, the distribution can be solved and the expected age is $f_4(s)/g_4(s)$.

%\section{Derivations of the Markov Chains}\label{markovderapp}

%\subsection{Derivations of Fig. \ref{markov1}}

%\subsection{Derivations of Fig. \ref{markov2}}

%\subsection{Derivations of Fig. \ref{markov3}}

%\subsection{Derivations of Fig. \ref{markov4}}

%\subsection{Derivations of Fig. \ref{markov5}}

\subsection{} \label{markov6app} %The threshold $s>d$ and the policy $\mu(1,1,0)= \mu(d,1,0)=2$}

\begin{figure}[htbp]
\centerline{\includegraphics[width=0.5\textwidth]{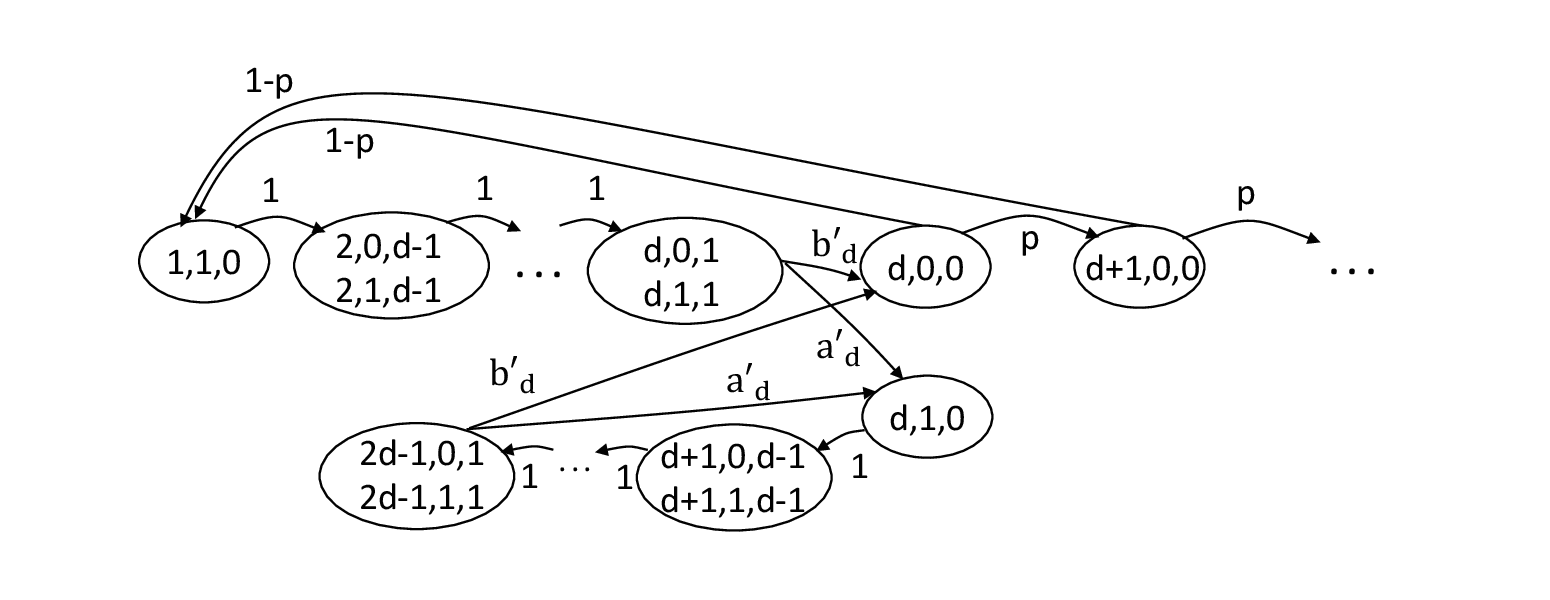}}
\caption{The optimal decisions $\mu^*(\delta,0,0)=1$ and $\mu^*(\delta,1,0)=2$ for all $\delta\geq 1$.}

\label{markov6}
\end{figure}

Referring to Fig. \ref{markov6}, the balance equations of the states $(d+1,0,0),...$ give \begin{equation}\label{markov6_1}
    \pi(\delta,0,0)=p^{\delta-d}\pi(d,0,0) \ \ \delta=d+1,d+2,...
\end{equation} The state $(1,1,0)$ and the $d-1$ combinations states from $(1,1,0)$ imply that \begin{equation}\label{markov6_2}\begin{split}
 &    \pi(1,1,0)=\pi(2,0,d-1)+\pi(2,1,d-1) \\
   & ...=\pi(d,0,1)+\pi(d,1,1),
 \\ &  \pi(1,1,0)=(1-p)\big{(} \pi(d,0,0)+\pi(d+1,0,0)+... \big{)}\end{split}
\end{equation} From \eqref{markov6_1},\eqref{markov6_2} we can get 
\begin{equation}
    \pi(1,1,0)=\pi(d,0,0).
\end{equation}
The $d-1$ combinations states from $(d,1,0)$ implies that \begin{equation}
    \begin{split}
    &   \pi(d,1,0)= \pi(d+1,0,d-1)+\pi(d+1,1,d-1)\\ & ...=\pi(2d-1,0,1)+\pi(2d-1,1,1).
    \end{split}
\end{equation}
The state $(d,1,0)$ gives \begin{equation}
    b'_d \pi(1,1,0)+b'_d \pi(d,1,0)=\pi(d,0,0).
\end{equation} thus, \begin{equation}
    \pi(d,1,0)=\frac{a'_d}{b'_d}\pi(d,0,0).
\end{equation} Thus, all the states probabilities can be expressed in terms of $\pi(d,0,0)$. By normalizing, we get $\pi(d,0,0)=1/g'_0$. Then the expected age is $f'_0/g'_0$.

%\section{Table for }

\section{Proof of Lemma \ref{fractional}}\label{fractionalapp}

%From Table \ref{fg}, for all $i=1,2,3,4$, we have $g_{i}(s)>0$. Also, $f_i(s)$ and $g_i(s)$ are bounded. Thus, Lemma \ref{fractional} is directly shown in Theorem of \cite{dinkelbach1967nonlinear}. 
%It is sufficient to show the following conditions: (i) If $c>\beta_i$, then $h_i(c)<0$.  
We rewrite $\beta_i'$:
\begin{align}
\beta'_{i}& =
\min_{s\in \{ d+1,...\} }\frac{f_{i}(s)}{g_{i}(s)}, \ i\in \{1,3,4\}, \label{algr2} \\ 
\beta'_{2}&=
\min_{s\in \{1,...d\} }\frac{f_{2}(s)}{g_{2}(s)}.
\label{algr2-2}\end{align}
Further, we rewrite \eqref{h-func-1} and \eqref{h-func-2}:
\begin{align}
h'_{i}(c) &=\min_{s\in \{ d+1,...\} }  f_{i}(s) - c g_{i}(s), \ i\in \{1,3,4\}, \label{inlayer2} \\
h'_{2}(c) &= \min_{s\in \{2,...d\} }  f_{2}(s) - c g_{2}(s). \label{inlayer2-2}
\end{align}

From Table \ref{fg}, it is easy to find that there exists a value $d'>0$ such that $g_{i}(s)>d'$ for all $s$ and $i$. Also, $g_i(s)$ is upper bounded. Thus, from \eqref{algr2},\eqref{algr2-2}, \eqref{inlayer2} and \eqref{inlayer2-2}, $h'_i(c)\lesseqqgtr 0$ is equivalent to $c \gtreqqless \beta_i'$, which proves our result.  

\section{proof of Lemma \ref{fracsolution}}\label{app2}
Notice that $(p,q,d)\in \textbf{B}_2\cup \textbf{B}_3$ if and only if $1-p<1/d$. Suppose that $r_{i}(c,s)\triangleq f_{i}(s)-cg_{i}(s)$. 

%suppose $s'$ is a local minimum. Then $l(s')\le l(s'+1)$ and $l(s')\le f(s'-1)-cg(s'-1)$. It is sufficient to show $l(s,c)$ is decreasing in $s\le s'$ and increasing on $s\ge s'$.

We find that:
\begin{equation}\label{f}
%\begin{split}
     p^{-(s-1)}\Big{(} f_{i}(s+1)- f_{i}(s)\Big{)}
  = \Big{(}1-d(1-p)\Big{)}s+l'_{i} %\end{split}
    \end{equation}
    where $l'_{i}$ is not related to $c,s$ and is described in Table \ref{fg}. Also, 
    \begin{equation}\label{g}%\begin{split}
    p^{-(s-1)}\Big{(} g_{i}(s+1)- g_{i}(s)\Big{)}  = o_{i},
%\end{split}
\end{equation} where  $o_{i}$ is not related to $c,s$ and are described in Table \ref{fg} as well. Thus, \eqref{f} and \eqref{g} give: 
\begin{equation}
    \begin{split}
    & p^{-(s-1)} (r_{i}(c,s+1)-r_{i}(c,s)) \\ & = \Big{(}1-d(1-p)\Big{)}s+l'_{i}-co_{i}. \label{positive}
    \end{split}
\end{equation}  Note that \eqref{positive} holds for $i\in \{ 1,2,3,4 \}$. Since $1-(1-p)d>0$, \eqref{threshold1} (for $i\in \{1,3,4\}$) and \eqref{threshold2} (for $i=2$) are the minimum point of $r_{i}(c,s)$. Thus, we complete the proof.